\documentclass[11pt]{amsart}
\usepackage{amsmath}
\usepackage{amsfonts}
\usepackage{mathrsfs}
\usepackage {amssymb, euscript}
\usepackage {amsmath}
\usepackage {amsthm}
\usepackage{graphicx,color}
\usepackage {amscd}
\usepackage{geometry}
\usepackage{hyperref}
\usepackage{color}
\usepackage{epsfig}
\usepackage{enumerate} 
\usepackage{verbatim}
\usepackage{slashed}
\usepackage{tikz}
\usepackage{cleveref}
\usepackage{bbm}
\usepackage{graphics}
\usepackage[pagewise]{lineno}
\usepackage{yhmath}

\usepackage{graphics}

\usepackage{sansmath}
\usetikzlibrary{shadings,intersections} 

\usepackage{enumitem}

\usepackage{slashed}
\title[Dynamics of 
Apparent Horizon and a Null Comparison Principle]{Dynamics of Apparent Horizon  \\ and \\ a null comparison principle}
\date{\today}

\author{Xinliang An} 
\address{Department of Mathematics, National University of Singapore, Singapore 119076}
\email{matax@nus.edu.sg}

\author{Taoran He}
\address{Department of Mathematics, National University of Singapore, Singapore 119076}
\email{taoran\textunderscore he@u.nus.edu}

\theoremstyle{plain}
\newtheorem{lemma}{Lemma}[section]
\newtheorem{definition}[lemma]{Definition}
\newtheorem{proposition}[lemma]{Proposition}
\newtheorem{theorem}[lemma]{Theorem}

\newtheorem{remark}{Remark}

\numberwithin{equation}{section}

\begin{document}

\newcommand{\ub}{\underline{u}}
\newcommand{\Cb}{\underline{C}}
\newcommand{\Lb}{\underline{L}}
\newcommand{\Lh}{\hat{L}}
\newcommand{\Lbh}{\hat{\Lb}}
\newcommand{\phib}{\underline{\phi}}
\newcommand{\Phib}{\underline{\Phi}}
\newcommand{\Db}{\underline{D}}
\newcommand{\Dh}{\hat{D}}
\newcommand{\Dbh}{\hat{\Db}}
\newcommand{\omb}{\underline{\omega}}
\newcommand{\omh}{\hat{\omega}}
\newcommand{\ombh}{\hat{\omb}}
\newcommand{\Pb}{\underline{P}}
\newcommand{\chib}{\underline{\chi}}
\newcommand{\chih}{\hat{\chi}}
\newcommand{\chibh}{\hat{\chib}}
\newcommand{\kb}{\overline{\kappa}}
\newcommand{\kbb}{\overline{\underline{\kappa}}}
\newcommand{\vsgmb}{\underline{\varsigma}}
\newcommand{\ud}{\underline}
\newcommand{\alb}{\underline{\alpha}}
\newcommand{\zeb}{\underline{\zeta}}
\newcommand{\beb}{\underline{\beta}}
\newcommand{\etb}{\underline{\eta}}
\newcommand{\Mb}{\underline{M}}
\newcommand{\oth}{\hat{\otimes}}

\newcommand{\tR}{\tilde{R}}
\newcommand{\mL}{\mathcal{L}}
\newcommand{\tK}{\widetilde{K}}
\newcommand{\rh}{\hat{r}}
\newcommand{\tauh}{\hat{\tau}}
\newcommand{\prub}{\partial_{\ub}}
\newcommand{\bfD}{\mathbf{D}}
\newcommand{\bfg}{\mathbf{g}}
\newcommand{\bfR}{\mathbf{R}}
\newcommand{\bfK}{\mathbf{K}}

\newcommand{\red}{\textcolor{red}}
\newcommand{\blue}{\textcolor{blue}}
\newcommand{\purple}{\textcolor{purple}}

\def\a {\alpha}
\def\b {\beta}
\def\ab {\alphab}
\def\bb {\betab}
\def\nab {\nabla}
\def\ssnab {\slashed{\nabla}}

\def\ub {\underline{u}}
\def\th {\theta}
\def\Lb {\underline{L}}
\def\Hb {\underline{H}}
\def\chib {\underline{\chi}}
\def\chih {\hat{\chi}}
\def\chibh {\hat{\underline{\chi}}}
\def\omegab {\underline{\omega}}
\def\etab {\underline{\eta}}
\def\betab {\underline{\beta}}
\def\alphab {\underline{\alpha}}
\def\Psib {\underline{\Psi}}
\def\hot{\widehat{\otimes}}
\def\Phib {\underline{\Phi}}
\def\thb {\underline{\theta}}
\def\t {\tilde}
\def\st {\tilde{s}}

\def\d {\delta}
\def\f {\frac}
\def\i {\infty}
\def\l {\bigg(}
\def\r {\bigg)}
\def\S {S_{u,\underline{u}}}
\def\o{\omega}
\def\O{\Omega}
\def\be{\begin{equation}\begin{split}}
\def\en{\end{split}\end{equation}}

\def\od{\omega^{\dagger}}
\def\ombd{\underline{\omega}^{\dagger}}
\def\K{K-\frac{1}{|u|^2}}
\def\ut{\frac{1}{|u|^2}}
\def\Kb{K-\frac{1}{(u+\underline{u})^2}}
\def\M{\mathcal}
\def\p{\psi}

\def\D{\Delta}
\def\T{\Theta}
\def\s{S_{u',\underline{u}'}}
\def\Hu{H_u^{(0,\underline{u})}}
\def\Hbu{\underline{H}_{\underline{u}}^{(u_{\infty},u)}}
\def\ee{(\eta,\underline{\eta})}

\def\at{a^{\f12}}
\def\sigmac{\check{\sigma}}
\def\p{\psi}
\def\q{\underline{\psi}}
\def\ls{\leq}
\def\de{\delta}
\def\ls{\lesssim}
\def\oo{\Omega\mbox{tr}\chib-\frac{2}{u}}
\def\om{\omega}
\def\Om{\Omega}

\renewcommand{\div}{\mbox{div }}
\newcommand{\curl}{\mbox{curl }}
\newcommand{\trchb}{\mbox{tr} \chib}
\def\trch{\mbox{tr}\chi}

\newcommand{\Ls}{{\mathcal L} \mkern-10mu /\,}
\newcommand{\eps}{{\epsilon} \mkern-8mu /\,}

\newcommand{\tr}{\mbox{tr}}

\newcommand{\xib}{\underline{\xi}}
\newcommand{\psib}{\underline{\psi}}
\newcommand{\rhob}{\underline{\rho}}
\newcommand{\thetab}{\underline{\theta}}
\newcommand{\gammab}{\underline{\gamma}}
\newcommand{\nub}{\underline{\nu}}
\newcommand{\lb}{\underline{l}}
\newcommand{\mub}{\underline{\mu}}
\newcommand{\Xib}{\underline{\Xi}}
\newcommand{\Thetab}{\underline{\Theta}}
\newcommand{\Lambdab}{\underline{\Lambda}}
\newcommand{\vphb}{\underline{\varphi}}

\newcommand{\ih}{\hat{i}}
\newcommand{\ui}{u_{\infty}}
\newcommand{\shb}{L^2_{sc}(\underline{H}_{\ub}^{(u_{\infty},u)})}
\newcommand{\sh}{L^2_{sc}(H_{u}^{(0,\ub)})}
\newcommand{\Rb}{\underline{\mathcal{R}}}
\newcommand{\tc}{\widetilde{\tr\chib}}

\newcommand{\tcL}{\widetilde{\mathscr{L}}}

\newcommand{\sRic}{Ric\mkern-19mu /\,\,\,\,}
\newcommand{\sL}{{\cal L}\mkern-10mu /}
\newcommand{\sLh}{\hat{\sL}}
\newcommand{\sg}{g\mkern-9mu /}
\newcommand{\seps}{\epsilon\mkern-8mu /}
\newcommand{\sd}{d\mkern-10mu /}
\newcommand{\sR}{R\mkern-10mu /}
\newcommand{\snab}{\nabla\mkern-13mu /}
\newcommand{\sdiv}{\mbox{div}\mkern-19mu /\,\,\,\,}
\newcommand{\scurl}{\mbox{curl}\mkern-19mu /\,\,\,\,}
\newcommand{\slap}{\mbox{$\triangle  \mkern-13mu / \,$}}
\newcommand{\sGamma}{\Gamma\mkern-10mu /}
\newcommand{\somega}{\omega\mkern-10mu /}
\newcommand{\somb}{\omb\mkern-10mu /}
\newcommand{\spi}{\pi\mkern-10mu /}
\newcommand{\sJ}{J\mkern-10mu /}
\renewcommand{\sp}{p\mkern-9mu /}
\newcommand{\su}{u\mkern-8mu /}

\maketitle

\begin{abstract}
This paper investigates the global dynamics of the apparent horizon. We present an approach to establish its existence and its long-term behaviors. Our apparent horizon is constructed by solving the marginally outer trapped surface (MOTS) along each incoming null hypersurface. Based on the nonlinear hyperbolic estimates established in \cite{K-S} by Klainerman-Szeftel under polarized axial symmetry, we prove that the corresponding apparent horizon is smooth, asymptotically null and converging to the event horizon eventually. To further address the local achronality of the apparent horizon, a new concept, called the \textit{null comparison principle,} is introduced in this paper. For three typical scenarios of gravitational collapse, our null comparison principle is tested and verified, which guarantees that the apparent horizon must be piecewise spacelike or piecewise null. In addition, we also validate and provide new proofs for several physical laws along the apparent horizon.
\end{abstract}

\section{Introduction}
The explorations of trapped surfaces and the trapped region trace back to 1960s, when Penrose established his renowned incompleteness theorem. The boundary of the trapped region is a quasi-local object, known as the apparent horizon. Unlike the event horizon, i.e., the boundary of the causal past of the future null infinity, which requires global information about spacetime, the apparent horizon describes the boundary of the trapped region and demands only local information. This feature holds importance in practical applications, especially for numerical simulations of general relativity.

Over the decades, numerous efforts have been dedicated to unraveling the properties of the apparent horizon. However, the detailed understanding regarding the long-time behavior and the final state of apparent horizon is still lacking. Under spherical symmetry, Williams studied the asymptotic behavior of the apparent horizon in \cite{W}. In this current paper, we initiate the study of the global dynamics of apparent horizon beyond spherical symmetry. Relying on the hyperbolic estimates established in \cite{K-S} by Klainerman-Szeftel, we prove the existence and uniqueness of MOTS along each incoming null hypersurface. These MOTSs collectively form a three-dimensional hypersurface and we further prove it to be smooth, asymptotically null and eventually approaching the event horizon.

In general, a hypersurface being asymptotically null could be timelike at certain points. In this paper, we introduce a new concept, called the \textit{null comparison principle} and we use it to guarantee that the apparent horizon is piecewise spacelike or piecewise null. Furthermore, we present three typical scenarios of gravitational collapse, in which the null comparison principle holds. These scenarios are in the perturbative Schwarzschild spacetimes \cite{K-S}, isotropic and anisotropic gravitational-collapse spacetimes \cite{An-Han} and spacetimes containing naked singularities \cite{An naked singularity}.

Additionally, by virtue of achronal property proved for the apparent horizon, we also investigate the black hole thermodynamics along the apparent horizon. We verify and  offer new  proofs for several physical laws first formulated by Ashtekar and Krishnan in \cite{A-K2, A-K1}.

\subsection{Main Results} 
One of the main objectives of this paper is to study the global existence and the long-term behavior of the apparent horizon with the help of the hyperbolic estimates derived in the striking work of Klainerman-Szeftel \cite{K-S}. There they introduced the novel and powerful generally covariant modulation (GCM) approach for the first time. In the nonlinear setting, they established the asymptotic stability of the Schwarzschild black holes for Einstein vacuum equations under axially symmetric polarized perturbations. They also established hyperbolic estimates beyond the event horizon and provided comprehensive decay estimates for all geometric quantities. We use their hyperbolic estimates in the interior to trace the global dynamics of the apparent horizon. We also want to mention another important work \cite{D-H-R-T} by Dafermos-Holzegel-Rodnianski-Taylor. There they proved co-dimensional 3 asymptotic stability of the Schwarzschild black holes. We remark that our geometric frame in this paper is quite general and it does not depend on symmetry assumption. Since the interior hyperbolic estimates are provided in \cite{K-S} by Klainerman-Szeftel, in below we borrow some of their notations.

\vspace{2mm}
Denoting $\Hb_{\ub}$ to be the incoming null hypersurface with $\ub$ standing for the incoming optical function, relying on the existence theory and the hyperbolic estimates in \cite{K-S}, our first main result of this paper is
\begin{theorem}\label{Main thm 1: existence of AH in K-S}
    With initial data prescribed in \cite{K-S}, the solved Einstein vacuum spacetime contains an  apparent horizon that is asymptotically null and approaches the timelike infinity. Specifically, along each incoming null hypersurface $\Hb_{\ub}$, there exists a unique MOTS $M_{\ub}$, and the collection of $M_{\ub}$ forms a smooth three dimensional apparent horizon $\mathcal{AH}$. Moreover, the induced metric on $\mathcal{AH}$ is asymptotically degenerate and $\mathcal{AH}$ converges to the event horizon eventually.
\end{theorem}
We remark that our geometric frame in this paper is not confined to the polarized  axial symmetry. It can apply to more general spacetimes as long as the hyperbolic decay estimates in the interior region as in \cite{K-S} are available.
\vspace{3mm}

When an apparent horizon is close to a null hypersurface, in general, it is still hard to draw a conclusion whether the apparent horizon is locally achronal (i.e., spacelike or null) or not. The difficulty stems from deriving the non-negative lower bound for solutions to the quasilinear elliptic equations (see \Cref{Rmk non-negative low bound} in \Cref{Subsubsec: application NCP Sch} for more details). In below, we introduce a new concept, called  the \textit{null comparison principle}. With the help of it, for three typical scenarios of gravitational collapse, we prove that the apparent horizon  must be piecewise spacelike or piecewise null. 

We now consider a $3+1$ dimensional spacetime $\mathcal{M}$ foliated by incoming null hypersurfaces $\Hb_{\ub}$. Let $(r, \theta_1, \theta_2)$ be the coordinate system within $\Hb_{\ub}$ so that $\partial_r$ is along the past-directed incoming null geodesic.  Define 2-sphere $S_{\ub, r}$ to be the level set of $r$ on $\Hb_{\ub}$. With the apparent horizon $\mathcal{AH}$ composed of MOTS $M_{\ub}$ along each $\Hb_{\ub}$,  the explicit definition of the null comparison principle is as below:
\begin{definition}
    Let $M_{\ub}=\{r= R(\theta_1, \theta_2)\}$ be a MOTS along $\Hb_{\ub}$. The null comparison principle with respect to the MOTS $M_{\ub}$ along $\Hb_{\ub}$ states that for any  (smooth) spacelike $2$-surface $\widetilde{\Sigma}=\{r=\tilde{R}(\theta_1, \theta_2) \}$  within $\Hb_{\ub}$ near $M_{\ub}$, if the outgoing null expansion of $\widetilde{\Sigma}$ is non-positive, then it must hold pointwise that $\tilde{R}(\theta_1, \theta_2) \leq R(\theta_1, \theta_2)$ and if the outgoing null expansion of $\widetilde{\Sigma}$ is non-negative, then it yields $\tilde{R}(\theta_1, \theta_2) \ge R(\theta_1, \theta_2)$. 
\end{definition}
\begin{figure}
    \centering

\tikzset{every picture/.style={line width=0.75pt}} 

\begin{tikzpicture}[x=0.75pt,y=0.75pt,yscale=-1,xscale=1]

\draw   (230.56,89.18) .. controls (230.56,78.68) and (264.01,70.16) .. (305.28,70.16) .. controls (346.55,70.16) and (380,78.68) .. (380,89.18) .. controls (380,99.68) and (346.55,108.2) .. (305.28,108.2) .. controls (264.01,108.2) and (230.56,99.68) .. (230.56,89.18) -- cycle ;

\draw    (246.4,44.56) -- (200.44,169.8) ;

\draw    (368.96,45.8) -- (405.6,176.56) ;

\draw    (210.91,141.37) .. controls (232.65,127.6) and (266.57,134.97) .. (307.71,149.26) .. controls (348.86,163.54) and (376.3,155.03) .. (395.43,140.69) ;

\draw  [dash pattern={on 0.84pt off 2.51pt}]  (210.91,141.37) .. controls (235.77,111.48) and (288,112.97) .. (330.29,129.83) .. controls (372.57,146.69) and (392.53,135.76) .. (395.43,140.69) ;

\draw (298.45,91.86) node [anchor=north west][inner sep=0.75pt]  [font=\tiny]  {$M_{\underline{u}}$};

\draw (300.34,158.56) node [anchor=north west][inner sep=0.75pt]  [font=\tiny]  {$\tilde{\Sigma }$};

\draw (394.27,106.49) node [anchor=north west][inner sep=0.75pt]  [font=\tiny]  {$\underline{H}_{\underline{u}}$};

\end{tikzpicture}

    \caption{Null comparison principle for MOTS $M_{\ub}$ along $\Hb_{\ub}$.}
    \label{NCP figure}
\end{figure}
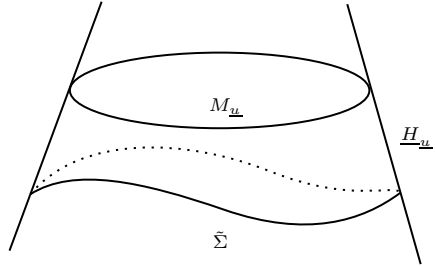

In  \Cref{Subsec: application NCP}, we will verify that the above  null comparison principle holds in the setting of Klainerman-Szeftel \cite{K-S}, in the anisotropic gravitational-collapse spacetimes \cite{An-Han} by the first author and Han, and in the spacetimes containing naked singularities studied by the first author in \cite{An naked singularity}. Our next result demonstrates the use of this principle.
\begin{theorem}\label{Main thm 2: null comparison principle}
    With $\mathcal{AH}=\cup_{\ub} M_{\ub}$, suppose that the null comparison principle holds for MOTS $M_{\ub}$ along $\Hb_{\ub}$. Then the apparent horizon $\mathcal{AH}$ must be spacelike everywhere or (outgoing) null  everywhere when restricted to $M_{\ub}$, i.e., the (non-zero) tangent vector of $\mathcal{AH}$ that is normal to $M_{\ub}$ is either spacelike everywhere or (outgoing) null  everywhere on $M_{\ub}$. 
\end{theorem}
The proof of  \Cref{Main thm 2: null comparison principle} is purely geometric. For both \textit{null comparison principle} and \Cref{Main thm 2: null comparison principle}, we do not impose any symmetry conditions. A quick corollary of \Cref{Main thm 2: null comparison principle} is that the solved apparent horizon $\mathcal{AH}$ in \Cref{Main thm 1: existence of AH in K-S} is piecewise spacelike or piecewise null. It is worthwhile to mention that compared with the proof of spacelikeness for the apparent horizon in the short-pulse regime as established in \cite{An: AH, An-Han}, the advantage of verifying and employing the null comparison principle is avoiding the non-negative lower-bound estimates for the $\partial_{\ub}$ derivative of solutions to the quasilinear elliptic equation of MOTS. To guarantee the (local) achronality of the apparent horizon in the short-pulse regime, our approach here also relaxes the requirement of the prescribed initial data.
\vspace{2mm}

Throughout this paper, our MOTS is constructed and studied along an incoming null cone. We adopt the approach designed by the first author in \cite{An: AH}. It is also worth mentioning that MOTS related to the null Penrose's inequality is investigated by Alexakis in \cite{A: Penrose} and by Le in \cite{Le1, Le2} and references therein. MOTS along a spacelike leaf is extensively explored by numerous mathematicians and physicists. We refer the interested readers to  works \cite{ A-E, A-M, A-Met, A-G, A-K1, A-K2, E, M, S-Y, Y} and references therein. In particular, we want to mention that Andersson-Metzger \cite{A-Met} and Eichmair \cite{E} proved the existence of MOTS within a spacelike hypersurface, employing Jang's equation and associated barrier arguments. Meanwhile, in \cite{Andersson: stability}, Andersson-Mars-Simon introduced the concept of stability operator for the MOTS and proved the achronality of apparent horizon around a strictly stable outermost MOTS. In our paper, we conduct arguments based on null foliations and prove achronality of apparent horizon via introducing the null comparison principle. 

\subsection{Key Points and New Ingredients of the Paper}

\subsubsection{Choice of Coordinate System}
To solve for MOTS, we perform a coordinate transformation in the interior region ${}^{(int)}\mathcal{M}$ constructed by Klainerman-Szeftel \cite{K-S}. There they adopted the incoming geodesic foliation with coordinates $(\ub, r, \theta, \varphi)$. Let $(e_1, e_2, e_3, e_4)$ be the associate null frame. In their coordinates, the coordinates derivative  $\partial_r$ can be expressed as 
\begin{equation*}
    \partial_r=[e_3(r)]^{-1}(e_3-\sqrt{\gamma} \ud{b} e_\theta ) \quad \text{with} \quad \ud{b}=e_3(\theta) \quad \text{and} \quad \gamma=\f{1}{(e_\theta(\theta))^2}. 
\end{equation*}

Due to the presence of $e_\theta$,  their coordinate derivative $\partial_r$ is not necessarily null. To construct the MOTS along a null hypersurface, we prefer a coordinate derivative to be null. This condition is achieved by choosing coordinates $(\widetilde{\ub}, \widetilde{r}, \widetilde{\theta}, \widetilde{\varphi} )$ with
\begin{equation*}
    \widetilde{\ub}=\ub, \quad \widetilde{r}=r, \quad \widetilde{\theta}=\theta+f_1(\ub, r, \theta, \varphi), \quad \widetilde{\varphi}=\varphi.
\end{equation*}
Here we require 
\begin{equation}\label{intro: e3(f1+theta) eqn}
    e_3(f_1+\theta)=0  \quad \text{with} \quad 1+\frac{\partial f_1}{\partial \theta}\neq 0.
\end{equation}
Following from chain rules, we get
\begin{equation*}
    \partial_{\widetilde{r}}=[e_3(r)]^{-1}\Big(e_3+(1+\frac{\partial f_1}{\partial \theta})^{-1}e_3(f_1+\theta)\partial_\theta \Big).
\end{equation*}
By the requirement \eqref{intro: e3(f1+theta) eqn}, in our new coordinate system $(\widetilde{\ub}, \widetilde{r}, \widetilde{\theta}, \widetilde{\varphi} )$ we have 
$\partial_{\widetilde{r}}=[e_3(r)]^{-1}e_3$ being null. It is important to note that in our analysis, we utilize the same null frame $(e_3, e_4, e_1, e_2)$ as presented in Klainerman-Szeftel \cite{K-S}. Hence, the hyperbolic estimates established in \cite{K-S} remain unchanged and applicable in our framework.

\subsubsection{Existence and Asymptotics of Apparent Horizon}
Along each incoming null hypersurface $\Hb_{\ub}$, to find the MOTS on it, we employ the below deformation equation: Setting the MOTS $M_{\ub}$ on $\Hb_{\ub}$ to be with coordinates $\{(\ub, r, \theta, \varphi):  r=R(\theta, \varphi)\}$, then the outgoing null expansion $\trch'$ on $M_{\ub}$ can be expressed as 
\begin{equation*}
	\trch'=\trch+2f\Delta R+2(\nab f+f(\eta+\zeta))\cdot \nab R-4f^2\chibh_{bc}\nab^b R\nab^c R+(2fe_3(f)-f^2\trchb-4\omegab f^2)|\nab R|^2.
\end{equation*}
Here $f=[e_3(r)]^{-1}$, $\Delta$ and $\nab$ represent the Laplace-Beltrami operator and the induced covariant derivative, respectively, on the 2-sphere $S_{\ub, r}$. With $(e_a)_{a=1, 2}$ a tangent frame on $S_{\ub, r}$, we have the following definition for Ricci coefficients
\begin{equation*}
\begin{split}
\chib_{ab}=\bfg(\bfD_a e_3,e_b), \quad \eta_a=-\frac 12 \bfg(\bfD_3 e_a,e_4),\quad \zeta_a=\frac 1 2 \bfg(\bfD_a e_4,e_3),\quad \omegab=-\frac 14 \bfg(\bfD_3 e_4,e_3).
\end{split}
\end{equation*}
Here $\bfD$ is the spacetime covariant derivative and $\trchb$, $\chibh$ denote the trace and traceless parts of $\chibh$, respectively.
	
Using \Cref{change laplacian} (see also calculations in \cite{An: AH}), for $M_{\ub}=\{(\ub, r, \theta, \varphi):  r=R(\theta, \varphi)\}$ being the MOTS to be solved,  it holds
\begin{equation*}
	\Delta R =\Delta'_{M_{\ub}} R+2f \chibh_{bc}\nab^b R\nab^c R.
\end{equation*}
We can then rewrite the equation of MOTS $\trch'=0$ as
\begin{equation}\label{intro: MOTS eqn}
    \begin{split}
    0=L(R, \ub):=&\Delta R-2f\chibh_{bc}\nab^b R\nab^c R+\big(f^{-1}\nab f+(\eta+\zeta)\big)\cdot \nab R\\&+(e_3(f)-\f{1}{2}f\trchb-2\omegab f)|\nab R|^2+\f{1}{2}f^{-1}\trch.
\end{split}
\end{equation}
\begin{remark}
In \cite{An: AH} the first author studied the below equation
\begin{equation}\label{intro: MOTS eqn in An}
     \begin{split}
   \Delta'_{M_{\ub}} R+2\eta_a\nab^a R+(\f{1}{2}\Omega\trchb+4\Omega\omegab )|\nab R|^2-\f{1}{2}\Omega^{-1}\trch=0.
\end{split}
\end{equation}
 near the center of gravitational collapse with double null foliations based on hyperbolic estimates in  \cite{A-L} by the first author and Luk. In this paper, we use the aforementioned incoming geodesic foliation and we consider the regime when the spacetime settles down to its equilibrium building upon hyperbolic estimates in \cite{K-S} by Klainerman-Szeftel. Also note that our MOTS equation \eqref{intro: MOTS eqn} involves function $f=[e_3(r)]^{-1}$ and it differs from the lapse function $\O$ appearing in \eqref{intro: MOTS eqn in An}.
\end{remark}

We use the below estimates in \cite{K-S}: For $0\le k\le 1$, in ${}^{(int)}\mathcal{M}$ there holds
\begin{equation}\label{intro: decay est}
\begin{split}
   & \|(e_3, e_4, \nab)^k(f+1, \eta, \zeta, \trchb+\f{2}{r}, \trch-\f{2}{r}+\f{4m_{\infty}}{r^2})\|_{L^\infty(S_{\ub, r})}\lesssim \f{\varepsilon_0}{\ub^{1+\delta_{dec}}} \quad \text{and} \quad \omegab\equiv 0.
    \end{split}
\end{equation}
Here $m_{\infty}$ denotes the mass of the final Schwarzschild spacetime, $\varepsilon_0$ represents the size of initial perturbation and $\delta_{dec}$ is a small positive constant. We emphasize that the smallness of $\varepsilon_0$ is important when to establish the existence of MOTS along each $\Hb_{\ub}$ and the decaying rate $\ub^{-1-\delta_{dec}}$ in \eqref{intro: decay est} plays a critical role in analyzing the asymptotic behaviors of the apparent horizon.

Employing \eqref{intro: decay est} we can simplify \eqref{intro: MOTS eqn} to the form
    \begin{equation}\label{intro: MOTS eqn general form}
\begin{split}
   &\Delta'_M R-\f{1}{R}|\nab R|^2-\f{1}{R}+\f{2m_{\infty}}{R^2}+O(\f{\varepsilon_0}{\ub^{1+\delta_{dec}}})(\nab R+|\nab R|^2+1)=0.
\end{split}
\end{equation}
By applying maximum principle on $M_{\ub}$, in \Cref{Subsubsec: C0 est} we deduce
\begin{equation}\label{intro: C0 est}
    |R(\theta, \varphi)-2m_{\infty}|\lesssim \f{\varepsilon_0}{\ub^{1+\delta_{dec}}} \qquad \text{for all} \quad (\theta, \varphi)\in \mathbb{S}^2.
\end{equation}
We then use the modified Bochner’s formula\footnote{This type of modified Bochner's formula was first derived and applied by the first author in \cite{An: AH}.} for $\Delta'_M\Big(h(R)|\nab R|^2\Big)$ with
\begin{equation*}
h(R)=1+\f{1}{2m_\infty ^2}(R-2m_\infty)^2.
\end{equation*} 
With details provided in \Cref{Subsubsec: C1 est}, we arrive at 
\begin{equation*}
\begin{split}
&\Delta'_M\Big(h(R)|\nab R|^2\Big)-\f{2h'(R)\nab^{'a} R}{h(R)}\cdot \nab'_a\Big(h(R)|\nab R|^2\Big)-\f{2}{R^2}\nab'_a(h(R)|\nab R|^2) \\
\geq& \f{1}{8m_\infty^2}|\nab R|^4+\f{1}{2R^2}|\nab R|^2-|O(\f{\varepsilon_0}{\ub^{1+\delta_{dec}}})|^2.
\end{split}
\end{equation*}
Since $h(R)$ is close to $1$, another application of maximum principle implies
\begin{equation*}
    |\nab R(\theta, \varphi)|\lesssim \f{\varepsilon_0}{\ub^{1+\delta_{dec}}}  \qquad \text{for all} \quad (\theta, \varphi)\in \mathbb{S}^2.
\end{equation*}
Following the steps in \Cref{Subsubsec: C 1 q est}, we can further improve our estimates to
   \begin{equation*}
        \| R-2m_{\infty}\|_{C^{2, q}(M_{\ub})} \lesssim  \f{\varepsilon_0}{\ub^{1+\delta_{dec}}} 
    \end{equation*}
with some constant $q\in (0, 1)$ independent of $\varepsilon_0$ and $R(\theta, \varphi)$.
\vspace{2mm}

With these derived apriori estimates, in \Cref{Continuity1} we conduct the method of continuity to solve for the MOTS. We construct
\begin{equation*}
\begin{split}
      F(R, \lambda):&=\Delta'_M R+(e_3(f)-\f{1}{2}f\trchb-2\omegab f)|\nab R|^2\\&+\lambda\Big(\big(f^{-1}\nab f+(\eta+\zeta)\big)\cdot \nab R+\f{1}{2}f^{-1}\trch\Big)+(1-\lambda)\l-\f{1}{R}+\f{2m_\infty}{R^2}\r.
\end{split}
\end{equation*}
Note that straightforward that $R\equiv 2m_{\infty}$ is the solution to equation $F(R, 0)=0$. Assume that $F(\tR, \tilde{\lambda})=0$ admits a solution $\tR$ for parameter $\tilde{\lambda}\in [0, 1]$. A calculation in \Cref{Continuity1} gives
\begin{equation*}
    \partial_R F (\tR, \lambda)[W]:=\f{d}{d\epsilon}\Big|_{\epsilon=0}F(\tR+\epsilon W, \lambda)=\Delta'_{S_{\ub, \tR}} W+d_{a}\nab^a W+\tR^{-3}\Big[-2m_{\infty} +d\Big] W \quad \text{for all} \ \lambda\in [0, 1].
\end{equation*}
Here $d=d(\ub, \tR, \theta, \varphi)$  and $d_{a}=d_a(\ub, \tR, \theta, \varphi)$ are functions satisfying
\begin{equation*}
    |d|+|d_a|\lesssim \f{\varepsilon_0}{\ub^{1+\delta_{dec}}}.
\end{equation*}
Notice that $\tR^{-3}\Big[-2m_{\infty} +d\Big]<0$ provided $\varepsilon_0>0$  sufficiently small. The invertibility of elliptic operator $\partial_R F (\tR, \lambda)[W]: C^{2, q}(\mathbb{S}^2)\to C^{0, q}(\mathbb{S}^2)$  follows consequently. Therefore, by the details provided in \Cref{Continuity1}, we then conclude that the equation for MOTS, i.e., $0=F(R, 1)=L(R, \ub)$ possesses a solution $R$ for each $\ub \ge 1$. Further proof in \Cref{Subsec: uniqueness of MOTS} also provides the uniqueness of $R$ along $\Hb_{\ub}$.
\vspace{3mm}

We then move to consider the property of the apparent horizon $\mathcal{AH}$ and  $\mathcal{AH}:=\cup_{\ub \ge 1} M_{\ub}$. To investigate the regularity of $\mathcal{AH}$ we appeal to the implicit function theorem with regard to the elliptic operator 
\begin{equation*}
    L( R, \ub): C^{\infty}(\mathbb{S}^2)\times [1, \infty)  \to C^{\infty}(\mathbb{S}^2).
\end{equation*}
 Since the invertibility of linearized operator $ \partial_R L (R, \ub)=\partial_R F (\ub, R, 1)$  for any $\ub \ge 1$ has been proved in \Cref{Continuity1}, in \Cref{Subsec: regularity of AH} we deduce that $R$ is smooth in both $\ub$ and $\theta, \varphi$.

To study the asymptotics of $\mathcal{AH}$, we employ the implicit  function theorem again. From $L(R, \ub)=0$, we first get
\begin{equation*}
    \partial_R L(R, \ub)[\partial_{\ub} R]+\partial_{\ub} L(R, \ub)=0.
\end{equation*}
With a priori estimates for $R$, the hyperbolic estimates and the invertibility of $ \partial_R L$, in \Cref{Subsec: final state of AH} we show that
\begin{equation}\label{intro: paritial ub R est}
   \max\limits_{\mathbb{S}^2} |\partial_{\ub} R|\lesssim \max\limits_{\mathbb{S}^2}| \partial_{\ub} L(R, \ub)|\lesssim \f{\varepsilon_0}{\ub^{1+\delta_{dec}}}.
\end{equation}
Recall that the components of the induced metric on $\mathcal{AH}$ take the forms
\begin{equation*}
	g'_{\theta_i \theta_j}=g_{\theta_i \theta_j}+\partial_{\theta_i} R\cdot\partial_{\theta_j} R\cdot g(\partial_r, \partial_r)=g_{\theta_i \theta_j}, 
\end{equation*}
\begin{equation*}
	g'_{\ub\, \ub}=g_{\ub\, \ub}+2\partial_{\ub} R\cdot g(\partial_r, \partial_{\ub})=\vsgmb^2\l\f{\kb+A}{\kbb}+\f{1}{4}\gamma b^2\r-2\partial_{\ub} R\cdot\f{2\vsgmb}{r\kbb},
\end{equation*}
\begin{equation*}
	g'_{\theta \ub}=g_{\theta \ub}+\partial_{\theta} R \cdot g(\partial_r, \partial_{\ub})=-\f{1}{2}\vsgmb \gamma b-\partial_{\theta} R\cdot\f{2\vsgmb}{r\kbb},
\end{equation*}
\begin{equation*}
	g'_{\varphi \ub}=g_{\varphi \ub}+\partial_{\varphi} R \cdot g(\f{\partial}{\partial r}, \partial_{\ub})=-\partial_{\varphi} R\cdot\f{2\vsgmb}{r\kbb}.
\end{equation*}
Here  $\vsgmb, \kb, A, \kbb, \gamma, b$ represent metric components as employed in \cite{K-S} and their precise definitions are listed in \eqref{KS metric}, \eqref{KS metric component}. Hence, plugging $C^1$ estimate, \eqref{intro: paritial ub R est}, we arrive at
\begin{equation*}
	|(g'_{\ub\, \ub}, g'_{\theta \ub}, g'_{\varphi \ub})|\lesssim \f{\varepsilon_0}{\ub^{1+\delta_{dec}}}.
\end{equation*}
Therefore, as $\ub$ approaches infinity, we conclude that the directional derivative $\partial'_{\ub}$ manifests as an asymptotically degenerate null direction along  $\mathcal{AH}$.  In other words, our constructed apparent horizon $\mathcal{AH}$ is asymptotically null.
\vspace{2mm}

To show that $\mathcal{AH}$ converges to the event horizon as $\ub \to +\infty$, we appeal to the estimates for the event horizon in \cite{K-S}. There Klainerman-Szeftel proved that the event horizon $\mathcal{H}_{+}$ of $\mathcal{M}$ locates in the following region of ${}^{(int)}\mathcal{M}$:
\begin{equation*}
    2m_{\infty}(1-\f{\sqrt{\varepsilon_0}}{\ub^{1+\delta_{dec}}}) \le r \le 2m_{\infty}(1+\f{\sqrt{\varepsilon_0}}{\ub^{1+\delta_{dec}}}) \quad \text{for any} \quad \ub \ge 1.
\end{equation*}
Denoting $\mathcal{H}_{+}:=\{r=R_{\mathcal{H}_{+}}(\ub, \theta, \varphi): \ub\ge 1 \}$ to be the event horizon, it then follows from our $C^0$ estimate \eqref{intro: C0 est} that
\begin{equation*}
    |R(\ub, \theta, \varphi)-R_{\mathcal{H}_{+}}(\ub, \theta, \varphi)|\lesssim \f{\sqrt{\varepsilon_0}}{\ub^{1+\delta_{dec}}} \quad \text{for all} \quad (\theta, \varphi)\in \mathbb{S}^2.
\end{equation*}
Hence, the apparent horizon $\mathcal{AH}$ constructed in our setting converges to the event horizon $\mathcal{H}_{+}$ as $\ub\rightarrow +\infty$.

\subsubsection{Null Comparison Principle and its Applications}    
In \cite{An: AH} and \cite{An-Han} the spacelikeness of apparent horizon is directly examined in the short-pulse regime. And $g'_{\ub\,\ub}>0$ is sufficiently large related to short-pulse data. However, to prove the local achronality of the apparent horizon, the method adopted in \cite{An: AH,An-Han} is ineffective in the perturbative Schwarzschild spacetime regime as in this paper due to the fact that $g'_{\ub\,\ub}$ is small depending on the size of initial perturbation and it can approach $0$ as $\ub\to +\infty$. In this paper, we give a new robust approach to establish the achronality of apparent horizon $\mathcal{AH}$ through a geometric argument.  We point out that this approach does not rely on any symmetry assumptions.

Suppose that we have the apparent horizon $\mathcal{AH}=\cup_{\ub} M_{\ub}$ and  let $M_{\ub}=\{r= R(\theta_1, \theta_2)\}$ be the MOTS lying along the incoming null hypersurface $\Hb_{\ub}$ with coordinates $(r, \theta_1, \theta_2)$. In this paper we introduce the concept of the null comparison principle. This principle holds if  for any  (smooth) spacelike $2$-surface $\widetilde{\Sigma}=\{r=\tilde{R}(\theta_1, \theta_2) \}$  on $\Hb_{\ub}$ near $M_{\ub}$, when its outgoing null expansion $\trch|_{\widetilde{\Sigma}}\le 0$,  we have $\tilde{R}(\theta_1, \theta_2) \leq R(\theta_1, \theta_2)$ for all $(\theta_1, \theta_2)\in \mathbb{S}^2$ and when $\trch|_{\widetilde{\Sigma}}\ge 0$, then it holds $\tilde{R}(\theta_1, \theta_2) \ge R(\theta_1, \theta_2)$ pointwise.  

With this definition, in \Cref{Section: NCP} we prove 
\begin{theorem}\label{intro: NCP section NCP thm}
 For $\mathcal{AH}=\cup_{\ub} M_{\ub}$, assuming that the null comparison principle holds for MOTS $M_{\ub}$ on $\Hb_{\ub}$, then the apparent horizon $\mathcal{AH}$ must be either spacelike everywhere or  null everywhere when restricted to  $M_{\ub}$.   
\end{theorem}
  The heuristic of the proof can be summarized as follows. Fix $\ub \ge 1$ and let $\ub'>\ub$ be suitable close to $\ub$. We consider the outgoing null cone $H$ originating from the MOTS $M_{\ub'}$, which intersects with incoming hypersurface $\Hb_{\ub}$ at $2$-sphere $\Sigma_{\ub}:=\{r=R'(\theta_1, \theta_2) \}\cap \Hb_{\ub}$. See Figure \ref{intro: proof M_ub' cap I^+(M_ub)=empty} as below.
\begin{figure}[h]
    \centering
     \tikzset{every picture/.style={line width=0.75pt}} 
     
\begin{tikzpicture}[x=0.75pt,y=0.75pt,yscale=-1,xscale=1]

\draw   (230.56,89.18) .. controls (230.56,78.68) and (264.01,70.16) .. (305.28,70.16) .. controls (346.55,70.16) and (380,78.68) .. (380,89.18) .. controls (380,99.68) and (346.55,108.2) .. (305.28,108.2) .. controls (264.01,108.2) and (230.56,99.68) .. (230.56,89.18) -- cycle ;
 
\draw    (220,59.36) -- (229.78,90.9) -- (253.6,167.76) ;

\draw    (391.6,56.16) -- (380.25,91.71) -- (355.6,168.96) ;

\draw  [dash pattern={on 0.84pt off 2.51pt}] (269.2,139.16) .. controls (269.2,134.41) and (285.32,130.56) .. (305.2,130.56) .. controls (325.08,130.56) and (341.2,134.41) .. (341.2,139.16) .. controls (341.2,143.91) and (325.08,147.76) .. (305.2,147.76) .. controls (285.32,147.76) and (269.2,143.91) .. (269.2,139.16) -- cycle ;

\draw    (269.2,139.16) -- (253.6,167.76) ;

\draw    (341.2,139.16) -- (355.6,168.96) ;

\draw    (253.6,167.76) .. controls (282,180.4) and (315.6,190.56) .. (355.6,168.96) ;

\draw  [dash pattern={on 0.84pt off 2.51pt}]  (253.6,167.76) .. controls (300,162.16) and (325.6,165.36) .. (353.6,167.76) ;

\draw    (246.4,44.56) -- (200.44,169.8) ;

\draw    (368.96,45.8) -- (405.6,176.56) ;

\draw (342.16,91) node [anchor=north west][inner sep=0.75pt]  [font=\tiny]  {$M_{\underline{u} '}$};

\draw (295.36,136.6) node [anchor=north west][inner sep=0.75pt]  [font=\tiny]  {$M_{\underline{u}}$};

\draw (298.96,183.8) node [anchor=north west][inner sep=0.75pt]  [font=\tiny]  {$\Sigma _{\underline{u} '}$};

\draw (336.16,154.2) node [anchor=north west][inner sep=0.75pt]  [font=\tiny]  {$\underline{H}_{\underline{u}}$};

\draw (394.56,122.2) node [anchor=north west][inner sep=0.75pt]  [font=\tiny]  {$\underline{H}_{\underline{u} '}$};

\draw (357.63,125.07) node [anchor=north west][inner sep=0.75pt]  [font=\tiny]  {$H$};

\end{tikzpicture}
\caption{}
    \label{intro: proof M_ub' cap I^+(M_ub)=empty}
\end{figure}
Noting that $\trch|_{M_{\ub'}}=0$, it follows from the Raychaudhuri equation that $\trch|_{\Sigma_{\ub}}\ge 0$. According to the definition of null comparison principle, we can conclude that $R'\le R$. Furthermore, based on the strong maximum principle for elliptic equations \cite{G-T} (see also \cite{A-G-H}), we must have either $R'<R$ pointwisely or $R'\equiv R$. As a consequence, we prove that the apparent horizon must be piecewise spacelike or null.

The concept of null comparison principle applies to various  scenarios, such as perturbative Schwarzschild spacetimes studied by Klainerman-Szeftel in \cite{K-S}, the anisotropic gravitational collapse spacetime investigated by the first author and Han in \cite{An-Han}, and the spacetime containing a naked singularity studied by the first author in \cite{An naked singularity}. The above \Cref{intro: NCP section NCP thm} allows us to validate the (local) achronality of the apparent horizon in these different scenarios.  
\vspace{3mm}

In  Klainerman-Szeftel's setting, let $\widetilde{\Sigma}=\{r=\tilde{R}(\theta, \varphi) \}$ be a spacelike $2$-surface and  $M_{\ub}=\{r=R(\theta, \varphi) \}$ be a MOTS which both lie in $\Hb_{\ub}$. Given a smooth 2-surface $\widetilde{\Sigma}$ near $M_{\ub}$, the key observation is that if the outgoing null expansion of $\widetilde{\Sigma}$ is non-negative, then the difference $\tR-R$ satisfies the elliptic inequality
\begin{equation*}
    \Delta_{S_{\tR}}(\tR-R)(\theta, \varphi)-d^i_2\f{\partial}{\partial \theta_i} (\tR-R)(\theta, \varphi)-(\f{1}{4m_\infty^2}+d_1)(\tR-R)(\theta, \varphi)\geq  0
\end{equation*}
with $|d^i_2|\ll 1,|d_1|\ll 1$. Applying the maximum principle then implies that $\tR \le R$.

As for the anisotropic gravitational-collapse spacetimes in \cite{An-Han}, we verify the null comparison principle by examining the linearized operator for the quasilinear elliptic equation of MOTS. More precisely, with double null foliation in $(u, \ub, \theta_1, \theta_2)$ coordinates, we define  
\begin{equation*}
     S(\tilde{\phi})=\f{1}{2}\O^{-1} \tR\tr \tilde{\chi}
\end{equation*}
for the spacelike $2$-surface $\widetilde{\Sigma}=\{u=1-\tilde{R}(\theta_1, \theta_2)= 1-\ub a e^{-\tilde{\phi}(\theta_1, \theta_2)}\}$. Here $\O$ represents the lapse function and remains positive,  $\tr \tilde{\chi}$ denotes the outgoing null expansion of $\widetilde{\Sigma}$ and $a$ is a fixed large parameter. Suppose $\phi(\theta_1, \theta_2)$ solves $S(\phi)=0$ and the corresponding MOTS is $M_{\ub}=\{u=1-R(\theta_1, \theta_2)=1-\ub a e^{-\phi(\theta_1, \theta_2)}  \}$. The null comparison principle can be verified by studying the linearized operator $\partial_{\phi} S(\phi)$.

With naked-singularity initial data, in \cite{An naked singularity} the first author showed that a tiny scale-invariant outgoing perturbation would lead to the apparent horizon formation and the naked-singularity censoring. Our null comparison principle can also apply to that setting.
\begin{remark}
    In \cite{Andersson: stability}, Andersson-Mars-Simon  explored the local achronality of apparent horizon through studying the stability operator of MOTS. They proved that, the apparent horizon is piecewise spacelike or null within a neighbourhood of the strictly stable outermost MOTS. They require that the linearized operator associated with the equation of MOTS possesses a strictly positive principle eigenvalue. Our null comparison principle is based on null foliations and we do not require the strictly positive condition. For example,  we  allow the principle eigenvalue of a corresponding elliptic operator to be zero and the null comparison principle is still applicable. 
\end{remark}

\subsubsection{Connections to Black Hole Thermodynamics}
With the proof of achronality for the apparent horizon $\mathcal{AH}$, we further investigate several formulas about the black hole thermodynamics along the apparent horizon $\mathcal{AH}$ first formulated in \cite{A-K2, A-K1} by Ashtekar-Krishnan. Denoting the null piece and the spacelike piece of  $\mathcal{AH}$ by $\mathcal{AH}_n$ and $\mathcal{AH}_s$, in \Cref{Sec: physical law} we verify and provide new proofs for
\begin{enumerate}
    \item The zeroth law: within the null piece $\mathcal{AH}_n$, the surface gravity  is a constant.
    \item The first law: along $\mathcal{AH}$, the change of energy can be expressed by the change of area and the change of angular momentum in the following differential form: 
\begin{equation*}
	d {E}_M=\f{\kappa}{8\pi}d A_M+\O d {J}_M,
\end{equation*}
where ${E}_M$ is the energy, $\kappa$ is the surface gravity, $A_M$ is the area, $\O$ is the angular velocity, ${J}_M$ is the angular momentum, with respect to the MOTS $M_{\ub}$. The specific definitions of these quantities are given in \Cref{Sec: physical law}.
    \item The second law: the area of MOTS $M_{\ub}$ is non-decreasing along $\mathcal{AH}=\cup_{\ub} M_{\ub}$ as the parameter $\ub$ grows. More precisely, it holds that
    $$\f{d}{d\ub}\text{Area}(M_{\ub})=0 \quad \text{in} \ \mathcal{AH}_n \quad \text{and} \quad \f{d}{d\ub}\text{Area}(M_{\ub})>0 \quad \text{in} \ \mathcal{AH}_s.$$
\end{enumerate}

\subsection{Acknowledgements}
XA is supported by MOE Tier 1 grants A-0004287-00-00, A-0008492-00-00 and MOE Tier 2 grant A-8000977-00-00.  TH acknowledges the support of the President Graduate Fellowship of NUS.

\section{Setup}\label{Sec: Basic setup}
We foliate our spacetime $(\mathcal{M} , \mathbf{g})$ by incoming null hypersurface $\Hb_{\ub}$ with $\ub$ being the incoming optical function. Let $\Lb^{\nu}:=-g^{\mu \nu}\partial_\mu \ub$ be the geodesic generator of $\ub$. Along $\Hb_{\ub}$, we then choose coordinate system $(r, \theta_1, \theta_2)$ with
\begin{equation*}
    \partial_r=f\Lb \quad \textrm{and} \quad f<0 \ \text{is a function to be fixed}.
\end{equation*}
The level set of $r$ on $\Hb_{\ub}$ are compact spacelike $2$-surfaces, and we denote them as $S_{\ub, r}$. Note that the function $r$ need not be a optical function.

Choosing  $e_3$ and $e_4$ to be the incoming and outgoing null vectors orthogonal to $S_{\ub , r}$\footnote{Note that $e_3$ must be parallel to $\Lb$.}, and normalizing them by $g(e_3, e_4)=-2$,  we then decompose curvature components and Ricci coefficients with respect to the null frame $\{ e_1, e_2, e_3, e_4\}$. Here $(e_a)_{a=1, 2}$ is an orthonormal tangent frame on $S_{\ub, r}$. We first define null curvature components of Riemann tensor
 \begin{equation}\label{def curvatures}
\begin{split}
\a_{ab}&=\mathbf{R}(e_a, e_4, e_b, e_4),\quad \, \,\,   \ab_{ab}=\mathbf{R}(e_a, e_3, e_b, e_3),\\
\b_a&= \frac 1 2 \mathbf{R}(e_a,  e_4, e_3, e_4) ,\quad \bb_a =\frac 1 2 \mathbf{R}(e_a,  e_3,  e_3, e_4),\\
\rho&=\frac 1 4 \mathbf{R}(e_4,e_3, e_4,  e_3),\quad \sigma=\frac 1 4  \,^*\mathbf{R}(e_4,e_3, e_4,  e_3). 
\end{split}
\end{equation}
Here $\, ^*\mathbf{R}$ stands for the Hodge dual of $\mathbf{R}$. Denoting $\bfD_a:=\bfD_{e_{a}}$, we proceed to define Ricci coefficients
\begin{equation}\label{def Ricci coefficients}
\begin{split}
&\chi_{ab}=\mathbf{g}(\bfD_a e_4,e_b),\, \,\, \quad \chib_{ab}=\mathbf{g}(\bfD_a e_3,e_b),\\
&\xi_a=\frac 12 \mathbf{g}(\bfD_4 e_4,e_a), \ \ \quad \xib_a=\frac 12 \mathbf{g}(\bfD_3 e_3,e_a),\\
&\eta_a=-\frac 12 \mathbf{g}(\bfD_3 e_a,e_4),\quad \etab_a=-\frac 12 \mathbf{g}(\bfD_4 e_a,e_3),\\
&\omega=-\frac 14 \mathbf{g}(\bfD_4 e_3,e_4),\quad\,\,\, \omegab=-\frac 14 \mathbf{g}(\bfD_3 e_4,e_3),\\
&\zeta_a=\frac 1 2 \mathbf{g}(\bfD_a e_4,e_3).
\end{split}
\end{equation}
Letting $\gamma_{ab}$ be the induced metric on $S_{\ub, r}$, we further decompose $\chi, \chib$ into
$$\chi_{ab}=\f12\tr\chi\cdot \gamma_{ab}+\chih_{ab}, \quad \chib_{ab}=\f12\tr\chib\cdot \gamma_{ab}+\chibh_{ab},$$
and $\chih_{ab}$, $\chibh_{ab}$ are the corresponding traceless parts. 

In below sections, we define $\nab$ to be the induced covariant derivative on $S_{\ub, r}$ and Let $\nab_3$ and $\nab_4$ be the projections of covariant derivatives $\bfD_3$, $\bfD_4$ to $S_{\ub, r}$. For the sake of convenience, we also denote 
\begin{equation*}
    \kappa:=\trch, \quad \underline{\kappa}:=\trchb.
\end{equation*}
With a given a scalar function $f$ on $\mathcal{M}$, we set
\begin{equation*}
    \overline{f}:=\int_{S_{\ub ,r}} f, \quad  \check{f}:=f-\overline{f}.
\end{equation*}

\subsection{Incoming Geodesic Foliation in Klainerman-Szeftel's Work \texorpdfstring{\cite{K-S}}{}}
Recall that in \cite{K-S}, Klainerman and Szeftel proved the global existence of solution to Einstein vacuum equations under axially symmetric polarized perturbations, which can be stated as follows: 
\begin{theorem}[Klainerman-Szeftel \cite{K-S}]\label{K-S thm}
	There exists a globally hyperbolic development $\mathcal{M}$ arising from an axially symmetric polarized perturbation  of Schwarzschild initial data set, which can be separated into two parts $\mathcal{M}={}^{(ext)}\mathcal{M}\cup {}^{(int)}\mathcal{M}$. Furthermore, in the interior region ${}^{(int)}\mathcal{M}=\mathcal{M}\cap\{ r\leq r_{\mathcal{T}}\}$, with coordinate system $(\ub, r, \theta, \varphi)$\footnote{Here $r$ denotes the area radius of $S_{\ub, r}$.} we have
	\begin{enumerate}
		\item the spacetime metric $\mathbf{g}$ takes the form
		\begin{equation}\label{KS metric}
		\mathbf{g}=-\f{4\vsgmb}{r\kbb}d\ub dr+\f{\vsgmb^2(\kb+A)}{\kbb}d\ub^2+\gamma(d\theta-\f{1}{2}\vsgmb b d\ub-\f{\ud{b}}{2}\underline{\Theta})^2+e^{2\Phi}d\varphi^2,
		\end{equation}
		where
		\begin{equation}\label{KS metric component}
		b=e_4(\theta),\quad \ud{b}=e_3(\theta), \quad \gamma=\f{1}{(e_\theta(\theta))^2}, \quad \ud{\Theta}:=\f{4\vsgmb}{r\kb}dr-\vsgmb^2\f{(\kb+A)}{\kbb}d\ub.
		\end{equation}
		\item the coordinate derivatives can be expressed as 
		\begin{equation*}
		\begin{split}
		&\partial_r=\f{2}{r\kbb}e_3-\f{2\sqrt{\gamma}}{r\kbb}\ud{b}e_\theta, \quad \partial_\theta=\sqrt{\gamma}e_\theta, \\
		&\partial_{\ub} =\f{1}{2}\vsgmb\Big[e_4-\f{\kb+A}{\kbb}e_3-\sqrt{\gamma}(b-\f{\kb+A}{\kbb}\ud{b})e_\theta\Big].
		\end{split}
		\end{equation*}
	\item there exist some integer $k_{small}\ge 2$ and  sufficiently small constants $\varepsilon_0, \delta_{dec}>0$ such that, with $1\le k\le k_{small}$, for the Hawking mass $m(\ub, r):=\f{r}{2}(1+\f{1}{16\pi}\int_{S_{\ub, r}} \kappa \ud{\kappa})$,  we have
	\begin{equation*}
	\f{|m-m_\infty|}{m_0}+ |\mathfrak{d}^k m|\lesssim \f{\varepsilon_0}{\ub^{1+\delta_{dec}}},
	\end{equation*}
	where  $\mathfrak{d}:=\{e_3, re_4, r\nab \}$, $m_0$ is the initial mass, $m_\infty$ denotes the final Bondi mass and they satisfy
	\begin{equation*}
	|\f{m}{m_0}-1|+|\f{m_\infty}{m_0}-1|\lesssim \f{\varepsilon_0}{\ub^{1+\delta_{dec}}}.
	\end{equation*}
	\item for $\check{\Gamma}=\{\check{\ud{\kappa}}, \chibh, \zeta, \etab, \check{\kappa}, \chih, \check{\omega}, \xi \}$ and $\check{R}=\{\a,\b, \check{\rho}, \sigma, \bb, \ab
	\}$, with $0\leq k\leq k_{small}$, they obey
	\begin{equation*}
	|\mathfrak{d}^k (\check{\Gamma}, \check{R})|\lesssim \f{\varepsilon_0}{\ub^{1+\delta_{dec}}}.
	\end{equation*}
	Note that the Ricci coefficient $\omegab$ is vanishing in their setting.
	\item  for the averages of $\ud{\kappa}, \kappa, \omega, \rho$, with $0\leq k\leq k_{small}$, the below estimates hold
	\begin{equation*}
	|\mathfrak{d}^k (\kbb+\f{2}{r}, \overline{\omega}+\f{m_{\infty}}{r^2}, \kb-\f{2}{r}(1-\f{2m_{\infty}}{r}), \overline{\rho}+\f{2m_{\infty}}{r^3})|\lesssim \f{\varepsilon_0}{\ub^{1+\delta_{dec}}}.
	\end{equation*}
	\item for the metric coefficients, they are bounded via
	\begin{equation*}
	\max\limits_{0\leq k\leq k_{small}}\sup\limits_{{}^{(int)}\mathcal{M}}\ub^{1+\delta_{dec}}\Big(|\mathfrak{d}^k(\vsgmb-1)|+|\mathfrak{d}^k(\f{\gamma}{r^2}-1)+|\mathfrak{d}^k b|+|\mathfrak{d}^k\ud{b}|+|\mathfrak{d}^k(\f{e^\Phi}{r\sin \theta}-1)|+|\mathfrak{d}^k A|\Big) \lesssim \varepsilon_0.
	\end{equation*}
\end{enumerate}
\end{theorem}
Based on above hyperbolic existence theorem, we conduct a deformation on $\Hb_{\ub}$ to solve for the MOTS.
\subsection{Deformation Equation}
 Along each incoming hypersurface $\Hb_{\ub}$, we will solve for a deformed $2$-sphere $M$. Each point on $M$ has coordinates $(\ub, r, \theta_1 ,\theta_2)=(\ub, R(\ub, \theta_1 ,\theta_2), \theta_1 ,\theta_2)$, with $R(\ub, \theta_1 ,\theta_2)$ being a function to be solved. Based on $M$, the corresponding null frame adapted to it, denoted by $(e'_a, e'_3, e'_4)\ (a=1,2)$, can be expressed as
\begin{equation}\label{new frames}
e_3'=e_3, \quad e_a '=e_a+f e_a (R) e_3, \quad e_4'=e_4+2fe^a(R) e_a+f^2|\nab R|^2 e_3,
\end{equation}
where $f=e_3(r)^{-1}<0$. With this frame, we have
$$e_a '(\ub)=0,  \quad e_a '(r-R)=f e_a (R) e_3(r)-e_a (R)=0.$$ Noting that $e_3$ is orthogonal to $T\Hb_{\ub}=\text{span}\{e_1, e_2, e_3 \}$, with this we can verify that
\begin{equation*}
g(e_a', e_b')=g(e_a, e_b)=\delta_{ab},\quad g(e_4', e_a')=g(e_4', e_4')=0,\quad g(e_3', e_4')=-2.
\end{equation*}
Since $e_3(R)=f^{-1}\partial_r (R)=0$, we also have $e_a '(R)=e_a (R)$. 

We proceed to derive the formula for $\trch'$.
\begin{lemma}\label{proposition deformation formula}
	The trace of the null second fundamental form $\chi'$, relatively to the new frame \eqref{new frames}, can be expressed as
	\begin{equation}\label{deformation formula}
	\trch'=\trch+2f\Delta R+2(\nab f+f(\eta+\zeta))\cdot \nab R-4f^2\chibh_{bc}\nab^b R\nab^c R+(2fe_3(f)-f^2\trchb-4\omegab f^2)|\nab R|^2.
	\end{equation}
 Here $\Delta$ represents the Laplace-Beltrami operator on  2-sphere $S_{\ub, r}$.
\end{lemma}
\begin{proof}
	Denoting $F=f\nab R$, by \eqref{new frames} we have $$e_4'=e_4+2F+|F|^2 e_3, \quad e_a'=e_a+F_a e_3.$$
	Applying the definition, we get
	\begin{equation}\label{formula for chi'}
	\begin{split}
	\chi'_{ab}=\bfg(D_a e_4, e_b)=\chi_{ab}&+(\nab_a F_b+\nab_b F_a)+\nab_3 (F_a F_b)+(\eta_b+\zeta_b)F_a+(\eta_a+\zeta_a)F_b \\
	&+|F|^2 \chib_{ab}-F_bF^c \chib_{ac}-F_aF^c\chib{bc}-4\omegab F_a F_b.
	\end{split}
	\end{equation}
	Taking the trace, we then deduce
	\begin{equation}
	\trch'=\trch+2\div F+\nab_3|F|^2+2(\eta+\zeta)\cdot F-2\chih_{bc} F^b F^c-4\omegab |F|^2.
	\end{equation}
    Using $\nab_3 R=0$ and the commutation formula for a scalar function
    \begin{equation*}
    [\nab_3, \nab]h=(\eta-\zeta)\nab_3 h-\chib \cdot \nab h,
    \end{equation*}
    we derive
    \begin{equation}
    \begin{split}
    \nab_3 F= \nab_3(f\nab R)=&e_3(f) \nab R+ f\nab_3  \nab R
    =(e_3(f)-f\chib)\cdot \nab R.
    \end{split}
    \end{equation}
    This further gives
    \begin{equation}
    \nab_3 |F|^2=2\nab_3 F \cdot F=(2f e_3(f)-f^2\trchb)|\nab R|^2-2f^2 \chibh_{bc}\nab^b R\nab^c R.
    \end{equation}
    Combining above together, we arrive at
    \begin{equation}
    \begin{split}
    \trch'=&\trch+2 \div(f\nab R)+(2f e_3(f)-f^2\trchb)|\nab R|^2-2f^2 \chibh_{bc}\nab^b R\nab^c R \\
    &+2f(\eta+\zeta)\cdot \nab R-2f^2 \chibh_{bc}\nab^b R \nab^c R-4\omegab f^2 |\nab R|^2 \\
    =&\trch+2f\Delta R+2\Big(\nab f+f(\eta+\zeta)\Big)\cdot \nab R-4f^2 \chibh_{bc}\nab^b R \nab^c R \\&+(2f e_3(f)-f^2\trchb-4f^2 \omegab)|\nab R|^2.
    \end{split}
    \end{equation}
\end{proof}
With the intention of future use, here we reformulate $\Delta R$.
\begin{lemma}(\cite{An: AH})\label{change laplacian} On $M$, its Laplace-Beltrami operator $\Delta'_M$  satisfies
	\begin{equation}\label{3.11} 
	\Delta_M R =\Delta'_M R+2f \chibh_{bc}\nab^b R\nab^c R.
	\end{equation}
\end{lemma}
\begin{proof}
	The proof of this lemma is similar to the corresponding one in \cite{An: AH}. But here we have an additional factor $f$. Denoting $(e_3 e_a) R=e_3 e_a (R)=\bfD^2_{e_3, e_a} R$,  we have
	\begin{equation*}
	\begin{split}
		(e_3 e_a) R=(e_a e_3) R=e_a(e_3(R))-(\bfD_{e_a} e_3 )R
		=-(\chib_{ab}e^b+\zeta_a e_3)R=-\chib_{ab}e^b(R).
	\end{split}
	\end{equation*}
	Recalling $e_3(R)=0, e_{a'}(R)=e_a(R)$ and noting that $\bfD_{e_3} e_3=-2\omegab e_3$ (since $e_3$ is pre-geodesic  in ${}^{(int)} \mathcal{M}$), we deduce 
	\begin{equation}\label{3.10}
	\begin{split}
	&e^a e_a(R)\\
	=&(e^{a'}-fe^a(R) e_3)\Big(e_{a'} (R)-fe_a(R) e_3(R)\Big)-\Big(\bfD_{e^{a'}-fe^a(R) e_3} (e_{a'}-fe_a(R) e_3)\Big) R \\
	=&e^{a'}e_{a'} (R)-fe^a(R)e_3 e_a (R)+fe_a(R)(\bfD_{e^a}e_3) R\\
	=&e^{a'}e_{a'} (R)+f\trchb\nab^a R \nab_a R+2f\chibh_{ac}\nab^a R \nab^c R.
	\end{split}
	\end{equation}
	
	On the other side, we also have
	\begin{equation*}
	\bfD_{e^{a'}}e_{a'}(R)=\overline{\bfD_{e^{a'}}e_{a'}}(R)-\f{1}{2}\bfg(\bfD_{e^{a'}}e_{a'}, e_3')e_4'(R)-
\f{1}{2}\bfg(\bfD_{e^{a'}}e_{a'}, e_4')e_3'(R),
	\end{equation*}
	where $\overline{\bfD_{e^{a'}}e_{a'}}$ is the projection of $\bfD_{e^{a'}}e_{a'}$ to the tangent space $T_p M$ at point $p$.  Thus, we can write
	\begin{equation*}
	\begin{split}
	e^{a'}e_{a'}(R)=&e^{a'}\big(e_{a'}(R)\big)-\bfD_{e^{a'}}e_{a'}(R) \\
	=&e^{a'}\big(e_{a'}(R)\big)-\overline{\bfD_{e^{a'}}e_{a'}}(R)-\f12 \trchb e_4'(R) \\
	=& \Delta'_M R-\f12 \trchb' e_4'(R).
	\end{split}
	\end{equation*}
	Combining with \eqref{3.10}, we then derive
	\begin{equation}\label{3.10.1}
	\begin{split}
	e^a e_a(R)&=e^{a'}e_{a'} (R)+f\trchb\nab^a R \nab_a R+2f\chibh_{ac}\nab^a R \nab^c R\\
	&=	\Delta'_M R-\f{1}{2}\trchb e_4'(R)+f\trchb\nab^a R \nab_a R+2f\chibh_{ac}\nab^a R \nab^c R.
	\end{split}
	\end{equation}
	Meanwhile, observing 
	\begin{equation*}
	\begin{split}
		\bfD_{e^{a}}e_{a}(R)=&\overline{\bfD_{e^{a}}e_{a}}(R)-\f{1}{2}\bfg(\bfD_{e^{a}}e_{a}, e_3)e_4(R)-\f{1}{2}\bfg(\bfD_{e^{a}}e_{a}, e_4)e_3(R) \\
		=&\overline{\bfD_{e^{a}}e_{a}}(R)+\f{1}{2}\trchb e_4 (R),
	\end{split}
	\end{equation*}
	it yields
	\begin{equation*}
	\begin{split}
	e^{a}e_{a}(R)=&e^{a}\big(e_{a}(R)\big)-\bfD_{e^{a}}e_{a}(R) \\
	=&e^{a}\big(e_{a}(R)\big)-\overline{\bfD_{e^{a}}e_{a}}(R)-\f{1}{2} \trchb e_4 (R)\\
	=& \Delta_M R-\f{1}{2}\trchb e_4(R).
	\end{split}
	\end{equation*}
	Comparing with \eqref{3.10.1} and utilizing \eqref{new frames}, we arrive at \eqref{3.11}.
\end{proof}

Back to \Cref{proposition deformation formula}, we can rewrite it as 
\begin{proposition}
	The trace of the null second fundamental form $\chi'$, relatively to the new frame \eqref{new frames}, satisfies
	\begin{equation}\label{null expansion on MOTS}
	\trch'=\trch+2f\Delta'_M R+2(\nab f+f(\eta+\zeta))\cdot \nab R+(2fe_3(f)-f^2\trchb-4\omegab f^2)|\nab R|^2.
	\end{equation}
\end{proposition}

\subsection{An Elliptic Operator Associated with the MOTS}\label{Subsec: MOTS eqn and its derivative}
Let $I$ be a connected interval. For fixed constants $0<R_1<R_2$,  with $\ub \in I $ and $R=R(\theta_1, \theta_2)\in  C^2(\mathbb{S}^2)$ satisfying $R_1\le R\le R_2$, we consider the below elliptic operator at $S_{\ub, R}$: 
\begin{equation*}
\begin{split}
    L(R, \ub):=(\f{1}{2} f^{-1} \trch')|_{S_{\ub, R}}=
&\Delta'_M R+\big(f^{-1}\nab f+(\eta+\zeta)\big)\cdot \nab R\\&+(e_3(f)-\f{1}{2}f\trchb-2\omegab f)|\nab R|^2+\f{1}{2}f^{-1}\trch.
\end{split}
\end{equation*}
Set
\begin{equation*}
    \mathscr{G}:=\{\mathbf{g}(\prub, e_4), \mathbf{g}(\prub, e_3), \mathbf{g}(\prub, e_a) \} \quad \text{with} \quad a=1,2.
\end{equation*}
We also denote
\begin{equation*}
    M:=\sum_{k=0}^2\max\limits_{I\times [R_1, R_2]\times \mathbb{S}^2} |(e_4, e_3, \nab)^{k} (\Gamma, \mathscr{R}, \log |f|,  \mathscr{G})|(\ub, r, \theta_1, \theta_2).
\end{equation*}
Here $\Gamma$ represents the collection of all the Ricci coefficients, and $\mathscr{R}$ denotes the set comprising all the curvature components.

We proceed to calculate the Fr\'echet derivative of $L$ with respect to $R$.
\begin{lemma}\label{partial R L}
    For any $(\ub, R)\in I \times C^2(\mathbb{S}^2)$, with $R_1\le R\le R_2$, we have
 \begin{equation*}
\begin{split}
&\partial_R L(R, \ub)[W]=\Delta'_{S_{\ub, R}} W+B_a \nab^a W+\l C+\f{1}{2}f\nab_3 (f^{-1}{\rm tr\chi}) \r W
\end{split}
\end{equation*}
with  $B_a, C$  satisfying
\begin{equation*}
    |B_a|+|C|\le C(M, R_1, R_2)\Big[|\nab^2 R|+|\nab R|+|\nab R|^2 \Big].
\end{equation*}
\end{lemma}
\begin{proof}
According to the definition of Fr\'echet derivative, for any $W\in  C^2(\mathbb{S}^2)$, we obtain
\begin{equation*} 
\begin{split}
\partial_R L(R, \ub)[W]
=&\lim\limits_{\varepsilon\to 0}\f{1}{\varepsilon}\Big(F(R+\varepsilon W, \ub)-F(R, \ub)\Big)\\
=&\lim\limits_{\varepsilon\to 0}\f{1}{\varepsilon}\Big(\Delta'_{S_{\ub, R+\varepsilon W}} (R+\varepsilon W)-\Delta'_{S_{\ub, R}} R\Big)\\
&+\lim\limits_{\varepsilon\to 0}\f{1}{\varepsilon}\Big[\l(e_3(f)-\f{1}{2}f\trchb-2\omegab f)|\nab (R+\varepsilon W)|^2\r|_{\ub, R+\varepsilon W}\\&-\l(e_3(f)-\f{1}{2}f\trchb-2\omegab f)|\nab R|^2\r|_{\ub, R} \Big]\\
&+\lim\limits_{\varepsilon\to 0}\f{1}{\varepsilon}\Big[ \l\big(f^{-1}\nab f+(\eta+\zeta)\big)\cdot \nab (R+\varepsilon W)+\f{1}{2}f^{-1}\trch\r|_{\ub, R+\varepsilon W}\\&-\l\big(f^{-1}\nab f+(\eta+\zeta)\big)\cdot \nab R+\f{1}{2}f^{-1}\trch\r|_{\ub, R}\Big]\\
=:&I_1+I_2+ I_3.
\end{split}
\end{equation*}
To evaluate $I_1, I_2$, $I_3$, we appeal to the following commutation formulas in \cite{An: AH}. For  
any scalar function $f$, we have
	\begin{equation}\label{commute nab3 nab}
	\begin{split}
	[\nab_3, \nab]f=\f12(\eta+\etb)\nab_3 f-\chib\cdot\nab f 
	\end{split}
	\end{equation}
and	
 \begin{equation}\label{commute nab3 laplace}
	\begin{split}
	[\nab_3, \D]f=&-\tr\chib \D f-2\chibh\cdot\nab^2 f+\beb\cdot\nab f+\f12(\eta+\etb)\cdot \nab_3 \nab f-\eta\cdot\chibh\cdot \nab f\\
	&+\f12 \tr\chib\eta\cdot\nab f+\div\l\f12(\eta+\etb)\nab_3 f\r-\div \chib\cdot \nab f.
	\end{split}
	\end{equation}
We start to compute $I_i$ for $i=1, 2, 3$.
\begin{lemma}\label{est I1}
For $I_1$, the following equality holds
	\begin{equation*}
	\begin{split}
	I_1=&\lim\limits_{\varepsilon\to 0}\f{1}{\varepsilon}\Big(\Delta'_{S_{\ub, R+\varepsilon W}} R-\Delta'_{S_{\ub, R}} R\Big) 
	=f\nab_3(\Delta_{S_{\ub, \tR}} \tR)W-f\nab_3\Big(2f \chibh_{ab}\nab^a R \nab^b R\Big)W.
	\end{split}
	\end{equation*}
\end{lemma}
\begin{proof}
	Set $g$ and $\theta_1, \theta_2$ to be the induced metric and independent angular variables on $S_{\ub, r}$. For a scalar function $f$, we have
	\begin{equation}\label{3.13}
	\begin{split}
	\D_{S_{\ub, r}} f=&\f{1}{\sqrt{\det g}}\f{\partial}{\partial \theta_i} (\sqrt{\det g}\,g^{\theta_i \theta_l}\f{\partial f}{\partial \theta_l})\\
	=&g^{\theta_1 \theta_1}\f{\partial^2 f}{\partial \theta_1 \partial \theta_1}+g^{\theta_2 \theta_2}\f{\partial^2 f}{\partial \theta_2 \partial \theta_2}+2 g^{\theta_1 \theta_2}\f{\partial^2 f}{\partial \theta_1 \partial \theta_2}+\f{\partial}{\partial \theta_1} (g^{\theta_1 \theta_1})\f{\partial f}{\partial \theta_1}+\f{\partial}{\partial \theta_1} (g^{\theta_1 \theta_2})\f{\partial f}{\partial \theta_2}\\
	&+\f{\partial}{\partial \theta_2} (g^{\theta_2 \theta_1})\f{\partial f}{\partial \theta_1}+\f{\partial}{\partial \theta_2} (g^{\theta_2 \theta_2})\f{\partial f}{\partial \theta_2}+\f12 g^{\theta_k \theta_j}\f{\partial g_{\theta_j \theta_k}}{\partial \theta_i} g^{\theta_i \theta_l}\f{\partial f}{\partial \theta_l},
	\end{split}
	\end{equation}
	where $i,j,k,l=1,2$ and $g^{\theta_k \theta_j}$ depends on $(\ub, r, \theta, \varphi)$. Here we also use the formula
	$$\f{\partial}{\partial \theta_i}\det g=\det g \cdot g^{\theta_k \theta_j} \cdot \f{\partial g_{\theta_j \theta_k}}{\partial \theta_i}.$$
	Applying \Cref{change laplacian}, together with $\partial_r R=0$, we deduce
	\begin{equation*}
	\begin{split}
	\lim\limits_{\varepsilon\to 0}\f{1}{\varepsilon}\Big(\Delta'_{S_{\ub, R+\varepsilon W}} R-\Delta'_{S_{\ub, R}} R\Big) 
	=&\lim\limits_{\varepsilon\to 0}\f{1}{\varepsilon}\Big(\Delta_{S_{\ub, R+\varepsilon W}} R-\Delta_{S_{\ub, R}} R\Big)-\lim\limits_{\varepsilon\to 0}\f{1}{\varepsilon}\Big[ \l 2f \chibh_{bc}\nab^b \tR\nab^c \tR \r|_{S_{\ub, \tR+\varepsilon W}}\\&-\l2f \chibh_{bc}\nab^b \tR\nab^c \tR\r |_{S_{\ub, \tR}}\Big]\\
	=&\f{\partial}{\partial r}(\Delta_{S_{\ub, \tR}} \tR) W-\f{\partial}{\partial r}\Big(2f \chibh_{bc}\nab^b \tR\nab^c \tR\Big)W \\
	=&f\nab_3(\Delta_{S_{\ub, \tR}} \tR)W-f\nab_3\Big(2f \chibh_{ab}\nab^a R \nab^b R\Big)W,
	\end{split}
	\end{equation*}
 where in the last line we use the fact that $\partial_r=fe_3$.
\end{proof}

Applying \Cref{est I1}, along with \eqref{commute nab3 nab} and \eqref{commute nab3 laplace} we can express $I_1$ and $I_2$ as
\begin{equation*}
\begin{split}
I_1=&\lim\limits_{\varepsilon\to 0}\f{1}{\varepsilon}\Big(\Delta'_{S_{\ub, R+\varepsilon W}} (R+\varepsilon W)-\Delta'_{S_{\ub, R}} R\Big) \\
=&f\nab_3(\Delta_{S_{\ub, R}} R)W-f\nab_3\Big(2f \chibh_{ab}\nab^a R \nab^b R\Big)W+\Delta'_{S_{\ub, R}} W\\
=&f\Big([\nab_3, \Delta]R+\Delta(\nab_3 R)\Big)W-f\nab_3\Big(2f \chibh_{ab}\nab^a R \nab^b R\Big)W+\Delta'_{S_{\ub, R}} W \\
=&f\Big(-\tr\chib \D R-2\chibh\cdot\nab^2 R+\beb\cdot\nab R+\f12(\eta+\etb)\cdot \nab_3 \nab R-\eta\cdot\chibh\cdot \nab R
+\f12 \tr\chib\eta\cdot\nab R\\&+\div\l\f12(\eta+\etb)\nab_3 R\r-\div \chib\cdot \nab R-\nab_3 \chibh_{ab}\nab^a R \nab^b R+2\chibh_{ab}\chib^a_{c}\nab^c R \nab^b R \Big) W
+\Delta'_{S_{\ub, R}} W,
\end{split}
\end{equation*}
\begin{equation*}
\begin{split}
I_2=&\lim\limits_{\varepsilon\to 0}\f{1}{\varepsilon}\Big[\l(e_3(f)-\f{1}{2}f\trchb-2\omegab f)|\nab (R+\varepsilon W)|^2\r|_{\ub, R+\varepsilon W}-\l(e_3(f)-\f{1}{2}f\trchb-2\omegab f)|\nab R|^2\r|_{\ub, R} \Big]\\
=&\f{\partial}{\partial r}\l (e_3(f)-\f{1}{2}f\trchb-2\omegab f)|\nab R|^2\r  W+2(e_3(f)-\f{1}{2}f\trchb-2\omegab f )\nab R \cdot \nab W \\
=&\f{\partial}{\partial r}(e_3(f)-\f{1}{2}f\trchb-2\omegab f )|\nab R|^2 W+(e_3(f)-\f{1}{2}f\trchb-2\omegab f )\f{\partial}{\partial r}|\nab R|^2 W\\
&+2(e_3(f)-\f{1}{2}f\trchb-2\omegab f )\nab R \cdot \nab W \\
=&fe_3\Big(e_3(f)-\f{1}{2}f\trchb-2\omegab f \Big)|\nab R|^2 W+f(e_3(f)-\f{1}{2}f\trchb-2\omegab f )(-\trchb|\nab R|^2-2\chibh_{bc}\nab^b R\nab^c R) W\\
&+2(e_3(f)-\f{1}{2}f\trchb-2\omegab f )\nab R \cdot \nab W.
\end{split}
\end{equation*}
Note that in the derivation of $I_2$ we use the equality
\begin{equation*}
\nab_3 |\nab R|^2=2\nab_3 \nab R \cdot \nab R=-\trchb|\nab R|^2-2\chibh_{bc}\nab^b R\nab^c R.
\end{equation*}
In the same manner, we also obtain
\begin{align*}
    I_3=&\partial_r \l\big(f^{-1}\nab f+(\eta+\zeta)\big)\cdot \nab R+\f{1}{2}f^{-1}\trch\r W+\big(f^{-1}\nab f+(\eta+\zeta)\big)\cdot \nab W\\
    =&f\l \nab_3 \big(f^{-1}\nab f+(\eta+\zeta)\big)-\big(f^{-1}\nab f+(\eta+\zeta)\big)\cdot \chib \r \cdot \nab R \\&+\f{1}{2}\nab_3 (f^{-1}\trch) \r W+\big(f^{-1}\nab f+(\eta+\zeta)\big)\cdot \nab W.
\end{align*}
Summarizing all the above, we then conclude
\begin{equation*}
\begin{split}
&\partial_R L(R, \ub)[W]=I_1+I_2+ I_3
=\Delta'_{S_{\ub, \tR}} W+B_a \nab^a W+\l C+\f{1}{2}f\nab_3 (f^{-1}\trch) \r W,
\end{split}
\end{equation*}
where
\begin{equation*}
    |B_a|+|C|\le C(M, R_1, R_2)\Big[|\nab^2 R|+|\nab R|+|\nab R|^2 \Big].
\end{equation*}
\end{proof}

For future use, in below we also evaluate $\partial_{\ub} L$.
\begin{lemma}\label{partial ub L}
     For any $(\ub, R)\in I \times C^2(\mathbb{S}^2)$  with $R_1\le R\le R_2$, the following inequality holds
     \begin{equation*}
         |\partial_{\ub} L(R, \ub)-\f{1}{2}\partial_{\ub}(f^{-1}{\rm tr\chi})|\le C(M, R_1, R_2)\Big[|\nab^2 R|+|\nab R|+|\nab R|^2 \Big].
     \end{equation*}
\end{lemma}
\begin{proof}
Via a direct computation, together with \Cref{change laplacian} and $\prub R=0$, we obtain
\begin{align*}
    \partial_{\ub} L(R, \ub)=&[\partial_{\ub}, \D] R-2\prub (f \chibh)_{ab}  \nab^a R \nab^b R-4f \chibh_{ab}[\prub, \nab^a] R \nab^b R
+\prub \big(f^{-1}\nab f+(\eta+\zeta)\big)\cdot \nab R\\&+\big(f^{-1}\nab f+(\eta+\zeta)\big)\cdot [\prub, \nab] R+\prub (e_3(f)-\f{1}{2}f\trchb-2\omegab f)|\nab R|^2\\&+2(e_3(f)-\f{1}{2}f\trchb-2\omegab f)[\prub, \nab] R \cdot \nab R+\f12 \prub(f^{-1}\trch).
\end{align*}
    To estimate $\partial_{\ub} L(R, \ub)$, we appeal to the commutation formulae
\begin{equation*}
\begin{split}
[e_4, \nab]R=&\f12(\eta+\etb)e_4 R-\chi\cdot\nab R, \\ 
[e_4, \D]R
=&-\tr\chi \D R-2\chih\cdot\nab^2 R+\beta\cdot\nab R+\f12(\eta+\etb)\cdot \nab_4 \nab R-\etb\cdot\chih\cdot \nab R\\
&+\f12 \tr\chi\etb\cdot\nab R+\div\l\f12(\eta+\etb)e_4 R\r-\div \chi\cdot \nab R.
\end{split}
\end{equation*}
Also notice that
$$\partial_{\ub} =-\f{1}{2}\mathbf{g}(\partial_{\ub}, e_3)e_4-\f12\mathbf{g}(\partial_{\ub}, e_4)e_3+\mathbf{g}(\partial_{\ub}, e_a)e^a.$$
Plugging all equalities above into the expression of $\partial_{\ub} L(R, \ub)$, combining with the definition of $M$, we thus conclude
\begin{equation*}
   |\partial_{\ub} L(R, \ub)-\f{1}{2}\partial_{\ub}(f^{-1}\trch)|\le C(M, R_1, R_2)\Big[|\nab^2 R|+|\nab R|+|\nab R|^2 \Big].
\end{equation*}
\end{proof}

\section[Existence and Asymptotics of the Apparent Horizon in the Setting of Klainerman-Szeftel]{\texorpdfstring{Existence and Asymptotics of the Apparent Horizon \\ in the Setting of Klainerman-Szeftel}{Existence and Asymptotics of the Apparent Horizon in the Setting of Klainerman-Szeftel}}\label{sec existence and asymptotics of apparent horizon in K-S setting}

\subsection{Coordinates transformation}
In \cite{K-S}, when Klainerman-Szeftel studied the interior region ${}^{(int)}\mathcal{M}$, they adopted the incoming geodesic foliation with coordinate $(\ub, r, \theta, \varphi)$.

 In ${}^{(int)}\mathcal{M}$, their spacetime metric takes the form
	\begin{equation*}
	g=-\f{4\vsgmb}{r\kb}d\ub dr+\f{\vsgmb^2(\kb+A)}{\kbb}d\ub^2+\gamma(d\theta-\f{1}{2}\vsgmb b d\ub-\f{\ud{b}}{2}\underline{\Theta})^2+e^{2\Phi}d\varphi^2,
	\end{equation*}
	where
	\begin{equation*}
	b=e_4(\theta),\quad \ud{b}=e_3(\theta), \quad \gamma=\f{1}{(e_\theta(\theta))^2},
	\quad \ud{\Theta}:=\f{4\vsgmb}{r\kb}dr-\vsgmb^2\f{(\kb+A)}{\kbb}d\ub.
	\end{equation*}
 The coordinate derivatives in ${}^{(int)}\mathcal{M}$ can also be expressed as 
	\begin{equation}\label{coordinates derivatives}
	\begin{split}
	&\partial_r=\f{2}{r\kbb}e_3-\f{2\sqrt{\gamma}}{r\kbb}\ud{b}e_\theta,\quad \partial_\theta=\sqrt{\gamma}e_\theta, \\
	&\partial_{\ub} =\f{1}{2}\vsgmb\Big[e_4-\f{\kb+A}{\kbb}e_3-\sqrt{\gamma}(b-\f{\kb+A}{\kbb}\ud{b})e_\theta\Big].
	\end{split}
	\end{equation}
Moreover, they show that for some integer $k_{small}\ge 2$ and  sufficiently small constants $\varepsilon_0, \delta_{dec}>0$, the following estimates hold 
	\begin{equation}\label{estimates for metric}
	\max\limits_{0\leq k\leq k_{small}}\sup\limits_{{}^{(int)}\mathcal{M}}\ub^{1+\delta_{dec}}\Big(|\mathfrak{d}^k(\vsgmb-1)|+|\mathfrak{d}^k(\f{\gamma}{r^2}-1)+|\mathfrak{d}^k b|+|\mathfrak{d}^k\ud{b}|+|\mathfrak{d}^k(\f{e^\Phi}{r\sin \theta}-1)|+|\mathfrak{d}^k A|\Big) \lesssim \varepsilon_0,
	\end{equation}
 where $\mathfrak{d}:=\{e_3, re_4, r\nab \}$. 
 
 Note that here $\partial_r$ is not necessarily a null direction, since $g(\partial_r, \partial_r)=\ud{b}^2\gamma$ may be non-zero.

To solve for the apparent horizon, in ${}^{(int)}\mathcal{M}$ we consider a new coordinates system $(\widetilde{\ub}, \widetilde{r}, \widetilde{\theta}, \widetilde{\varphi})$. We fix $$\widetilde{\ub}=\ub, \quad \widetilde{r}=r, \quad \widetilde{\varphi}=\varphi,$$ 
slightly modify $\theta$ and set $$\widetilde{\theta}:=\theta+f_1(\ub, r, \theta, \varphi).$$
By the chain rule we get
\begin{equation*}
\partial_r=\partial_{\widetilde{r}}+\frac{\partial f_1}{\partial r}\partial_{\widetilde{\theta}}, \quad \partial_\theta=(1+\frac{\partial f_1}{\partial \theta})\partial_{\widetilde{\theta}}.
\end{equation*}
This implies that	
\begin{equation*}
\partial_{\widetilde{r}}=\frac{2}{r \kbb}\Big(e_3+(1+\frac{\partial f_1}{\partial \theta})^{-1}e_3(f_1+\theta)\partial_\theta \Big). 
\end{equation*}

To ensure the vector $\partial_{\widetilde{r}}$ to be null, we set $$e_3(f_1+\theta)=0 \quad \text{with} \ 1+\frac{\partial f_1}{\partial \theta}\neq 0.$$ 
In this new setting, we have the below estimate for $f_1$.
\begin{lemma}\label{coordinates transform}
	 Let $f_1$ be the solution to $e_3(f_1+\theta)=0$ with $f_1=0$ along $r=r_{\mathcal{T}}$. Then in ${}^{(int)}\mathcal{M}$, $f_1$ verifies the estimate 
	 \begin{equation*}
	 	|f_1|+|\mathfrak{d} f_1|\lesssim \f{\varepsilon_0}{\ub^{1+\delta_{dec}}}.
	 \end{equation*}
\end{lemma}
\begin{proof}
    Via integrating $e_3(f_1+\theta)=0$ along $e_3$ direction from $r=r_{\mathcal{T}}$, together with the fact that $m_0\le r\le 3m_0$ in ${}^{(int)}\mathcal{M}$, we derive the estimates for $f_1$. When $\mathfrak{d}=e_3$, we have 
    \begin{equation*}
        |e_3(f_1)|=|-e_3(\theta)|=|\underline{b}|\lesssim \f{\varepsilon_0}{\ub^{1+\delta_{dec}}}.
    \end{equation*}
 For $\mathfrak{d}\in \{e_4, \nab \} $, commuting $e_3(f_1+\theta)$ with $\mathfrak{d}$ and noting that $\omegab=0$ in ${}^{(int)}\mathcal{M}$,  we infer that
    \begin{align*}
    &e_3\big(e_4(f_1+\theta)\big)=[e_3, e_4](f_1+\theta)=2(\eta-\etab) \cdot \nab (f_1+\theta),
    \\
   & e_3\big(\nab(f_1+\theta)\big)=[\nab_3, \nab](f_1+\theta)=-\chib \cdot \nab (f_1+\theta).
    \end{align*}
    Therefore, together with Gr\"onwall's inequality, the desired bound of $\mathfrak{d} f_1$ follows from integrating the corresponding equation along $e_3$ direction.
\end{proof}
\begin{remark}
	In ${}^{(int)}\mathcal{M}$ we do not change the null frame employed in Klainerman-Szeftel \cite{K-S}. We only adopt a slightly changed  $(\ub, r, \theta, \varphi)$ coordinate system. The estimates established in \cite{K-S} still hold in this paper.
\end{remark}

\subsection{A Priori Estimates}

For notational simplicity, when there is no danger of confusion, we still write $(\ub, r, \theta, \varphi)$  instead of $(\widetilde{\ub}, \widetilde{r}, \widetilde{\theta}, \widetilde{\varphi})$ to represent our new coordinate system. Along each null hypersurface $\Hb_{\ub}$, at $S_{\ub, R}$ we introduce the elliptic operator
\begin{equation*}
\begin{split}
L(R(\theta, \varphi)):=&\Delta'_M R+\big(f^{-1}\nab f+(\eta+\zeta)\big)\cdot \nab R\\&+(e_3(f)-\f{1}{2}f\trchb-2\omegab f)|\nab R|^2+\f{1}{2}f^{-1}\trch.
\end{split}
\end{equation*}
Here $f=e_3(r)^{-1}=\f{2}{r} \kbb^{-1}<0$, and $\Delta'_M$ denotes the Laplace-Beltrami operator on the 2-dimensional sphere $(\ub, r, \theta, \varphi)=(\ub, R(\ub,\theta, \varphi),\theta, \varphi)$ along fixed $\ub$ hypersurface. 

Recall that in \cite{K-S}, the following estimates in ${}^{(int)}\mathcal{M}$ are proved
\begin{enumerate}
	\item For the Hawking mass $m=\f{r}{2}(1+\f{1}{16\pi}\int_S \kappa \ud{\kappa})$, there holds
	\begin{equation*}
	\f{|m-m_\infty|}{m_0}+|\mathfrak{d}^k (e_3(m), e_4(m))|\lesssim \f{\varepsilon_0}{\ub^{1+\delta_{dec}}}.
	\end{equation*}
	Here $m_0$ is the initial mass and $m_\infty$ is the final Bondi mass satisfying
	\begin{equation*}
	|\f{m}{m_0}-1|+|\f{m_\infty}{m_0}-1|\lesssim \f{\varepsilon_0}{\ub^{1+\delta_{dec}}}.
	\end{equation*}
	\item For $\check{\Gamma}=\{\check{\ud{\kappa}}, \chibh, \zeta, \etab, \check{\kappa}, \chih, \check{\omega}, \xi \}$ and $\check{R}=\{\a,\b, \check{\rho}, \sigma, \bb, \ab
	 \}$, with $0\leq k\leq k_{small}$, they obey
	 \begin{equation*}
	  |\mathfrak{d}^k (\check{\Gamma}, \check{R})|\lesssim \f{\varepsilon_0}{\ub^{1+\delta_{dec}}}.
	 \end{equation*}
	 Note that $\omegab=0$ in ${}^{(int)}\mathcal{M}$.
	 \item  For $\kbb, \overline{\omega}, \kb, \overline{\rho}$ (the averages of $\ud{\kappa}, \kappa, \omega, \rho$), with $0\leq k\leq k_{small}$, there hold
	 \begin{equation*}
	  |\mathfrak{d}^k (\kbb+\f{2}{r}, \overline{\omega}+\f{m_{\infty}}{r^2}, \kb-\f{2}{r}(1-\f{2m_{\infty}}{r}), \overline{\rho}+\f{2m_{\infty}}{r^3})|\lesssim \f{\varepsilon_0}{\ub^{1+\delta_{dec}}}.
	 \end{equation*}
\end{enumerate}
Here we use the convention $\mathfrak{d}=\{e_3, re_4, r\nab \}$.

Since $f=e_3(r)^{-1}=\f{2}{r} \kbb^{-1}$ and $m_0\leq r\leq 3m_0$ in ${}^{(int)}\mathcal{M}$, we obtain
\begin{equation*}
|f+1|=|\kbb^{-1}(\f{2}{r}+\kbb)|\lesssim \f{\varepsilon_0}{\ub^{1+\delta_{dec}}},
\end{equation*}
and
\begin{equation*}
\begin{split}
e_3(f)=&e_3(\f{2}{r} \kbb^{-1})=e_3(\f{2}{r}) \kbb^{-1} -\f{2}{r}\f{e_3(\kbb)}{\kbb^2} 
=e_3(\f{2}{r}) \kbb^{-1}-\f{2}{r}\f{1}{\kbb^2}\Big(e_3(\kbb+\f{2}{r})-e_3(\f{2}{r})\Big) \\
=&-\f{2}{r}\f{1}{\kbb^2}e_3(\kbb+\f{2}{r})+\kbb^{-2} e_3 (\f{2}{r})(\kbb+\f{2}{r}) 
=-\f{2}{r}\f{1}{\kbb^2}e_3(\kbb+\f{2}{r})-\f{1}{r}\f{1}{\kbb} (\kbb+\f{2}{r}).
\end{split}
\end{equation*}
This implies 
\begin{equation*}
|e_3(f)|\lesssim \f{\varepsilon_0}{\ub^{1+\delta_{dec}}}.
\end{equation*}
Differentiating $e_3(f)$ with respect to $e_3$, we further deduce
\begin{align*}
e_3\big(e_3(f)\big)=&-\f12 e_3\l rf^2 e_3(\kbb+\f{2}{r})+f(\kbb+\f{2}{r}) \r\\=&-\f12 \l fe_3(\kbb+\f{2}{r})+2rf e_3(f) e_3(\kbb+\f{2}{r})+rf^2 e_3(e_3(\kbb+\f{2}{r}))+e_3(f)(\kbb+\f{2}{r})+ fe_3(\kbb+\f{2}{r})   \r,
\end{align*}
which provides the bound
\begin{equation*}
|e_3\big(e_3(f)\big)|\lesssim \f{\varepsilon_0}{\ub^{1+\delta_{dec}}}.
\end{equation*}

In below context, we also use $R(\theta, \varphi)$ as a shorthand notation for $R(\ub, \theta, \varphi)$.
Noting $\trch'=2f L(R(\theta, \varphi))$  and the fact that $f<0$ within the region ${}^{(int)}\mathcal{M}$, we have that $\trch'=0$ holds true if and only if $L(R(\theta, \varphi))=0$. This gives the elliptic equation
\begin{equation}\label{MOTS equation}
\Delta'_M R+\big(f^{-1}\nab f+(\eta+\zeta)\big)\cdot \nab R+(e_3(f)-\f{1}{2}f\trchb-2\omegab f)|\nab R|^2+\f{1}{2}f^{-1}\trch
=0.
\end{equation}
The aim of this section is to solve this quasilinear elliptic equation \eqref{MOTS equation} on the sphere $\mathbb{S}^2$. 

With $0\leq \lambda \leq 1$, we first establish a priori estimates for solutions to
\begin{equation}\label{continuity eqn 1}
\begin{split}
    0=&\Delta'_M R(\theta, \varphi)+(e_3(f)-\f{1}{2}f\trchb-2\omegab f)|\nab R(\theta, \varphi)|^2\\&+\lambda\Big((f^{-1}\nab f+(\eta+\zeta)+f^{-1}c)\cdot \nab R+\f{1}{2}f^{-1}\trch\Big)\\&+(1-\lambda)\l-\f{1}{R(\theta, \varphi)}+\f{2m_\infty}{R(\theta, \varphi)^2}\r.
\end{split}
\end{equation}

In ${}^{(int)}\mathcal{M}$, using the fact that $m_0\le R(\theta, \varphi)\le 3m_0$, for any fixed $0\leq \lambda \leq 1$, we can rewrite \eqref{continuity eqn 1} in the form below
\begin{equation}\label{continuity eqn general form}
\begin{split}
0=&\Delta'_M R(\theta, \varphi)-\f{1}{R(\theta, \varphi)}|\nab R(\theta, \varphi)|^2-\f{1}{R(\theta, \varphi)}+\f{2m_{\infty}}{R(\theta, \varphi)^2}\\&+\f{1}{R(\theta, \varphi)^2}\l d_{1a}\nab^a R+d_{2bc}\nab^b R \nab^c R+d_3 \r.
\end{split}
\end{equation}
Here $d_{1a}, d_{2bc}, d_3$ are functions of $(\ub, r, \theta, \varphi)$ obeying
$$\max\limits_{0\leq k\leq 1}|\mathfrak{d}^k(d_{1a}, d_{2bc}, d_3)|\lesssim \f{\varepsilon_0}{\ub^{1+\delta_{dec}}}.$$

\begin{remark}
    For the sake of clarity, throughout this paper, we will utilize the symbol $d, d_{a}, d_{bc}$ to represent functions obeying
\begin{equation*}
    |d|, |d_a|, |d_{bc}| \lesssim  \f{\varepsilon_0}{\ub^{1+\delta_{dec}}}.
\end{equation*}
\end{remark}

We now can start to derive a priori estimates for solutions $R(\theta, \varphi)$ to the equation \eqref{continuity eqn general form}.
\subsubsection{$C^0$ Estimate}\label{Subsubsec: C0 est} A direct application of the maximum (minimum) principle on $\mathbb{S}^2$ yields
\begin{lemma}\label{lemma C0 Est}
    For $\ub \ge 1$, with $R(\theta, \varphi)$ being the solution to \eqref{continuity eqn general form},  the following inequality holds
   \begin{equation}\label{C^0 estimate}
|R(\theta, \varphi)-2m_\infty|\lesssim \f{\varepsilon_0}{\ub^{1+\delta_{dec}}}.
\end{equation}
\end{lemma}
\begin{proof}
    Denote 
    $$R(\theta_1, \varphi_1):=\max\limits_{(\theta, \varphi)\in \mathbb{S}^2} R(\theta, \varphi), \quad R(\theta_2, \varphi_2):=\min\limits_{(\theta, \varphi)\in \mathbb{S}^2} R(\theta, \varphi).$$
    Then at $(\theta_1, \varphi_1)$, we have $\nab' R(\theta_1, \varphi_1)=\nab R(\theta_1, \varphi_1)=0$ and $\Delta'_M R(\theta_1, \varphi_1)\leq 0$. This gives
\begin{equation*}
\begin{split}
0=&\Delta'_M R(\theta_1, \varphi_1)-\f{1}{R(\theta_1, \varphi_1)}+\f{2m_{\infty}}{R(\theta_1, \varphi_1)^2}+\f{d_3}{R(\theta_1, \varphi_1)^2} 
\leq \f{-R(\theta_1, \varphi_1)+2m_{\infty}+d_3}{R(\theta_1, \varphi_1)^2}.
\end{split}
\end{equation*}
Hence we deduce 
\begin{equation*}
R(\theta_1, \varphi_1)\leq 2m_{\infty}+d_3\leq 2m_\infty\Big(1+O(\f{\varepsilon_0}{\ub^{1+\delta_{dec}}})\Big).
\end{equation*}

In the same fashion, we can also prove
\begin{equation*}
R(\theta_2, \varphi_2)\geq  2m_\infty\Big(1-O(\f{\varepsilon_0}{\ub^{1+\delta_{dec}}})\Big).
\end{equation*} 
Combining these two estimates, we thus obtain \eqref{C^0 estimate}.
\end{proof}

\subsubsection{$C^1$ Estimate} \label{Subsubsec: C1 est}
To establish the gradient estimate, we will employ Bochner's formula, which takes the form:
\begin{equation*}
\Delta'_M |\nab' R|^2 = 2|\nab'^2 R|^2 + 2\textrm{Ric}_M (\nab' R, \nab' R) + 2\nab^{'a} R\nab'_a(\Delta'_M R).
\end{equation*}
Here $\nab'$ represents the induced covariant derivative on the deformed 2-sphere $M$. It is important to note that $\nab' R=\nab R$, with $\nab$ being the the induced covariant derivative on $S_{\ub, r}$. In the below, when there is no danger of confusion, we will use   $\nab R$ to substitute $\nab'R$. 

To derive the desired $C^0$ estimate for $\nab R$, we express $\Delta'_M\Big(h(R)|\nab R|^2\Big)$ with 
\begin{equation*}
h(R)=1+\f{1}{2m_\infty ^2}(R-2m_\infty)^2.\footnote{Such type of auxiliary function was introduced by the first auther in \cite{An: AH}.}
\end{equation*} 
Applying the Bochner's formula and plugging in \eqref{continuity eqn general form}, we obtain
\begin{equation}\label{3.12}
\begin{split}
&\Delta'_M\Big(h(R)|\nab R|^2\Big) \\
=&h'(R)\Delta_M' R|\nab R|^2+h''(R)|\nab R|^4+\f{2h'(R)\nab^{'a} R}{h(R)}\cdot \nab'_a\Big(h(R)|\nab R|^2\Big)-\f{2h'(R)h'(R)}{h(R)}|\nab R|^4\\
&+h(R)\l(2|\nab'^2 R|^2+2\textrm{Ric}_M(\nab R, \nab R)+2\nab^{'a} R\nab'_a(\Delta_M' R)\Big)\\
=&h'(R)|\nab R|^2\l\f{1}{R}|\nab R|^2+\f{1}{R}-\f{2m_{\infty}}{R^2}-\f{1}{R^2}\big( d_{1a}\nab^a R+d_{2bc}\nab^b R \nab^c R+d_3 \big) \r+h''(R)|\nab R|^4\\
&+\f{2h'(R)\nab^{'a} R}{h(R)}\cdot \nab'_a\Big(h(R)|\nab R|^2\Big)-\f{2h'(R)h'(R)}{h(R)}|\nab R|^4 
+h(R)\Big(2|\nab'^2 R|^2+2\textrm{Ric}_M(\nab R, \nab R)\Big)\\
&+2h(R)\nab^{'a} R\nab'_a\l\f{1}{R}|\nab R|^2+\f{1}{R}-\f{2m_{\infty}}{R^2}-\f{1}{R^2}\big( d_{1a}\nab^a R+d_{2bc}\nab^b R \nab^c R+d_3 \big) \r.
\end{split}
\end{equation}

In view of the fact
\begin{align*}
\nab'_{a} \l d(R(\theta, \varphi), \theta, \varphi) \r=& \l \nab_a (d)+ f \nab_a R \cdot e_3 (d) \r(R(\theta, \varphi), \theta, \varphi)
=\f{d_{1a} +d_2 \nab_a R}{R},
\end{align*}
we deduce
\begin{equation*}
\begin{split}
&\nab^{'a} R\nab'_a\l\f{1}{R}|\nab R|^2+\f{1}{R}-\f{2m_{\infty}}{R^2}-\f{1}{R^2}\big( d_{1a}\nab^a R+d_{2bc}\nab^b R \nab^c R+d_3 \big) \r\\
=&-\f{1}{R^2}|\nab R|^4+(-\f{1}{R^2}+\f{4m_{\infty}}{R^3})|\nab R|^2+\f{1}{R} \nab^a R \nab'_a (|\nab R|^2) +\f{1}{R^2} \nab^a R \l d_{6b}+d_{7bc} \nab^c R  \r \nab'_a \nab^{'b} R \\
&+\f{1}{R^3}\l d_{1a} |\nab R|^2 \nab^a R+d_{2ab} \nab^a R \nab^b R|\nab R|^2+d_{3ab} \nab^a R \nab^b R+d_4|\nab R|^2+d_{5a}\nab^a R \r.
\end{split}
\end{equation*}
Together with
\begin{equation*}
2h(R)\cdot \nab^{a} R\nab'_a(|\nab R|^2)=2\nab^a R \nab'_a(h(R)|\nab R|^2)- 2h'(R) |\nab R|^4, 
\end{equation*}
we then arrive at the below estimate
\begin{equation}\label{Laplace M ineq main}
\begin{split}
\Delta'_M\Big(h(R)|\nab R|^2\Big) 
\geq&h(R)\Big(2|\nab'^2 R|^2+2\textrm{Ric}_M(\nab R, \nab R)\Big)
+\l \f{2h'(R)}{h(R)} +\f{2}{R} \r\nab'^a R \nab'_a\Big(h(R)|\nab R|^2\Big)\\
&+\Big(\f{h'(R)}{R}+h''(R)-\f{2h'(R)h'(R)}{h(R)}-\f{2h(R)}{R^2}-\f{2h'(R)}{R}+\f{d_1}{R^2}\Big)|\nab R|^4 \\
&+\Big(h'(R)(\f{1}{R}-\f{2m}{R^2}+\f{d_2}{R^2})+2h(R)(\f{d_3}{R^3}-\f{1}{R^2}+\f{4m_{\infty}}{R^3})\Big)|\nab R|^2\\
&
+ \f{h(R)}{R^2} \nab^a R \l d_{5b}+d_{6bc} \nab^c R  \r \nab'_a \nab^{'b} R\\&+d_4\l |\f{h'(R)}{R^2}|\cdot |\nab R|^3+\f{h(R)}{R^3}|\nab R|^3+\f{h(R)}{R^3}|\nab R| \r.
\end{split}
\end{equation}
Using the construction
\begin{equation*}
h(R)=1+\f{1}{2m_\infty ^2}(R-2m_\infty)^2,
\end{equation*}
and the established $C^0$ estimate $$|R-2m_\infty|\lesssim \f{\varepsilon_0}{\ub^{1+\delta_{dec}}} \quad \text{with} \ \varepsilon_0 \ \text{being sufficiently small},$$  we derive
\begin{equation*}
\begin{split}
&|h(R)-1|\lesssim \f{\varepsilon_0}{\ub^{1+\delta_{dec}}},\quad |h'(R)|=\f{1}{m_\infty^2}|R-2m_\infty|\lesssim \f{\varepsilon_0}{\ub^{1+\delta_{dec}}}, \quad h''(R)=\f{1}{m_\infty^2}, \\
&\f{h'(R)}{R}+h''(R)-\f{2h'(R)h'(R)}{h(R)}-\f{2h(R)}{R^2}-\f{2h'(R)}{R}+\f{d_1}{R^2}\geq \f{1}{4m_\infty^2},\\
&h'(R)(\f{1}{R}-\f{2m}{R^2}+\f{d_2}{R^2})+2h(R)(\f{d_3}{R^3}-\f{1}{R^2}+\f{4m_{\infty}}{R^3})\geq \f{1}{R^2}.
\end{split}
\end{equation*}
With these estimates, we go back to \eqref{Laplace M ineq main} and obtain
\begin{equation*}
\begin{split}
&\Delta'_M\Big(h(R)|\nab R|^2\Big)-\f{2h'(R)\nab^{'a} R}{h(R)}\cdot \nab'_a\Big(h(R)|\nab R|^2\Big)-\f{2}{R}\nab'_a(h(R)|\nab R|^2)\\
\geq&h(R)\Big(2|\nab'^2 R|^2+2\textrm{Ric}_M(\nab R, \nab R)+\f{d_1}{R^2} (|\nab R| +|\nab R|^2  ) |\nab'^2 R|\Big)\\
&+\f{1}{4m_\infty^2}|\nab R|^4+\f{1}{R^2}|\nab R|^2+\f{d_2}{R^3} (|\nab R|^3+|\nab R|).
\end{split}
\end{equation*}

To treat $2\textrm{Ric}_M(\nab R, \nab R)$, we have the following lemma.
\begin{lemma} 
For $\ub \ge 1$, with $R(\theta, \varphi)$ being the solution to \eqref{continuity eqn general form}, the following inequality holds 
	\begin{equation}\label{estimate for Ric}
	2\textrm{Ric}_M(\nab R, \nab R)\geq -d(|\nab R|+|\nab R|^2+|\nab^2 R|)|\nab R|^2 \quad \text{with} \quad 0\le d\lesssim \f{\varepsilon_0}{\ub^{1+\delta_{dec}}}.
	\end{equation}
\end{lemma}
\begin{proof}
	For the 2-dimensional manifold $M$, we have $$\text{Ric}_{M}=K' g_M$$ 
	with $K'$ being its Gaussian curvature.    Hence  we get $$2\textrm{Ric}_M(\nab R, \nab R)=2K'|\nab R|^2.$$
	 Using the frame $\{e_1', e_2', e_3', e_4'\}$, by the Gauss equation, we can express $K'$ as
	\begin{equation}\label{Gausss eqn}
	K'=-\rho'+\f{1}{2}\chih'\cdot \chibh'-\f{1}{4}\trch'\trchb'.
	\end{equation}
	Employing \eqref{formula for chi'} and observing that $\trch'=0$ on $M$, we infer
	\begin{equation*}
	\begin{split}
	\chih'=\chi'_{ab}-\trch' g'_{ab} 
	=&\chi_{ab}+\nab_a f \nab_b R+f\nab_a \nab_b R+\nab_b f\nab_a R+f\nab_b \nab_a R\\
	&+e_3(f)f(\nab_a R\nab_b R+\nab_b R\nab_a R)-f^2(\chib_{ac}\nab^c R\nab_b R+\chib_{bc}\nab^c R\nab_a R)\\
	&+f(\eta_b+\zeta_b)\nab_a R+f(\eta_a+\zeta_a)\nab_b R+f^2|\nab R|^2\chi_{ab}-f^2\chi_{ac}\nab_b R \nab^c R\\
	&-f^2\chi_{bc}\nab_a R \nab^c R-4\omegab f^2\nab_a R \nab_b R.
	\end{split}
	\end{equation*}
Recall that in our setting the deformed null frame  $\{e_1', e_2', e_3', e_4'\}$ and $\{e_1, e_2, e_3, e_4\}$ have the following connection 
  \begin{equation*}
     e_3'=e_3, \quad e_a '=e_a+f e_a (R) e_3, \quad e_4'=e_4+2fe^a(R) e_a+f^2|\nab R|^2 e_3.
  \end{equation*}
  Therefore, a direct computation yields 
	\begin{equation*}
	\begin{split}
	\chib'(e_a', e_b')=&g(D_{e'_a}e_3, e_b')=g(D_{e_a}e_3+f e_a (R) e_3, e_b+f e_b (R) e_3) 
	=g(D_{e_a}e_3, e_b)=\chib(e_a, e_b),
	\end{split}
	\end{equation*}
	which gives \begin{equation*}
	\trchb'=\trchb, \quad  \chibh'=\chibh.
	\end{equation*}
	
	As for $\rho'$, by definition we have
	\begin{equation*}
	\begin{split}
	\rho'&=\f{1}{4}R(e_4+2fe^a(R) e_a+f^2|\nab R|^2 e_3, e_3, e_4+2fe^b(R) e_b+f^2|\nab R|^2 e_3, e_3)\\
	&=\rho-2f\bb_b\nab^b R+f^2\ab_{bc}\nab^c R\nab^b R.
	\end{split}
	\end{equation*}
	
 Plugging this into \eqref{Gausss eqn}, together with the estimate for $f$ as well as hyperbolic estimates in \cite{K-S},  we can bound $K'$ from below as
	\begin{equation*}
	\begin{split}
	K'=&-\rho+2f\bb_b\nab^b R-f^2\ab_{bc}\nab^c R\nab^b R+\f{1}{2}\chibh^{ab}\chi_{ab}\\
	&+\f{1}{2}\chibh^{ab}(\nab_a f \nab_b R+f\nab_a \nab_b R+\nab_b f\nab_a R+f\nab_b \nab_a R)\\
	&+\f{1}{2}\chibh^{ab}e_3(f)f(\nab_a R\nab_b R+\nab_b R\nab_a R)-\f{1}{2}\chibh^{ab}f^2(\chib_{ac}\nab^c R\nab_b R+\chib_{bc}\nab^c R\nab_a R)\\
	&+\f{1}{2}\chibh^{ab}f(\eta_b+\zeta_b)\nab_a R+\f{1}{2}\chibh^{ab}f(\eta_a+\zeta_a)\nab_b R+\f{1}{2}\chibh^{ab}f^2|\nab R|^2\chi_{ab}-\f{1}{2}\chibh^{ab}f^2\chi_{ac}\nab_b R \nab^c R\\
	&-\f{1}{2}\chibh^{ab}f^2\chi_{bc}\nab_a R \nab^c R-2\chibh^{ab}\omegab f^2\nab_a R \nab_b R\\
	\geq& \f{2m_{\infty}}{R^3}-d(|\nab R|+|\nab R|^2+|\nab^2 R|+1)
	\geq -d(|\nab R|+|\nab R|^2+|\nab^2 R|)  \quad \text{with} \ 0\le d\lesssim \f{\varepsilon_0}{\ub^{1+\delta_{dec}}}.
	\end{split}
	\end{equation*}
 Thus, we can conclude \eqref{estimate for Ric}.
\end{proof}

We also need to compare $|\nab'^2 R|$ with $|\nab^2 R|$. By the property of covariant derivative, we have
\begin{equation*}
\begin{split}
\nab_a' \nab_b' R=\nab_a' \nab_b R
=&(\nab_a\nab_b+f\nab_a R \nab_3\nab_b)R \\
=&\nab_a\nab_b R-f\chibh_{bc}\nab_a R \nab^c R-\f{1}{2}f\trchb\nab_a R \nab_b R.
\end{split}
\end{equation*}
Using established estimates for $f, \chibh, \trchb$ in \cite{K-S}, this implies
\begin{equation*}
|\nab^{\prime 2 } R|\geq |\nab^{2} R|-\f{1}{2R}|\nab R|^2.
\end{equation*}
Incorporating with \eqref{estimate for Ric}, we hence deduce
\begin{align*}
&2|\nab'^2 R|^2+2\textrm{Ric}_M(\nab R, \nab R)+\f{d_1}{R^2} (|\nab R| +|\nab R|^2  ) |\nab'^2 R|\\
\ge& 2|\nab'^2 R|^2-|d| \l|\nab R|+|\nab R|^2+|\nab'^2 R|+\f{1}{2R}|\nab R|^2\r |\nab R|^2-|d_1| (|\nab R| +|\nab R|^2  ) |\nab'^2 R| \\
\ge& (2-|d|) |\nab'^2 R|^2- |d_2| \l|\nab R|^2+|\nab R|^4  \r.
\end{align*}

Combining all above estimates, with $\varepsilon_0$ being sufficiently small, we then arrive at
\begin{equation}
\begin{split}
&\Delta'_M\Big(h(R)|\nab R|^2\Big)-\f{2h'(R)\nab^{'a} R}{h(R)}\cdot \nab'_a\Big(h(R)|\nab R|^2\Big)-\f{2}{R^2}\nab'_a(h(R)|\nab R|^2) \\
\geq& -|d|(|\nab R|+|\nab R|^2+|\nab R|^3+|\nab R|^4)+\f{1}{4m_\infty^2}|\nab R|^4+\f{1}{R^2}|\nab R|^2 \\
\geq& \f{1}{8m_\infty^2}|\nab R|^4+\f{1}{2R^2}|\nab R|^2-|d|^2,
\end{split}
\end{equation}
where $|d|\lesssim \f{\varepsilon_0}{\ub^{1+\delta_{dec}}}$.

We denote $$[h(R)|\nab R|^2](\theta_3, \varphi_3):=\max\limits_{\theta, \varphi\in \mathbb{S}^2}[h(R)|\nab R|^2](\theta, \varphi).$$ 
Applying the maximum principle, we obtain
\begin{equation*}
0\geq \f{1}{8m_\infty^2}|\nab R|^4(\theta_3, \varphi_3)+\f{1}{2R^2}|\nab R|^2(\theta_3, \varphi_3)-d\geq \f{1}{2R^2}|\nab R|^2(\theta_3, \varphi_3)-|d|^2.
\end{equation*}
By the construction of $h$, we further deduce
\begin{equation*}
h(R)|\nab R|^2\leq h\big(R(\theta_3, \varphi_3)\big)|\nab R|^2(\theta_3, \varphi_3)\leq 4|d|^2 R^2.
\end{equation*}
Noting that $1/2\leq h(R)\leq 2$, we thus derive 
\begin{equation*}
|\nab R|^2 \leq 8|d|^2 R^2 \lesssim (\f{\varepsilon_0}{\ub^{1+\delta_{dec}}})^2.
\end{equation*}
This gives the desired $C^1$ estimate
\begin{equation}\label{C^1 estimate}
|\nab R|(\theta, \varphi) \lesssim \f{\varepsilon_0}{\ub^{1+\delta_{dec}}} \quad \text{for all} \  (\theta, \varphi)\in \mathbb{S}^2.
\end{equation}

\subsubsection{$C^{1,q}$ Estimate and $C^{2,q}$ Estimate} \label{Subsubsec: C 1 q est}
 Examining the proof of Theorem 13.6 in \cite{G-T}, one can see that obtaining $C^{1, q}$ estimate from the gradient estimate for $R$ follows directly. We have
\begin{proposition}\label{C^{1,q} estimate}
    Assuming that $R$ is the solution to \eqref{continuity eqn general form}, then there exists $0<q<1$ independent of $\varepsilon_0$, $\ub$, such that 
    \begin{equation*}
        \|\nab R\|_{C^{0, q}(M_{\ub})} \lesssim  \f{\varepsilon_0}{\ub^{1+\delta_{dec}}}.
    \end{equation*}
\end{proposition}
With a priori $C^{1, q}$ estimates for $R$, through the direct application of the standard interior Schauder estimate, we obtain the $C^{2, q}$ estimates.
\begin{lemma}\label{C^{2,q} estimate}
    Assuming that $R$ is the solution to \eqref{continuity eqn general form}, with the same $q\in (0, 1)$ as in \Cref{C^{1,q} estimate}, the following inequality holds
    \begin{equation*}
        \| R-2m_{\infty}\|_{C^{2, q}(M_{\ub})} \lesssim  \f{\varepsilon_0}{\ub^{1+\delta_{dec}}}.
    \end{equation*}
\end{lemma}

\subsection{Continuity Argument}\label{Continuity1}
With above derived a priori estimates, we conduct the method of continuity to solve for the MOTS. We will solve the below quasilinear elliptic equation
\begin{equation*}
\Delta'_M R+\big(f^{-1}\nab f+(\eta+\zeta)\big)\cdot \nab R+(e_3(f)-\f{1}{2}f\trchb-2\omegab f)|\nab R|^2+\f{1}{2}f^{-1}\trch
=0.
\end{equation*}
Based on this equation, we insert a parameter $\lambda \in [0, 1]$ and consider
\begin{equation*}
\begin{split}
      F(R(\theta, \varphi), \lambda):&=\Delta'_M R(\theta, \varphi)+(e_3(f)-\f{1}{2}f\trchb-2\omegab f)|\nab R(\theta, \varphi)|^2\\&+\lambda\Big(\big(f^{-1}\nab f+(\eta+\zeta)\big)\cdot \nab R+\f{1}{2}f^{-1}\trch\Big)\\&+(1-\lambda)\l-\f{1}{R(\theta, \varphi)}+\f{2m_\infty}{R(\theta, \varphi)^2}\r.
\end{split}
\end{equation*}
For $\lambda=0$, we verify that $R(\theta, \varphi)\equiv 2m_\infty$ is a solution to
\begin{equation*}
 F(R(\theta, \varphi), 0)=\Delta'_M R(\theta, \varphi)+(e_3(f)-\f{1}{2}f\trchb-2\omegab f)|\nab R(\theta, \varphi)|^2-\f{1}{R(\theta, \varphi)}+\f{2m_\infty}{R(\theta, \varphi)^2}
=0.
\end{equation*}
For $0\le \tilde{\lambda}\le 1$, we assume that $\tilde{R}(\theta, \varphi)$ is a solution to
\begin{align*}
0=&\Delta'_{S_{\ub, \tilde{R}}}\tilde{R}(\theta, \varphi)+(e_3(f)-\f{1}{2}f\trchb-2\omegab f)|\nab \tilde{R}(\theta, \varphi)|^2\\&+\tilde{\lambda}\Big(\big(f^{-1}\nab f+(\eta+\zeta)\big)\cdot \nab \tR+\f{1}{2}f^{-1}\trch\Big)\\&+(1-\tilde{\lambda})\l-\f{1}{\tilde{R}(\theta, \varphi)}+\f{2m_\infty}{\tilde{R}(\theta, \varphi)^2}\r.
\end{align*}
Given any $\lambda \in [0, 1]$ (which may be different from $\tilde{\lambda}$), we are interested in the Fréchet derivative of $F$ with respect to $\tR$ at $(\tR(\theta, \varphi), \tilde{\lambda})$. We have
\begin{equation}\label{banach derivative for F} 
\begin{split}
\partial_R F(\tilde{R}(\theta, \varphi), \lambda)[W]
=&\lim\limits_{\varepsilon\to 0}\f{1}{\varepsilon}\Big(F(\tilde{R}+\varepsilon W, \lambda)-F(\tilde{R}, \lambda)\Big)\\
=&\lim\limits_{\varepsilon\to 0}\f{1}{\varepsilon}\Big(\Delta'_{S_{\ub, \tilde{R}+\varepsilon W}} (\tilde{R}+\varepsilon W)-\Delta'_{S_{\ub, \tilde{R}}} \tilde{R}\Big)\\
&+\lim\limits_{\varepsilon\to 0}\f{1}{\varepsilon}\Big[\l(e_3(f)-\f{1}{2}f\trchb-2\omegab f)|\nab (\tilde{R}+\varepsilon W)|^2\r|_{\ub, \tilde{R}+\varepsilon W}\\&-\l(e_3(f)-\f{1}{2}f\trchb-2\omegab f)|\nab \tilde{R}|^2\r|_{\ub, \tilde{R}} \Big]\\
&+\lambda\lim\limits_{\varepsilon\to 0}\f{1}{\varepsilon}\Big[ \l\big(f^{-1}\nab f+(\eta+\zeta)\big)\cdot \nab (\tilde{R}+\varepsilon W)+\f{1}{2}f^{-1}\trch\r|_{\ub, \tilde{R}+\varepsilon W}\\&-\l\big(f^{-1}\nab f+(\eta+\zeta)\big)\cdot \nab \tilde{R}+\f{1}{2}f^{-1}\trch\r|_{\ub, \tilde{R}}\Big]\\
&+(1-\lambda)\lim\limits_{\varepsilon\to 0}\f{1}{\varepsilon}\Big(-\f{1}{\tilde{R}+\varepsilon W}
+\f{1}{\tilde{R}}+\f{2m_\infty}{(\tilde{R}+\varepsilon W)^2}-\f{2m_\infty}{\tilde{R}^2}\Big)\\
=:&I_1+I_2+\lambda I_3+(1-\lambda) I_4.
\end{split}
\end{equation}

Utilizing \Cref{partial R L} and its proof, we now establish the invertibility of $\partial_R F(\tR,\lambda)$.
\begin{proposition}\label{FR invert}
    Let $\tilde{\lambda}\in [0, 1]$ and assume that $\tR$ is the solution to $F(\tR, \tilde{\lambda})=0$. Then the operator $\partial_R F(\tR,\lambda)[W]: C^{2, q}(\mathbb{S}^2)\rightarrow C^{0,q}(\mathbb{S}^2)$ is invertible for all $\lambda\in [0, 1]$. Here $q\in (0, 1)$ is the same H\"older exponent as in the a priori estimates for  $\tR$.
\end{proposition}
\begin{proof}
Applying \Cref{partial R L}, together with the established hyperbolic estimates in \cite{K-S}, $C^1$ estimate \eqref{C^1 estimate} and $C^{2,q}$ estimate in \Cref{C^{2,q} estimate}, we deduce that
\begin{equation*}
\begin{split}
I_1+I_2=&\Delta'_{S_{\ub, \tR}} W+d_{1a} \nab^a W+d_2 W, \\
    I_3=&\partial_r \l\big(f^{-1}\nab f+(\eta+\zeta)\big)\cdot \nab \tilde{R}+\f{1}{2}f^{-1}\trch\r W+\big(f^{-1}\nab f+(\eta+\zeta)\big)\cdot \nab W\\
    =&f \nab_3 \l\big(f^{-1}\nab f+(\eta+\zeta)\big)\cdot \nab \tilde{R}+\f{1}{2}f^{-1}\trch\r W+\big(f^{-1}\nab f+(\eta+\zeta)\big)\cdot \nab W \\
    =&(\f{1}{\tR^2}-\f{4m_{\infty}}{\tR^3}+d_3)W+d_{4a} \nab^a W, \\
    I_4=&\l\f{1}{\tilde{R}^2}-\f{4m_\infty}{\tilde{R}^3}\r W .
\end{split}
\end{equation*}
Here $d_{1a}, d_2, d_3, d_{4a}$ are functions obeying
\begin{equation*}
    |(d_{1a}, d_2, d_3, d_{4a})|\lesssim\f{\varepsilon_0}{\ub^{1+\delta_{dec}}}.
\end{equation*}
We also use the estimate
\begin{equation*}
    e_3(\trch)|_{S_{\ub, \tR}}=e_3\big(\f{2}{r}(1-\f{2m_{\infty}}{r})\big)_{S_{\ub, \tR}}+O(\f{\varepsilon_0}{\ub^{1+\delta_{dec}}})=(-\f{2}{\tR^2}+\f{8m_{\infty}}{\tR^3})+O(\f{\varepsilon_0}{\ub^{1+\delta_{dec}}}).
\end{equation*}
\vspace{1mm}

Summarizing these expressions for $I_i$ ($i=1, 2, 3, 4$), we arrive at
\begin{equation*}
\begin{split}
&\partial_R F(\tilde{R}(\theta, \varphi), \lambda)[W]=I_1+I_2+\lambda I_3+(1-\lambda) I_4
=\Delta'_{S_{\ub, \tR}} W+d_{1a}\nab^a W+\Big[\f{1}{\tilde{R}^2}-\f{4m_\infty}{\tilde{R}^3}+d_2 \Big] W.
\end{split}
\end{equation*}
Notice that in this expression, the coefficient in front of $W$ in $\partial_R F(\tilde{R}(\theta, \varphi), \lambda)[W]$  reads
\begin{equation*}
\begin{split}
\f{1}{\tilde{R}^2}-\f{4m_\infty}{\tilde{R}^3}+d_2 =\tR^{-3}\Big[-2m_{\infty} +d_3\Big]
\end{split}
\end{equation*}
and it is negative due to the fact that
\begin{equation*}
    |d_3|\lesssim  \f{\varepsilon_0}{\ub^{1+\delta_{dec}}}\le \varepsilon_0 \ll 2m_{\infty}.
\end{equation*}
 In view of the solvability condition for elliptic equations, since $W$ has a negative coefficient, we then show that  the elliptic operator $\partial_R F(\tR(\theta, \varphi),\lambda)[W]: C^{2, q}(\mathbb{S}^2)\rightarrow C^{0,q}(\mathbb{S}^2)$ is invertible for $W$.
\end{proof}

By the implicit function theorem, the invertibility of $\partial_R F(\tR(\theta, \varphi),\lambda)$ guarantees that with $\tilde{\lambda}$ replaced by a nearby $\lambda$, we still have the solution $R(\theta, \varphi)$ to $F(R(\theta, \varphi),\lambda)=0$. Repeating this process, from solving $F(R(\theta, \varphi), 0)=0$, we can eventually solve  $F(R(\theta, \varphi), 1)=0$ and find the solution verifying the equation of MOTS, i.e.,
$$\Delta'_M R+\big(f^{-1}\nab f+(\eta+\zeta)\big)\cdot \nab R+(e_3(f)-\f{1}{2}f\trchb-2\omegab f)|\nab R|^2+\f{1}{2}f^{-1}\trch
=0.$$

\subsection{Uniqueness of MOTS}\label{Subsec: uniqueness of MOTS}
In the previous subsection, we find a solution (being a MOTS) to \eqref{MOTS equation} along each $\Hb_{\ub}$. Here we demonstrate that the MOTS $M_{\ub}$ just constructed is unique within $\Hb_{\ub}$. To demonstrate so, we remark that from \Cref{change laplacian}, the equation of MOTS \eqref{MOTS equation} is equivalent to
\begin{equation}\label{MOTS equation 2}
\begin{split}
0=&\Delta_{S_{\ub, R}} R+\big(f^{-1}\nab f+(\eta+\zeta)\big)\cdot \nab R+(e_3(f)-\f{1}{2}f\trchb-2\omegab f)|\nab R|^2\\&-2(f\chibh_{bc})\nab^b R \nab^c R+\f{1}{2}f^{-1}\trch.
\end{split}
\end{equation}
\begin{proposition}\label{uniqueness of MOTS}
    For a fixed $\ub \ge 1$, assume that both $R(\theta, \varphi)$ and $\tR(\theta, \varphi)$ are solutions to \eqref{MOTS equation} along $\Hb_{\ub}$. Then we must have $R(\theta, \varphi)\equiv\tR(\theta, \varphi)$.
\end{proposition}
\begin{proof}
Employing \eqref{MOTS equation 2}, we have that
\begin{equation*}
\begin{split}
0=&\Delta_{S_R} R(\theta, \varphi)+\big(f^{-1}\nab f+(\eta+\zeta)\big)(R(\theta, \varphi),\theta, \varphi)\cdot \nab R(\theta, \varphi)\\&+(e_3(f)-\f{1}{2}f\trchb-2\omegab f)(R(\theta, \varphi), \theta, \varphi)|\nab R(\theta, \varphi)|^2\\&-2(f\chibh_{bc})(R(\theta, \varphi), \theta, \varphi)\nab^b R(\theta, \varphi) \nab^c R(\theta, \varphi)
+\f{1}{2}(f^{-1}\trch)(R(\theta, \varphi), \theta, \varphi)
\end{split}
\end{equation*}
and
\begin{equation*}
\begin{split}
0=&\Delta_{S_{\tR}} \tR(\theta, \varphi)+\big(f^{-1}\nab f+(\eta+\zeta)\big)(\tR(\theta, \varphi),\theta, \varphi)\cdot \nab \tR(\theta, \varphi)\\&+(e_3(f)-\f{1}{2}f\trchb-2\omegab f)(\tR(\theta, \varphi), \theta, \varphi)|\nab \tR(\theta, \varphi)|^2\\
&-2(f\chibh_{bc})(\tR(\theta, \varphi), \theta, \varphi)\nab^b \tR(\theta, \varphi) \nab^c \tR(\theta, \varphi)+\f{1}{2}(f^{-1}\trch)(\tR(\theta, \varphi), \theta, \varphi).
\end{split}
\end{equation*}
Note that in \eqref{MOTS equation 2} we utilize $\Delta_{S_{\ub, R}}$ instead of $\Delta'_M$.

We then consider
\begin{equation}\label{3.14}
\begin{split}
\Delta_{S_{\tR}}(\tR(\theta, \varphi)-R(\theta, \varphi))
=&-(\Delta_{S_{\tR}}-\Delta_{S_R})R(\theta, \varphi)+\l\Delta_{S_{\tR}} \tR(\theta, \varphi)- \Delta_{S_R} R(\theta, \varphi)\r\\
=&:I+II.
\end{split}
\end{equation}
Back to \eqref{3.13}, for a scalar function $f$, we have
	\begin{equation*}
	\begin{split}
	\D_{S_{\ub, r}} f=&\f{1}{\sqrt{\det g}}\f{\partial}{\partial \theta_i} (\sqrt{\det g}\,g^{\theta_i \theta_l}\f{\partial f}{\partial \theta_l})\\
	=&g^{\theta_1 \theta_1}\f{\partial^2 f}{\partial \theta_1 \partial \theta_1}+g^{\theta_2 \theta_2}\f{\partial^2 f}{\partial \theta_2 \partial \theta_2}+2 g^{\theta_1 \theta_2}\f{\partial^2 f}{\partial \theta_1 \partial \theta_2}+\f{\partial}{\partial \theta_1} (g^{\theta_1 \theta_1})\f{\partial f}{\partial \theta_1}+\f{\partial}{\partial \theta_1} (g^{\theta_1 \theta_2})\f{\partial f}{\partial \theta_2}\\
	&+\f{\partial}{\partial \theta_2} (g^{\theta_2 \theta_1})\f{\partial f}{\partial \theta_1}+\f{\partial}{\partial \theta_2} (g^{\theta_2 \theta_2})\f{\partial f}{\partial \theta_2}+\f12 g^{\theta_k \theta_j}\f{\partial g_{\theta_j \theta_k}}{\partial \theta_i} g^{\theta_i \theta_l}\f{\partial f}{\partial \theta_l},
	\end{split}
	\end{equation*}
	where $i,j,k,l=1,2$ and $g^{\theta_k \theta_j}$ depends on $(\ub, r, \theta, \varphi)$.

With the assistance of a priori estimates and estimates for metric coefficients from \eqref{estimates for metric}, we deduce that
\begin{equation*}
I=d_1(\tR-R),
\end{equation*} 
\begin{equation*}
\begin{split}
II=&\f{1}{\tR}|\nab \tR|^2+\f{1}{\tR}-\f{2m(\tR, \theta, \varphi)}{\tR^2}+d_{1a}(\tR, \theta, \varphi)\nab^a \tR+d_{2bc}(\tR, \theta, \varphi)\nab^b \tR \nab^c \tR+d_3(\tR, \theta, \varphi)\\
&-\Big(\f{1}{R}|\nab R|^2+\f{1}{R}-\f{2m(R, \theta, \varphi)}{R^2}+d_{1a}(R, \theta, \varphi)\nab^a R+d_{2bc}(R, \theta, \varphi)\nab^b R \nab^c \tR+d_3(R, \theta, \varphi)\Big) \\
=&(-\f{1}{\tR R}+\f{2m(\tR+R)}{\tR^2 R^2}+d_1)(\tR-R)+d^i_2\f{\partial}{\partial \theta_i} (\tR-R)\\
=&(\f{1}{4m_\infty^2}+d_2)(\tR-R)+d^i_3\f{\partial}{\partial \theta_i} (\tR-R)
\end{split}
\end{equation*}
with  $d_1, d_2,  d^i_3$ being functions of $(\ub, r, \theta, \varphi)$ that satisfy the bound
$$|(d_{1}, d_{2}, d^i_3)|\lesssim \f{\varepsilon_0}{\ub^{1+\delta_{dec}}}.$$
Hence, we can reformulate \eqref{3.14} as
\begin{equation*}
\Delta_{S_{\tR}}(\tR-R)(\theta, \varphi)-d^i_3\f{\partial}{\partial \theta_i} (\tR-R)(\theta, \varphi)-(\f{1}{4m_\infty^2}+d_1+d_2)(\tR-R)(\theta, \varphi)=0.
\end{equation*}
Note that the coefficient in front of $(\tR-R)(\theta, \varphi)$ is negative. Applying the maximum principle, we then conclude
\begin{equation*}
\tR(\theta, \varphi)=R(\theta, \varphi) \qquad \text{for any} \  (\theta, \varphi)\in \mathbb{S}^2.
\end{equation*}
\end{proof}

\subsection{Regularity of Apparent Horizon}\label{Subsec: regularity of AH}
With various values of $\ub$, the collection of these MOTSs $\{ M_{\ub} \}$ forms an apparent horizon. Denote $\mathcal{AH}:=\cup_{\ub \ge 1} M_{\ub}$. In this subsection, we study the regularity of $\mathcal{AH}$. We have
\begin{proposition}
    The apparent horizon $\mathcal{AH}:=\cup_{\ub \ge 1} M_{\ub}$ constructed in our setting is a smooth three-dimensional hypersurface.
\end{proposition}
\begin{proof}
    For $\mathcal{AH}=\{r=R(\ub, \theta, \varphi) \}$, with $\ub \ge 1 $, it suffices to show that $R(\ub, \theta, \varphi)$ is a smooth function of both $\ub \ge 1$ and $(\theta, \varphi)\in \mathbb{S}^2$. Recall that $R(\ub, \theta, \varphi)$ solves the equation 
    \begin{equation*}
    \begin{split}
        0=L(R, \ub):=&\Delta_{S_{\ub, R}} R+\big(f^{-1}\nab f+(\eta+\zeta)\big)\cdot \nab R+(e_3(f)-\f{1}{2}f\trchb-2\omegab f)|\nab R|^2\\&-2(f\chibh_{bc})\nab^b R \nab^c R+\f{1}{2}f^{-1}\trch.
\end{split}
    \end{equation*}
   In \Cref{Continuity1}, we prove that the operator $L(\ub, R): [1, \infty)\times C^{\infty}(\mathbb{S}^2) \to C^{\infty}(\mathbb{S}^2)$ smoothly depends on $\ub$, $R$, and its Fréchet derivative with respect to $R$, i.e.,
    \begin{equation*}
        \partial_R L (R, \ub)=\partial_R F (\ub, R, \lambda=1): C^{\infty}(\mathbb{S}^2) \to C^{\infty}(\mathbb{S}^2)
    \end{equation*}
    is invertible. This is verified in \Cref{FR invert}. Employing the implicit function theorem, we thus show that $R(\ub, \theta, \varphi)$ is smooth of $\ub$ and $(\theta, \varphi)\in \mathbb{S}^2$. Hence, our constructed $\mathcal{AH}$ is a smooth hypersurface.
\end{proof}

\subsection{Final State of Apparent Horizon}\label{Subsec: final state of AH}
In this subsection, we proceed to establish 
\begin{theorem}\label{Theorem AH asymp null converge to EH}
    The apparent horizon $\mathcal{AH}$ solved in our setting is an asymptotically null hypersurface, and would converge to the event horizon eventually.
\end{theorem}
To prove this theorem, the crucial step is to derive the below estimate for $\partial_{\ub} R$.
\begin{lemma}\label{partial ub R est}
   Let $\mathcal{AH}=\cup_{\ub \ge 1} R(\ub, \theta, \varphi)=\{r=R(\ub, \theta, \varphi): \ub \ge 1  \}$ be the apparent horizon solved in our setting. Then for any $\ub \ge 1$, we have
    \begin{equation*}
        |\partial_{\ub} R(\ub, \theta, \varphi)|\lesssim \f{\varepsilon_0}{\ub^{1+\delta_{dec}}}.
    \end{equation*}
\end{lemma}
\begin{proof}
Taking $\partial_{\ub}$ on both sides of \eqref{MOTS equation 2}, we get
    \begin{equation}\label{partial R L chain rules}
    \partial_R L(R, \ub)[\partial_{\ub} R]+\partial_{\ub} L(R, \ub)=0.
\end{equation}
Following the proof of \Cref{FR invert}, by setting $(\tR, \tilde{\lambda})=(R, 1)$, we derive that
\begin{equation*}
    \partial_R L(R, \ub)[W]=\Delta'_{S_{\ub, R}} W+d_{1a}\nab^a W+\Big[-\f{1}{4m_{\infty}^2}+d_2 \Big] W,
\end{equation*}
where $d_{1a}, d_2$ satisfy
\begin{equation*}
    |(d_{1a}, d_2)|\lesssim\f{\varepsilon_0}{\ub^{1+\delta_{dec}}}.
\end{equation*}
Substituting it into \eqref{partial R L chain rules} gives
\begin{equation*}
   \Delta'_{S_{\ub, R}} (\partial_{\ub} R)+d_{1a}\nab^a (\partial_{\ub} R)+\Big[-\f{1}{4m_{\infty}^2}+d_2 \Big] (\partial_{\ub} R)=-\partial_{\ub} L(R, \ub).
\end{equation*}
Applying maximum principle for $\partial_{\ub} R$ and $-\partial_{\ub} R$ respectively, we deduce
\begin{equation}\label{partial ub R ineq}
     \max\limits_{\mathbb{S}^2} |\partial_{\ub} R|\lesssim \max\limits_{\mathbb{S}^2}| \partial_{\ub} L(R, \ub)|.
\end{equation}
Next we turn to estimate $\partial_{\ub} L(R, \ub)$. Notice that
$$\partial_{\ub} =\f{1}{2}\vsgmb\Big[e_4-\f{\kb+A}{\kbb}e_3-\sqrt{\gamma}(b-\f{\kb+A}{\kbb}\ud{b})e_\theta\Big].$$
By \Cref{partial ub L} (with $(R_1, R_2)=(m_0, 3m_0)$), together with hyperbolic estimates from \cite{K-S}, we obtain the bound
\begin{equation}\label{pr ub L ineq}
         |\partial_{\ub} L(R, \ub)-\f{1}{2}\partial_{\ub}(f^{-1}\trch)|\lesssim|\nab^2 R|+|\nab R|+|\nab R|^2.
     \end{equation}
To estimate  $\partial_{\ub}(f^{-1}\trch)$, we use the null structure equation for $\trch$:\footnote{See Proposition 7.4.1 in \cite{Chr-Kl} by Christodoulou-Klainerman.}
\begin{equation*}
    e_4 (\trch)+\frac 12 (\trch)^2=2\div \xi -2\omega \trch+2\xi \cdot (\eta+\etab+2\zeta)-|\chih|^2.
\end{equation*}
Utilizing hyperbolic estimates from \cite{K-S} again and the fact that $\trch|_{S_{\ub, R}}=O(\varepsilon_0 \ub^{-1-\delta_{dec}})$, it  follows
\begin{equation*}
    |\partial_{\ub}(f^{-1}\trch)|\lesssim |e_4 (\trch) |+(|\kb|+|A|)|e_3(\trch)|+|\nab\trch|+|\trch|\lesssim \f{\varepsilon_0}{\ub^{1+\delta_{dec}}} \quad \text{on} \quad M_{\ub}.
\end{equation*}
Inserting the inequality above as well as $C^1$ estimate \eqref{C^1 estimate}, $C^{2,q}$ estimate in \Cref{C^{2,q} estimate} into \eqref{pr ub L ineq}, we hence deduce
\begin{equation*}
    | \partial_{\ub} L(R, \ub)|\lesssim \f{\varepsilon_0}{\ub^{1+\delta_{dec}}}.
\end{equation*}
Back to \eqref{partial ub R ineq}, we then arrive at
\begin{equation*}
        |\partial_{\ub} R(\ub, \theta, \varphi)|\lesssim \f{\varepsilon_0}{\ub^{1+\delta_{dec}}}.
    \end{equation*}
\end{proof}
To proceed, we note that the components of the induced metric on $\mathcal{AH}$ take the form below
\begin{equation*}
	g'_{\theta_i \theta_j}=g_{\theta_i \theta_j}+\partial_{\theta_i} R\cdot\partial_{\theta_j} R\cdot g(\partial_r, \partial_r)=g_{\theta_i \theta_j}, 
\end{equation*}
\begin{equation*}
	g'_{\ub\, \ub}=g_{\ub\, \ub}+2\partial_{\ub} R\cdot g(\partial_r, \partial_{\ub})=\vsgmb^2\l\f{\kb+A}{\kbb}+\f{1}{4}\gamma b^2\r-2\partial_{\ub} R\cdot\f{2\vsgmb}{r\kbb},
\end{equation*}
\begin{equation*}
	g'_{\theta \ub}=g_{\theta \ub}+\partial_{\theta} R \cdot g(\partial_r, \partial_{\ub})=-\f{1}{2}\vsgmb \gamma b-\partial_{\theta} R\cdot\f{2\vsgmb}{r\kbb},
\end{equation*}
\begin{equation*}
	g'_{\varphi \ub}=g_{\varphi \ub}+\partial_{\varphi} R \cdot g(\f{\partial}{\partial r}, \partial_{\ub})=-\partial_{\varphi} R\cdot\f{2\vsgmb}{r\kbb}.
\end{equation*}
Here we use \eqref{coordinates derivatives}, namely,
\begin{equation*}
	\begin{split}
	&\partial_r=\f{2}{r\kbb}e_3-\f{2\sqrt{\gamma}}{r\kbb}\ud{b}e_\theta,\quad \partial_\theta=\sqrt{\gamma}e_\theta, \\
	&\partial_{\ub} =\f{1}{2}\vsgmb\Big[e_4-\f{\kb+A}{\kbb}e_3-\sqrt{\gamma}(b-\f{\kb+A}{\kbb}\ud{b})e_\theta\Big]
	\end{split}
	\end{equation*}
and the fact that $\ud{b}=0$ in our new $(\ub, r, \theta, \varphi)$ coordinates.

Combining $C^1$ estimate \eqref{C^1 estimate}, \Cref{partial ub R est}, hyperbolic estimates and metric estimates \eqref{estimates for metric} in \cite{K-S}, we arrive at
\begin{equation*}
	|(g'_{\ub\, \ub}, g'_{\theta \ub}, g'_{\varphi \ub})|\lesssim \f{\varepsilon_0}{\ub^{1+\delta_{dec}}},
\end{equation*}
which implies that
\begin{equation*}
	\lim\limits_{\ub\to +\infty}g'_{\ub\, \ub}=\lim\limits_{\ub\to +\infty} g'_{\theta \ub}=\lim\limits_{\ub\to +\infty}g'_{\varphi \ub}=0.
\end{equation*}

Hence, $\partial'_{\ub}$ is an asymptotically degenerate direction for our apparent horizon $\mathcal{AH}$. In other words, the apparent horizon we construct is asymptotically null.
\vspace{4mm}

The second conclusion of \Cref{Theorem AH asymp null converge to EH} is a direct consequence of our $C^0$ estimate 
\begin{equation*}
    |R(\ub, \theta, \varphi)-2m_{\infty}|\lesssim \f{\varepsilon_0}{\ub^{1+\delta_{dec}}}
\end{equation*}
established in \Cref{lemma C0 Est} and the asymptotic of the future event horizon as presented in \cite{K-S}, i.e., the event horizon $\mathcal{H}_{+}$ of $\mathcal{M}$ is located in the following region of ${}^{(int)}\mathcal{M}$:
\begin{equation*}
    2m_{\infty}(1-\f{\sqrt{\varepsilon_0}}{\ub^{1+\delta_{dec}}}) \le r \le 2m_{\infty}(1+\f{\sqrt{\varepsilon_0}}{\ub^{1+\delta_{dec}}}) \quad \text{for any} \quad \ub \ge 1.
\end{equation*}
Denote $\mathcal{H}_{+}:=\{r=R_{\mathcal{H}_{+}}(\ub, \theta, \varphi): \ub\ge 1 \}$. We obtain
\begin{equation*}
    |R(\ub, \theta, \varphi)-R_{\mathcal{H}_{+}}(\ub, \theta, \varphi)|\lesssim \f{\sqrt{\varepsilon_0}}{\ub^{1+\delta_{dec}}} \quad \text{for all} \quad (\theta, \varphi)\in \mathbb{S}^2,
\end{equation*}
which implies
\begin{equation*}
    \lim\limits_{\ub\to +\infty} \l R(\ub, \theta, \varphi)-R_{\mathcal{H}_{+}}(\ub, \theta, \varphi) \r=0.
\end{equation*}
Therefore, the apparent horizon $\mathcal{AH}$ constructed in our setting would converge to the event horizon $\mathcal{H}_{+}$ as $\ub$ approaches the infinity.

\section{Comparison Principle along the incoming null hypersurface}\label{Section: NCP}
In this section, we will study the (local) achronality of the apparent horizon. In this paper, our spacetime $(\mathcal{M}, \mathbf{g})$ is foliated by incoming null hypersurfaces $\Hb_{\ub}$ with $\ub$ being the optical function satisfying 
\begin{equation*}
    \mathbf{g}^{\mu \nu}\partial_{\mu} \ub \partial_{\nu} \ub=0.
\end{equation*}
Along $\Hb_{\ub}$, we choose a coordinate system $(r, \theta_1, \theta_2)$ so that 
    \begin{equation*}
    \partial_r=-f \mathbf{g}^{\mu \nu}\partial_{\mu} \ub \partial_{\nu} \quad \textrm{with} \quad f<0 \quad \text{being a fixed function}.
\end{equation*}
 
 Our apparent horizon $\mathcal{AH}$ can be decomposed as $\mathcal{AH}=\cup_{\ub \in J} M_{\ub}$ with $M_{\ub}$ lying in $\Hb_{\ub}$. Here $J$ is a connected interval. Assuming that in $(\ub, r, \theta_1, \theta_2) $ coordinates  $M_{\ub}$ is characterised as $\{ r=R(\ub, \theta_1, \theta_2)\}$ along $\Hb_{\ub}$, we then introduce the null comparison principle as below.
\begin{definition}\label{comparison principle}
    For a given $\ub\in J$, we say that the null comparison principle with respect to the MOTS $M_{\ub}=\{ r=R(\ub, \theta_1, \theta_2)\}$ holds if the following is true: for any smooth 2-surface $\widetilde{\Sigma}=\{r=\tilde{R}(\theta_1, \theta_2) \}$ lying in $\Hb_{\ub}$ close to $M_{\ub}$, if its outgoing null expansion satisfies $\tr \tilde{\chi}\le 0$, then we have $\tilde{R}(\theta_1, \theta_2) \leq R(\theta_1, \theta_2)$ and if $\tr \tilde{\chi}\ge 0$, then it yields $\tilde{R}(\theta_1, \theta_2) \ge R(\theta_1, \theta_2)$.
\end{definition}
Leveraging the strong maximum principle for elliptic equations arising from the definition of outgoing null expansion, we can exclude situations  that $\tilde{R} = R$ only occurs at some but not all points $(\theta_1, \theta_2)\in \mathbb{S}^2$.
\begin{proposition}\label{strong comparison theorem}
	Assume that the null comparison principle holds for the MOTS $M_{\ub}=\{ r=R(\ub, \theta_1, \theta_2)\}$. Let $\widetilde{\Sigma}=\{r=\tilde{R}(\theta_1, \theta_2) \}$ be a smooth $2$-surface lying on $\Hb_{\ub}$. Then if its outgoing null expansion $\tr \tilde{\chi}|_{\widetilde{\Sigma}} \leq 0$, we have either $\tR(\theta_1, \theta_2)\equiv R(\theta_1, \theta_2)$ or $\tR(\theta_1, \theta_2)<  R(\theta_1, \theta_2)$, and if $\tr \tilde{\chi}|_{\widetilde{\Sigma}} \geq 0$, it holds either $\tR(\theta_1, \theta_2)\equiv R(\theta_1, \theta_2)$ or $\tR(\theta_1, \theta_2)>  R(\theta_1, \theta_2)$.
\end{proposition}
\begin{proof}
 Without loss of generality, we assume that $\tr \tilde{\chi}|_{\widetilde{\Sigma}}\le 0$. From  \Cref{comparison principle} we have $\tR\leq R$.  Recalling $L(R):=(\f{1}{2} f^{-1} \trch')|_{S_{\ub, R}}$ and using the condition $f<0$, we get 
\begin{equation*}
    L(\tR)\ge  L(R)=0.
\end{equation*}
We then define the below operator $\mathcal{L}$ and it satisfies
\begin{equation*}
    \mathcal{L}(\tR-R)=L(\tR)-L(R)=\int_0^1 \partial_R L\big(t\tR+(1-t) R \big)[\tR-R] dt \ge 0 ,
\end{equation*}
Following the derivation for $\partial_R L$ in \Cref{partial R L}, we can express  $\mathcal{L}$ (a linear elliptic operator) as
\begin{equation*}
	\mathcal{L}[\tR-R]=A^{i j} D_{ij} (\tR-R) +B^iD_i (\tR-R) +C(\tR-R).
\end{equation*}
With smooth $\widetilde{\Sigma}$, we further have that the coefficients of $\mathcal{L}$ (depending on $\tR, R$) satisfy
\begin{equation*}
		D^{-1} |\xi|^2\le A^{i j}\xi_i \xi_j \le 	D |\xi|^2
\end{equation*}
and 
\begin{equation*}
		\max\limits_{i, j}|A^{ij}|+\max\limits_{i} |B^i |+|C|\le D.
\end{equation*}
with some constant $D\geq 1$ depending on $\tR$ and $R$. Denote $C^+=\max(C, 0)$ and $C^-=\min(C, 0)$. We also get
\begin{equation*}
   A^{i j} D_{ij} (\tR-R) +B^iD_i (\tR-R)+C^- (\tR-R)=( \mathcal{L}-C^+ )[\tR-R] \ge  -C^+ (\tR-R)\ge 0.
\end{equation*}
The desired result thus follows from the strong maximum principle (refer to Theorem 3.5 in \cite{G-T}) and noting that the above elliptic operator on the left has a non-positive zeroth order coefficient.
\end{proof}
The rest of this section is to show that the apparent horizon $\mathcal{AH}=\cup_{\ub \in J} M_{\ub}$ must be piecewise spacelike or null if the null comparison principle holds. And we will demonstrate three typical scenarios about gravitational collapse, in which the null comparison principle applies. Finally, we verify  several physical laws about dynamical horizon and isolated horizon.

\subsection{Local Achronality of Apparent Horizon}
In this subsection,  we prove that for any $p\in M_{\ub}$, the non-zero vector $X_p$ tangent to $\mathcal{AH}$ and normal to $M_{\ub}$ must be spacelike or (outgoing) null, if $M_{\ub}$ satisfies the null comparison principle. We first introduce some basic definitions and notations in Lorentzian geometry. \footnote{For more discussions about Lorentzian geometry, interested readers are  referred to \cite{Galloway}.}
\begin{definition}
    Let  $p, q$ be two points in spacetime $\mathcal{M}$.

    (1) We say that $q$ is in the timelike future (resp. past) of $p$ provided that there exists a future-directed (resp. past-directed) timelike curve in $\mathcal{M}$ from $p$ to $q$.

    (2)  We say that $q$ is in the causal future (resp. past) of $p$ if there exists a future-directed (resp. past-directed)  causal curve in $\mathcal{M}$ from $p$ to $q$. 
\end{definition}
\begin{definition}
  Given a set $S\subset \mathcal{M}$, the timelike future and causal future (resp. past) of $S$, are defined as
\begin{equation*}
		I^{\pm}(S)=\{q\in \mathcal{M}: \ q \ \textrm{is in the timelike future (resp. past) of} \ p\  \textrm{for some}\  p \in S \}
  \end{equation*}
  and
  \begin{equation*}
		J^{\pm}(S)=\{q\in \mathcal{M}: \ q \ \textrm{is in the causal future (resp. past) of} \ p\  \textrm{for some}\  p \in S \}.
\end{equation*}
\end{definition}
\noindent Note that the above definition immediately implies that $I^{\pm}(S)\subseteq 	J^{\pm}(S)$. 

We also define the local achronality of a hypersurface $\Delta\subset \mathcal{M}$ as follows:
\begin{definition}
    Let $\Delta$ be a smooth hypersurface in $\mathcal{M}$ foliated by spacelike sections $\Sigma_s$. For a fixed $s$, we say that

  (1)  $\Delta$ is (outgoing) null when restricted to $\Sigma_s$, if any non-zero tangent vector $X_p$ of $\mathcal{AH}$ at $p\in \Sigma_s$, that is normal to $\Sigma_s$, is (outgoing) null;

    (2)   $\Delta$ is spacelike when restricted to $\Sigma_s$, if any non-zero tangent vector $X_p$ of $\mathcal{AH}$ at $p\in \Sigma_s$, that is normal to $\Sigma_s$, is spacelike.
\end{definition}
\vspace{2mm}

Now we present a lemma that provides a sufficient condition to guarantee the (local) achronality of a hypersurface $\Delta$.
	\begin{lemma}\label{lemma spacelike or null}
	Let $\Delta$ be a hypersurface in $\mathcal{M}$ foliated by spacelike sections $\Sigma_s$. For a fixed $s$, if there exists some $\delta>0$, such that
	\begin{equation}\label{spcacelike or null cond}
	\Sigma_{s'}\cap I^{\pm} (\Sigma_s)=\emptyset \qquad \textrm{for all} \quad s'\in (s, s+\delta),
	\end{equation}
	then $\Delta$ is spacelike or (outgoing) null, when restricted to $\Sigma_s$. Moreover, if we replace $I^{\pm} (\Sigma_s)$ by $J^{\pm} (\Sigma_s)$ in \eqref{spcacelike or null cond}, then $\Delta$ must be spacelike, when restricted to $\Sigma_s$.
\end{lemma}
\begin{proof}
	It suffices to show that any non-zero tangent vector $X\in T\Delta$ normal to $\Sigma_s$ is spacelike or null. For the sake of contradiction, assume that there exist a point $p\in \Sigma_s$ and a future directed timelike vector $X_p\in T_p \Delta$. Let $\alpha: (-t_0, t_0)\to \Delta$ be a curve with $\alpha(0)=p, \alpha'(0)=X_p$. We can then find a sufficiently small $\epsilon>0$ such that $\alpha|_{[-\epsilon, \epsilon]}$ is contained within $\cup_{|s'-s|<\delta} \Sigma_{s'}$ and the following properties hold
\begin{enumerate}
	\item The tangent vector $\alpha'|_{[-\epsilon, \epsilon]}$ remains future-directed timelike;
	\item The curve $\alpha|_{[-\epsilon, \epsilon]}$ only intersects with $\Sigma_s$ at $p$.
\end{enumerate}
	Now we select $q\in\{\alpha(\epsilon), \alpha(-\epsilon) \}$ so that $q\in \Sigma_{s'}$ with some $s<s'<s+\delta$. If $q=\alpha(\epsilon)$, then we have $q\in I^+(p)\subset I^+(\Sigma_s)$, which yields a contradiction since $\Sigma_{s'}\cap I^{+} (\Sigma_s)\neq \emptyset$.
	On the other hand, if $q=\a(-\epsilon)$, then $q\in I^-(p)\subset I^-(\Sigma_{s})$. This is contrary to $\Sigma_{s'}\cap I^{-} (\Sigma_{s})\neq \emptyset$.

Hence we have proved the first conclusion of this lemma. The second conclusion can be derived in a similar manner if we change the timelike curve  $\a$ to be causal in the preceding proof.
\end{proof}

Our next lemma reveals that the pointwise value of the outgoing null expansion along the corresponding null hypersurface is independent of the choice of 2-sphere.
\begin{lemma}\label{null expasion indenpendence along null hypersurface}
	Let $H$ be a null hypersurface and $L$ be its normal null vector field along it. Assume that $\Sigma_1$ and $\Sigma_2$ are two spacelike $2$-surfaces within $H$ and they intersect at $p$. Denote $\trch_i$ as the outgoing null expansion of $\Sigma_i$ with respect to $L$ for $i=1,2$. Then at $p$ we have $\trch_1|_p=\trch_2|_p$.
\end{lemma}
\begin{proof}
	Set $\{e_a^{(i)} \}_{a=1,2}$ to be the orthonormal frame of $T\Sigma_i$ for $i=1, 2$. At $p\in \Sigma_1\cap \Sigma_2$, since $(e_a^{(i)})_{a=1, 2}\in TH$, there exists $(\wideparen{\alpha}_a, \wideparen{\beta}_a, \wideparen{\gamma}_a)_{a=1,2}$ such that
	\begin{equation*}
	e_{a}^{(2)}=\wideparen{\alpha}_{a} e_{1}^{(1)}+\wideparen{\beta}_a e_{2}^{(1)}+\wideparen{\gamma}_{a} L.
	\end{equation*}
	We observe that 
 \begin{equation*}
     \{ e_{a}^{(3)}:=\wideparen{\alpha}_{a} e_{1}^{(1)}+\wideparen{\beta}_a e_{2}^{(1)} \}_{a=1, 2} 
 \end{equation*}
   also serves as an orthonormal basis of $T\Sigma_1$. Consequently, a direct computation gives
	\begin{equation*}
	\begin{split}
	\trch_2|_p=\sum_{A=1}^{2}\bfg(\bfD_{e_a^{(2)}}L, e_a^{(2)}) 
	=&\sum_{A=1}^{2}\bfg(\bfD_{e_a^{(3)}+\wideparen{\gamma}_a L}L, e_a^{(3)}+\wideparen{\gamma}_a L) \\
	=& \sum_{A=1}^{2}\bfg(\bfD_{e_a^{(3)}}L, e_a^{(3)})=\trch_1|_p.
	\end{split}
	\end{equation*}
Here we use $\bfD$ to denote the spacetime covariant derivative of $\mathcal{M}$.
\end{proof}
\vspace{2mm}
We are ready to prove the main conclusion of this section.
\begin{proposition}\label{M_ub' cap I^+(M_ub)=empty}
If the null comparison principle holds for MOTS $M_{\ub}$ along $\Hb_{\ub}$, then there exists some $\delta>0$, so that for all $\ub'\in(\ub, \ub+\delta)$, we have
\begin{equation*}
M_{\ub'}\cap I^{\pm} (M_{\ub})=\emptyset.
\end{equation*}
\end{proposition}
\begin{proof}
With $\delta>0$ we select $\ub'\in (\ub, \ub+\delta)$. To characterise the boundary of the causal past of the MOTS $M_{\ub'}$, we construct an outgoing null hypersurface $H$ originating from $M_{\ub'}$. For any $p\in M_{\ub'}$, we designate  $L_p$ to be the future-directed outgoing null vector that is normal to $M_{\ub'}$. Along $L_p$ there exists an unique geodesic $l_p$ emanating from $p$. We then extend $L_p$ to be vector $L$ along $l_p$ satisfying $D_L L=0$. The collection of all such geodesics $\{l_p\}$ constitutes a null hypersurface $H=\cup_{p\in M_{\ub'}} l_p$ that contains $M_{\ub'}$ and $L$ serves as the null vector normal to $H$.  For future use, we denote the intersection of $H$ and $\Hb_{\ub}$ to be $\Sigma_{\ub'}$ and assume that with coordinates on $\Hb_{\ub}$ we can represent it as $\Sigma_{\ub'}=\{r=\tR(\ub', \theta_1, \theta_2) \} $. Notice that from the construction of $H$ and $\Sigma_{\ub'}$,  we have that $\Sigma_{\ub'}$ is a spacelike 2-surface and it holds
\begin{equation}\label{proof M_ub' cap I^+(M_ub)=empty inclusion}
	\Big( \{r\le \tR(\ub', \theta_1, \theta_2) \}\cap \Hb_{\ub} \Big)\cap I^-(M_{\ub'})=\emptyset.
	\end{equation}
\begin{figure}[h]
    \centering
 \tikzset{every picture/.style={line width=0.75pt}} 

\begin{tikzpicture}[x=0.75pt,y=0.75pt,yscale=-1,xscale=1]

\draw   (230.56,89.18) .. controls (230.56,78.68) and (264.01,70.16) .. (305.28,70.16) .. controls (346.55,70.16) and (380,78.68) .. (380,89.18) .. controls (380,99.68) and (346.55,108.2) .. (305.28,108.2) .. controls (264.01,108.2) and (230.56,99.68) .. (230.56,89.18) -- cycle ;
 
\draw    (220,59.36) -- (229.78,90.9) -- (253.6,167.76) ;

\draw    (391.6,56.16) -- (380.25,91.71) -- (355.6,168.96) ;

\draw  [dash pattern={on 0.84pt off 2.51pt}] (269.2,139.16) .. controls (269.2,134.41) and (285.32,130.56) .. (305.2,130.56) .. controls (325.08,130.56) and (341.2,134.41) .. (341.2,139.16) .. controls (341.2,143.91) and (325.08,147.76) .. (305.2,147.76) .. controls (285.32,147.76) and (269.2,143.91) .. (269.2,139.16) -- cycle ;

\draw    (269.2,139.16) -- (253.6,167.76) ;

\draw    (341.2,139.16) -- (355.6,168.96) ;

\draw    (253.6,167.76) .. controls (282,180.4) and (315.6,190.56) .. (355.6,168.96) ;

\draw  [dash pattern={on 0.84pt off 2.51pt}]  (253.6,167.76) .. controls (300,162.16) and (325.6,165.36) .. (353.6,167.76) ;

\draw    (246.4,44.56) -- (200.44,169.8) ;

\draw    (368.96,45.8) -- (405.6,176.56) ;

\draw (342.16,91) node [anchor=north west][inner sep=0.75pt]  [font=\tiny]  {$M_{\underline{u} '}$};

\draw (295.36,136.6) node [anchor=north west][inner sep=0.75pt]  [font=\tiny]  {$M_{\underline{u}}$};

\draw (298.96,183.8) node [anchor=north west][inner sep=0.75pt]  [font=\tiny]  {$\Sigma _{\underline{u} '}$};

\draw (336.16,154.2) node [anchor=north west][inner sep=0.75pt]  [font=\tiny]  {$\underline{H}_{\underline{u}}$};

\draw (394.56,122.2) node [anchor=north west][inner sep=0.75pt]  [font=\tiny]  {$\underline{H}_{\underline{u} '}$};

\draw (357.63,125.07) node [anchor=north west][inner sep=0.75pt]  [font=\tiny]  {$H$};

\end{tikzpicture}
    \caption{Relevant region for applying the null comparison principle.}
    \label{proof M_ub' cap I^+(M_ub)=empty}
\end{figure}
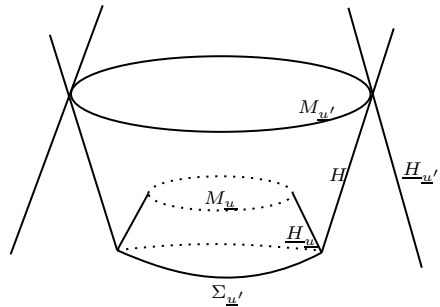

	We proceed to consider the flow map $\mathcal{F}_\tau$ generated by $L$ along $H$ and we require $L\tau=1$ and $\tau=0$ on $M_{\ub'}$. Define $S_\tau=\mathcal{F}_\tau (M_{\ub'})$ and let $\chi(\tau)$ be the second fundamental form of $S_\tau$ with respect to $L$. By the Raychaudhuri equation, we deduce 
	\begin{equation*}
	\f{d}{d \tau} \trch=-\f{1}{2}(\trch)^2-|\chih|_{\slashed{g}}^2\leq 0.
	\end{equation*}
	Here $\slashed{g}$ represents the induced metric on $S_\tau$ and $\chih=\chi-\f{1}{2}\trch \slashed{g}$ is the traceless part of $\chih$.
	
	Since $\trch=0$ along $M_{\ub}$, we have
	\begin{equation}\label{4.4.1}
	\trch|_{S_\tau}\leq  0 \qquad \textrm{if} \  \tau> 0, \quad \text{and} \quad \trch|_{S_\tau}\geq  0 \qquad \textrm{if} \  \tau< 0.
	\end{equation}
	Applying \Cref{null expasion indenpendence along null hypersurface}, the null expansion of $\Sigma_{\ub'}$ with respect to $L$ then obeys
	\begin{equation*}
	\trch|_{\Sigma_{\ub'}}\ge 0.
	\end{equation*}
	Recall that the MOTS $M_{\ub}$ possesses the coordinate $\{r=R(\theta_1, \theta_2) \}$ along $\Hb_{\ub}$. Employing the null comparison principle along $\Hb_{\ub}$, we hence obtain that $\tR(\ub', \theta_1, \theta_2)\ge R(\theta_1, \theta_2)$ for all $(\theta_1, \theta_2)\in \mathbb{S}^2$. In other words, the following inclusion relation for $M_{\ub}$ is obtained
	\begin{equation*}
	M_{\ub} \subset \{r\le \tR(\ub', \theta_1, \theta_2) \}\cap \Hb_{\ub}.
	\end{equation*}
	Together with \eqref{proof M_ub' cap I^+(M_ub)=empty inclusion} this implies $M_{\ub}\cap I^{-} (M_{\ub'})=\emptyset$, or equivalently $M_{\ub'}\cap I^{+} (M_{\ub})=\emptyset$.
 
 The remaining task is to show 
 \begin{equation*}
	M_{\ub'}\cap I^{-} (M_{\ub})=\emptyset.
	\end{equation*}
 It follows from the fact that 
 \begin{equation*}
     M_{\ub'}\subset \Hb_{\ub'} \quad \text{and} \quad  \Hb_{\ub'} \cap I^{-} (M_{\ub})=\emptyset.
 \end{equation*}
 We therefore finish the proof of this proposition.
\end{proof}
Combining \Cref{lemma spacelike or null} and \Cref{M_ub' cap I^+(M_ub)=empty}, together with \Cref{strong comparison theorem} we then prove the desired conclusion:
\begin{theorem}\label{AH piecewise spacelike or null}
	For each MOTS $M_{\ub}$ along the apparent horizon $\mathcal{AH}=\cup_{\ub \in J} M_{\ub}$, assuming that the null comparison principle holds for MOTS $M_{\ub}$ along $\Hb_{\ub}$, then $\mathcal{AH}$ is either spacelike everywhere or (outgoing) null everywhere, when restricted to $M_{\ub}$.
\end{theorem}
\begin{proof}
 Recall that in the proof of \Cref{M_ub' cap I^+(M_ub)=empty}, by virtue of Raychaudhuri equation and \Cref{null expasion indenpendence along null hypersurface}, we derive that
\begin{equation*}
	\trch\big|_{\Sigma_{\ub'}} \ge 0.
\end{equation*}
Applying \Cref{strong comparison theorem}, for any $\ub'\in (\ub, \ub+\delta)$, we obtain that either $\tR(\ub', \theta_1, \theta_2)\equiv R(\theta_1, \theta_2)$ or $\tR(\ub', \theta_1, \theta_2)> R(\theta_1, \theta_2)$ for all $(\theta_1, \theta_2) \in \mathbb{S}^2$. In the former case, it follows that $\Sigma_{\ub'}=M_{\ub}$ and thus $M_{\ub'}$ lies within the outgoing null hypersurface generated by $M_{\ub}$. We proceed to consider the following two scenarios:
\begin{enumerate}
\item  \textbf{Case 1}: \textit{There exist $\delta'>0$, so that for any  $\ub'\in (\ub, \ub+\delta')$, we have pointwise inequality $\tR(\ub', \theta_1, \theta_2)> R(\theta_1, \theta_2)$.} For this scenario, we have a stronger inclusion relation, i.e., 
\begin{equation*}
		M_{\ub} \subset \{r< \tR(\ub', \theta_1, \theta_2) \}\cap \Hb_{\ub}.
\end{equation*}
Together with
\begin{equation*}
		\left( \{r< \tR(\ub', \theta_1, \theta_2) \}\cap \Hb_{\ub} \right) \cap J^-(M_{\ub'})=\emptyset,
\end{equation*}
this yields
\begin{equation*}
	M_{\ub}\cap J^{-} (M_{\ub'})=\emptyset, \quad \text{that is,} \quad M_{\ub'}\cap J^{+} (M_{\ub})=\emptyset.
\end{equation*}
Also noting that
 \begin{equation*}
     M_{\ub'}\subset \Hb_{\ub'} \quad \text{and} \quad  \Hb_{\ub'} \cap J^{-} (M_{\ub})=\emptyset,
 \end{equation*}
 we hence deduce
\begin{equation*}
	M_{\ub'}\cap J^{-} (M_{\ub})=\emptyset.
\end{equation*}
Utilizing \Cref{lemma spacelike or null}, we then conclude that $\mathcal{AH}$ is spacelike when restricted to $M_{\ub}$.
\vspace{2mm}
    \item \textbf{Case 2}: \textit{There exists a sequence  $\{\ub'_n\}\subset (\ub, \ub+\delta)$, such that $\lim\limits_{n\to \infty} \ub'_n=\ub$ and $\tR(\ub'_n, \theta_1, \theta_2)\equiv R(\theta_1, \theta_2)$ for all $n\in \mathbb{N}$.} Consequently, this implies that each MOTS $M_{\ub'_n}$ lies in the outgoing null hypersurface originating from $M_{\ub}$. We still denote it by $H$. Since $M_{\ub'_n}$ converges to $M_{\ub}$ as $n\to \infty$, for any $p\in M_{\ub}$ it must hold
    \begin{equation*}
        T_p \mathcal{AH} \subset T_p H.
    \end{equation*}
    Therefore, we conclude that $\mathcal{AH}$ is tangent to the outgoing null hypersurface $H$ at any $p\in M_{\ub}$, that is, $\mathcal{AH}$ is (outgoing) null everywhere when restricted to $M_{\ub}$.
\end{enumerate}
\end{proof}

\subsection{Applications of the Null Comparison Principle}\label{Subsec: application NCP}
In three typical scenarios about gravitational collapse, we verify that the null comparison principle holds.
\subsubsection{In the Perturbed Schwarzschild Spacetime}\label{Subsubsec: application NCP Sch}
\begin{proposition}\label{comparison theorem}
	For every $\ub\geq 1$, along $\Hb_{\ub}$ the null comparison principle holds for the MOTS $M_{\ub}$ constructed in \Cref{Main thm 1: existence of AH in K-S} based on Klainerman-Szeftel's setting in \cite{K-S}. Hence the apparent horizon $\mathcal{AH}$ constructed in \Cref{Main thm 1: existence of AH in K-S} is either piecewise spacelike or piecewise null.
\end{proposition}
\begin{proof}
	Along $\Hb_{\ub}$,  let $\widetilde{\Sigma}=\{r=\tilde{R}(\theta, \varphi) \}$ be a spacelike $2$-surface with property $\|\tR-R\|_{C^2(\mathbb{S}^2)}\ll 1$. We denote $\tr \tilde{\chi}$ as the outgoing null expansion with respect to $\widetilde{\Sigma}$ and assume first that $\tr \tilde{\chi} \leq  0$. According to the deformation formula \eqref{deformation formula}, we have
	\begin{equation*}
	\tr \tilde{\chi}=\trch+2f\Delta_{S_{\tilde{R}}} \tilde{R}+2(\nab f+f(\eta+\zeta))\cdot \nab \tilde{R}-4f^2\chibh_{bc}\nab^b \tilde{R}\nab^c \tilde{R}+(2fe_3(f)-f^2\trchb-4\omegab f^2)|\nab \tilde{R}|^2.
	\end{equation*}
	Since $f=[e_3(r)]^{-1}<0$, this implies
\begin{equation*}
\begin{split}
0\leq  &\Delta_{S_{\tR}} \tR(\theta, \varphi)+\big(f^{-1}\nab f+(\eta+\zeta)\big)(\tR(\theta, \varphi),\theta, \varphi)\cdot \nab \tR(\theta, \varphi)\\&+(e_3(f)-\f{1}{2}f\trchb-2\omegab f)(\tR(\theta, \varphi), \theta, \varphi)|\nab \tR(\theta, \varphi)|^2\\
&-2(f\chibh_{bc})(\tR(\theta, \varphi), \theta, \varphi)\nab^b \tR(\theta, \varphi) \nab^c \tR(\theta, \varphi)+\f{1}{2}(f^{-1}\trch)(\tR(\theta, \varphi), \theta, \varphi).
\end{split}
\end{equation*}
For MOTS $M_{\ub}=\{r=R(\theta, \varphi) \}$ along $\Hb_{\ub}$, recall that $R(\theta, \varphi)$ satisfies
\begin{equation*}
\begin{split}
0=&\Delta_{S_R} R(\theta, \varphi)+\big(f^{-1}\nab f+(\eta+\zeta)\big)(R(\theta, \varphi),\theta, \varphi)\cdot \nab R(\theta, \varphi)\\&+(e_3(f)-\f{1}{2}f\trchb-2\omegab f)(R(\theta, \varphi), \theta, \varphi)|\nab R(\theta, \varphi)|^2\\&-2(f\chibh_{bc})(R(\theta, \varphi), \theta, \varphi)\nab^b R(\theta, \varphi) \nab^c R(\theta, \varphi)
+\f{1}{2}(f^{-1}\trch)(R(\theta, \varphi), \theta, \varphi).
\end{split}
\end{equation*}
Now we consider the elliptic equation for $\tR-R$:
\begin{equation}\label{Laplace tR-R}
\begin{split}
\Delta_{S_{\tR}}(\tR-R)
=-(\Delta_{S_{\tR}}-\Delta_{S_R})R+\l\Delta_{S_{\tR}} \tR- \Delta_{S_R} R\r.
\end{split}
\end{equation}
By \eqref{3.13}, we can write the Laplace-Beltrami  operator $\D_M$ in local coordinates as 
\begin{equation*}
\begin{split}
\D_M f=&g^{\theta_1 \theta_1}\f{\partial^2 f}{\partial \theta_1 \partial \theta_1}+g^{\theta_2 \theta_2}\f{\partial^2 f}{\partial \theta_2 \partial \theta_2}+2 g^{\theta_1 \theta_2}\f{\partial^2 f}{\partial \theta_1 \partial \theta_2}\\
&+\f{\partial}{\partial \theta_1} (g^{\theta_1 \theta_1})\f{\partial f}{\partial \theta_1}+\f{\partial}{\partial \theta_1} (g^{\theta_1 \theta_2})\f{\partial f}{\partial \theta_2}\\
&+\f{\partial}{\partial \theta_2} (g^{\theta_2 \theta_1})\f{\partial f}{\partial \theta_1}+\f{\partial}{\partial \theta_2} (g^{\theta_2 \theta_2})\f{\partial f}{\partial \theta_2}\\
&+\f12 g^{\theta_k \theta_j}\f{\partial g_{\theta_j \theta_k}}{\partial \theta_i} g^{\theta_i \theta_l}\f{\partial f}{\partial \theta_l},
\end{split}
\end{equation*}
where $i,j,k,l=1,2$ and $g^{\theta_k \theta_j}$ depends on $(\ub, r, \theta, \varphi)$.  Combining a priori estimates 
\begin{equation}\label{a priori}
\|R-2m_{\infty}\|_{C^2(\mathbb{S}^2)} \ll 1 
\end{equation}
and estimates of metric coefficients \eqref{estimates for metric}, we deduce that
\begin{equation*}
-(\Delta_{S_{\tR}}-\Delta_{S_R})R=d_1(\tR-R) \quad \text{with} \quad |d_1|\ll 1.
\end{equation*} 

For $\Delta_{S_{\tR}} \tR- \Delta_{S_R} R$ in \eqref{Laplace tR-R},  utilizing a priori estimates \eqref{a priori} and \eqref{estimates for metric} again, together with hyperbolic estimates for Ricci coefficients in \cite{K-S} and the condition $\|\tR-R\|_{C^2(\mathbb{S}^2)}\ll 1$, we arrive at
\begin{equation*}
\begin{split}
&\Delta_{S_{\tR}} \tR- \Delta_{S_R} R\\ \geq &\f{1}{\tR}|\nab \tR|^2+\f{1}{\tR}-\f{2m_{\infty}}{\tR^2}+d_{1a}(\tR, \theta, \varphi)\nab^a \tR+d_{2bc}(\tR, \theta, \varphi)\nab^b \tR \nab^c \tR+d_3(\tR, \theta, \varphi)\\
&-\Big(\f{1}{R}|\nab R|^2+\f{1}{R}-\f{2m_{\infty}}{R^2}+d_{1a}(R, \theta, \varphi)\nab^a R+d_{2bc}(R, \theta, \varphi)\nab^b R \nab^c \tR+d_3(R, \theta, \varphi)\Big) \\
=&(-\f{1}{\tR R}+\f{2m_{\infty}(\tR+R)}{\tR^2 R^2}+d_1')(\tR-R)+d^i_2\f{\partial}{\partial \theta_i} (\tR-R)\\
=&(\f{1}{4m_\infty^2}+d_1')(\tR-R)+d^i_2\f{\partial}{\partial \theta_i} (\tR-R).
\end{split}
\end{equation*}
Here we have $|d_{1a}|, |d_{2bc}|, |d_3|, |d_1'|, |d_2^i|\ll 1$. 

Back to \eqref{Laplace tR-R}, we therefore deduce
\begin{equation*}
\Delta_{S_{\tR}}(\tR-R)-d^i_2\f{\partial}{\partial \theta_i} (\tR-R)-(\f{1}{4m_\infty^2}+d_1+d_1')(\tR-R)\geq  0.
\end{equation*}
Applying the maximum principle, we then conclude
\begin{equation*}
\tR(\theta, \varphi)\leq  R(\theta, \varphi) \qquad \textrm{for any} \quad (\theta, \varphi)\in \mathbb{S}^2.
\end{equation*}
In the same manner, we can draw the desired conclusion for the case $ \tr \tilde{\chi} \geq  0$ as well, by changing the direction of each inequality in the proceeding proof.
\end{proof}
\begin{remark}\label{Rmk non-negative low bound}
    With the null frame $\{e'_1, e'_2, e_3', e_4'\}$ associated with the MOTS $M_{\ub}$, we can write the tangent vector of $\mathcal{AH}=\{r=R(\ub, \theta, \varphi)  \}$ normal to $M_{\ub}$ as
    \begin{equation*}
        X=e_3'(r-R)e_4'-e_4'(r-R) e_3'.
    \end{equation*}
    Hence, the local achronality of $\mathcal{AH}$ depends on the non-negativity for the sign of
    \begin{equation*}
        \bfg(X, X)=4e_3'(r-R)e_4'(r-R).
    \end{equation*}
    Recall \eqref{new frames}
         \begin{equation*}
e_3'=e_3, \quad e_4'=e_4+2fe^a(R) e_a+f^2|\nab R|^2 e_3.
\end{equation*}
Together with the fact that $e_3 (R)=0$ and $e_3(r)=f^{-1}<0$, we have
\begin{equation*}
    e_3'(r-R)e_4'(r-R)=f^{-1}\big( e_4(r-R)-f|\nab R|^2 \big),
\end{equation*}
which is non-negative if and only if
\begin{equation}\label{achronal condition eqn}
    e_4(R-r)+f|\nab R|^2 \ge 0.
\end{equation}
Via the purely analytic approach, plugging the estimates for  $\nab R$, $\prub R$ from   \eqref{C^1 estimate} and \Cref{partial ub R est} as well as established estimates in \cite{K-S}, noting that
\begin{equation*}
    \partial_{\ub} =\f{1}{2}\vsgmb\Big[e_4-\f{\kb+A}{\kbb}e_3-\sqrt{\gamma}(b-\f{\kb+A}{\kbb}\ud{b})e_\theta\Big],
\end{equation*}
we can only obtain the lower bound
\begin{equation*}\label{weaker lower bound achronal}
    e_4(R-r)+f|\nab R|^2 \ge -\f{\varepsilon_0}{\ub^{1+\delta_{dec}}}.
\end{equation*}
This is not sufficient to imply \eqref{achronal condition eqn}. However, with the assistance of the null comparison principle, we are able to show the desired stronger lower-bound estimate, i.e.,
\begin{equation*}
    e_4(R-r)+f|\nab R|^2 \ge 0.
\end{equation*}
\end{remark}

\subsubsection{In the Anisotropic Gravitational Collapse Spacetime}\label{Subsec application anisotropic spacetime}
In \cite{An-Han} the first author and Han constructed the anisotropic apparent horizon arising in gravitational collapse. Specifically,  they consider a spacetime with double null foliation, whose metric takes the form of
\begin{equation*}
    \bfg=-2\O^2(du\otimes d\ub+d\ub \times du)+\slashed{g}_{AB}(d\theta^A-d^A d\ub)\times (d\theta^B-d^B d\ub).
\end{equation*}
Here $u$ and $\ub$ denote the incoming and outgoing optical functions. They define  the  level sets  of $u$ and $\ub$ by $H_{u}$ and $\Hb_{\ub}$ respectively, and let $S_{u, \ub}$ be the intersection of $H_u$ and $\Hb_{\ub}$. Taking $\{e_1, e_2 \}$ as a tangent frame on the $2$-sphere $S_{u, \ub}$, it can be verified readily that
$   \{ e_1, e_2, \ e_3=\O^{-1}\partial_u, \ e_4=\O^{-1}(\partial_{\ub}+d^A \partial_{\theta^A})  \}$ forms a null frame.

In the short pulse regime, they let parameters $a, \delta$ to be positive constants satisfying $1\ll a\le \delta^{-1}$. For any given $\ub \in (0, \delta]$, they find the MOTS $M_{\ub}=\{u=1-R(\ub, \theta_1, \theta_2)=1-\ub a e^{-\phi(\ub, \theta_1, \theta_2)} \}$ along the incoming null cone $\Hb_{\ub}$ via solving the following quasilinear elliptic equation
\begin{equation}
    \begin{split}
       0=S(\phi, \ub)= \D_{\gamma} \phi+1-\f{1}{2}fe^{\phi}+F(\phi, \ub).
    \end{split}
\end{equation}
In their setting, $\gamma=R^{-2} \slashed{g}$ with $\slashed{g}$ representing the induced metric on $S_{u, \ub}$, and $\gamma$ is uniformly elliptic. They also have
\begin{equation*}
    F(\phi, \ub)=2R^{-1} \O \chibh_{kl} \gamma^{ik} \gamma^{jl} \nab_i \phi \nab_j \phi+R\a_1 |\nab_{\gamma} \phi|^2+R \a_2
\end{equation*}
with
\begin{equation*}
    \begin{split}
        \a_1=-(\f{1}{2}\O \tr_g \chib+\f{1}{R}+4\O \omegab), \quad 
        \a_2=\f{1}{2}\O^{-1} \tr_g \chi-\f{1}{R}+\f{\ub a}{2R^2} f(\omega, \ub).
    \end{split}
\end{equation*}
In their anisotropic setting \cite{An-Han}, the function $f(\theta_1, \theta_2, \ub)$ is smooth on $\mathbb{S}^2\times (0, \delta]$ satisfying $0\le f\le 1$ and, for any $\ub \in(0, \delta]$, it holds
\begin{equation}\label{f condition 2}
    f(\cdot, \ub)\ge m \ \text{on} \ B_{p}(\epsilon)
\end{equation}
for some ball $B_{p}(\epsilon)$ in $\mathbb{S}^2$ centered at $p$ with radius $\epsilon\in (0, \pi/2)$ and some constant $m>0$.
\vspace{3mm}

Picking $r=1-u$, we then verify the null comparison principle for MOTS $M_{\ub}$ along $\Hb_{\ub}$ with $\ub \in (0, \delta]$.
\begin{proposition}
   For the apparent horizon $\mathcal{AH}$ constructed in \cite{An-Han}, for every $\ub \in (0, \delta]$,  along $\Hb_{\ub}$ the null comparison principle holds for the MOTS $M_{\ub}$. Hence $\mathcal{AH}$ is either piecewise spacelike or piecewise null.
\end{proposition}
\begin{proof}
  Let $\widetilde{\Sigma}=\{r=\tilde{R}(\theta_1, \theta_2):=\ub a e^{-\tilde{\phi}(\theta_1, \theta_2)} \}$ be a spacelike $2$-surface along $\Hb_{\ub}$, with $(\theta_1, \theta_2)$ being a local coordinate on $\mathbb{S}^2$. We denote $\tr \tilde{\chi}$ to be the outgoing null expansion with respect to $\widetilde{\Sigma}$ and assume that $\tr \tilde{\chi} \leq 0$ holds. 

    Recall that the MOTS solution $\phi$ obeys the a priori bound
    \begin{equation*}
        \|\phi\|_{C^2(\mathbb{S}^2)} \le K 
    \end{equation*}
for certain positive constant $K$ depending on $\epsilon, m$. By virtue of 
    \begin{equation*}
       \tr \tilde{\chi}=2\O \tR^{-1} S(\tilde{\phi}),
    \end{equation*}
    and noting that $\O>0$, we have
    \begin{equation*}
        S(\tilde{\phi})\le  0=S(\phi).
    \end{equation*}
     With the assistance of the linearized operator $\partial_{\phi} S(\phi)$, we can write
    \begin{equation*}
        \mathcal{L}(\tilde{\phi}-\phi):=S(\tilde{\phi})-S(\phi)=\int_{0}^{1} \partial_\phi S(t\tilde{\phi}+(1-t)\phi)dt [\tilde{\phi}-\phi]\le 0,
    \end{equation*}
    where $\mathcal{L}$ is viewed as a linear elliptic operator acting on $\tilde{\phi}-\phi$. Choosing $\sigma>0$ suitably small so that for any $\|\tR-R\|_{C^2(\mathbb{S}^2)}<\sigma$, it holds $$\|\tilde{\phi}-\phi\|_{C^2(\mathbb{S}^2)}\le K.$$
   Hence, by taking $a$ sufficiently large, it follows from the proof of Lemma 5.2-5.5 in \cite{An-Han} that there exists a positive function $\psi$ on $\mathbb{S^2}$ satisfying
   \begin{equation*}
       \mathcal{L}(\psi)\le -c(m, \epsilon, K)<0.
   \end{equation*}
   Applying the strong maximum principle (see for example  Theorem 2.11 in \cite{H-L: elliptic pde textbook}), we obtain
    either $\tilde{\phi}>   \phi $ or $\tilde{\phi}\equiv \phi$ on $\mathbb{S}^2$. This indicates
    \begin{equation*}
        \tR(\theta_1, \theta_2) \le  R(\theta_1, \theta_2) \quad \text{for all} \ (\theta_1, \theta_2)\in \mathbb{S}^2.
    \end{equation*}
    
For the case $\tr \tilde{\chi} \ge 0$, we have $S(\tilde{\phi})\ge  0=S(\phi)$. Utilizing Lemma 5.5 in \cite{An-Han}  again, we deduce $\tilde{\phi}\le   \phi$ and thus obtain $ \tR(\theta_1, \theta_2) \ge  R(\theta_1, \theta_2) $ for any $(\theta_1, \theta_2)\in \mathbb{S}^2$.
\end{proof}
\subsubsection{In Spacetime with Naked-Singularity Initial Data} 
In \cite{An naked singularity} the first author studied the instability of naked singularities for the Einstein-Scalar field system with no symmetry assumption. Similar to the argument as in \Cref{Subsec application anisotropic spacetime}, we can see that the null comparison principle also holds in this setting and we have  
\begin{proposition}
     For the apparent horizon $\mathcal{AH}$ constructed in \cite{An naked singularity}, it is locally achronal, i.e., piecewise spacelike or piecewise null.
\end{proposition}

\section{Physical laws}\label{Sec: physical law}
Once the apparent horizon $\mathcal{AH}:=\cup_{\ub \in J} M_{\ub}$ is proved to be piecewise spacelike or piecewise outgoing null, in this subsection, with our foliations along $\mathcal{AH}$, we further verify and provide new proofs for several physical laws found by Ashtekar, Krishnan and collaborators. 

\vspace{2mm}
For notational simplicity, we denote the spacelike piece and the null piece of the apparent horizon $\mathcal{AH}$ by $\mathcal{AH}_s$ and $\mathcal{AH}_n$, respectively. In addition, in \Cref{Subsec: the zero law} and \Cref{Subsec: the first law} we use $\{e_1, e_2, e_3, e_4 \}$ to represent the null frame associated with the MOTS $M_{\ub}$, when there is no danger of confusion. In \Cref{Subsec: the second law}, since $\{e_1, e_2, e_3, e_4 \}$ is employed to denote the null frame of the spacetime $\mathcal{M}$, we instead adopt the notation  $\{e'_1, e'_2, e'_3, e'_4 \}$ to stand for the  null frame adapted  to $M_{\ub}$.

\subsection{The zeroth law}\label{Subsec: the zero law}
In black hole thermodynamics the zeroth law states that the surface gravity of the null apparent horizon is constant. Before presenting an explicit form of the zeroth law in our setting, we first introduce some basic setup\footnote{For more detailed relevant notations, interested readers are referred to \cite{A-K1} by Ashtekar and Krishnan.}.

Along $\mathcal{AH}_n$, it is clear that we can find a future-directed null vector $l$ which is normal to $\mathcal{AH}_n$, such that the corresponding outgoing null expansion satisfies $\Theta_{(l)}=0$. Together with the Raychaudhuri equation, this forces
\begin{equation*}
    \mathbf{g}(\bfD_{e_a} l, e_b)=0 \quad \textrm{for any} \quad a,b=1,2,
\end{equation*}
where $\bfD$ represents the spacetime covariant derivative of $\mathcal{M}$ and $\{e_1, e_2 \}$ is a basis of $TM_{\ub}$. Combining with $ \mathbf{g}(D_{X} l, l)=0$ and the fact that for any $ X\in T\mathcal{AH}_n$ it holds
\begin{equation*}
   \mathbf{g}(\bfD_{l} l, X)=-\mathbf{g}( l, \bfD_l X)=-\mathbf{g}( l, [l, X]+\bfD_X l)=-\mathbf{g}( l, \bfD_X l)=0,\footnote{Here we use the fact that $[X, Y]\in T\mathcal{AH}$ for any $X, Y\in T\mathcal{AH}$ since $\mathcal{AH}$ is a submanifold of $\mathcal{M}$.}
\end{equation*}
 we hence obtain
\begin{equation*}
	\mathbf{g}(\bfD_X l, Y)=0 \quad \textrm{for any} \quad X, Y\in T\mathcal{AH}_n.
\end{equation*}
 This implies that there must exist a 1-form $\tilde{\omega}$ on $\mathcal{AH}_n$ such that
\begin{equation}\label{rotational potential}
	\bfD_X l= \tilde{\omega}(X) l \quad \textrm{for any} \ X\in T\mathcal{AH}_n.
\end{equation}
Taking the covariant derivative with respect to $Y\in T\mathcal{AH}_n$, we obtain
\begin{equation}\label{DD l eqn}
    \bfD_Y \bfD_X l=\bfD_Y (\tilde{\omega}(X)) l+\tilde{\omega}(X)\bfD_Y l=\l \bfD_Y (\tilde{\omega}(X))+\tilde{\omega}(X)\tilde{\omega}(Y) \r l,
\end{equation}
which is parallel to $l$. The surface gravity of $\mathcal{AH}_n$ associated with $l$ is defined to be
\begin{equation}\label{surface gravity def}
	\kappa_{(l)}:=\tilde{\omega}(l).
\end{equation}
Following the definition in \cite{A-K1}, we further call $(\mathcal{AH}_n, l)$ the weakly isolated horizon (WIH) if 
\begin{equation}\label{WIH eqn}
    \mathcal{L}_{l} \tilde{\omega}=0,
\end{equation}
where $\mathcal{L}$ is the Lie derivative of spacetime $\mathcal{M}$. Note that we can always find a $l$ such that \eqref{WIH eqn} holds. With this notion, in \cite{A-K1} Ashtekar-Fairhurst-Krishnan demonstrated that the surface gravity $\kappa_{(l)}$ is conserved along the WIH by utilizing the expression of $d\tilde{\omega}$ in terms of the Wely tensor components and by employing the Newman-Penrose formalism. In below, we give an alternative proof via a direct computation using Einstein's field equations.
\begin{proposition}\label{The zeroth law}
	 Along the null piece of apparent horizon $\mathcal{AH}_n$ that is weakly isolated with respect to $l$, the surface gravity $\kappa_{(l)}$ defined in \eqref{surface gravity def} is a constant.
\end{proposition}
\begin{proof}
    It suffices to show that 
    \begin{equation*}
        X \big(\kappa_{(l)}\big) =0  \quad \text{for any} \ X\in T\mathcal{AH}_n.
    \end{equation*}
    Taking the Lie derivative with respect to $l$ on \eqref{rotational potential}, we deduce
    \begin{align*}
        \bfD_l \bfD_X l-\bfD_{\bfD_X l} l=\mL_{l} \bfD_X l=\mL_{l} \l \tilde{\omega}(X) l \r=(\mL_{l}\tilde{\omega})(X)l+\tilde{\omega}(\mL_l X)l=(\mL_{l}\tilde{\omega})(X) l+\bfD_{\mL_l X} l.
    \end{align*}
    This yields
    \begin{equation}\label{Lie derivative of omega}
    \begin{split}
        (\mL_{l}\tilde{\omega})(X)l=\bfD_l \bfD_X l-\bfD_{\bfD_l X} l=&\bfR(l, X)l+\bfD_X \bfD_l l-\bfD_{\bfD_X l} l \\
        =&\bfR(l, X)l+ \bfD_X \big(\kappa_{(l)} l\big)-\tilde{\omega}(X)\kappa_{(l)} l\\
        =&\bfR(l, X)l+X \big(\kappa_{(l)}\big)l.
    \end{split}
    \end{equation}
    Let $n$ be the incoming null vector that is normal to $M_{\ub}$ and satisfies $ \mathbf{g}(l, n)=-2$. From the Einstein vacuum equations we assert
    \begin{equation}\label{R l x l n eqn}
        -\f{1}{2}\bfR(l, X , l, n)-\f{1}{2}\bfR(n, X , l, l)+\bfR(e_1, X, l, e_1)+\bfR(e_2, X, l, e_2)=\text{Ric}(X, l)=0.
    \end{equation}
    In view of \eqref{rotational potential} and \eqref{DD l eqn}, for all $a=1, 2$, we have that the vector
    \begin{equation*}
        \bfR(e_a, X)l=\bfD_{e_a} \bfD_X l-\bfD_{X} \bfD_{e_a} l-\bfD_{\mL_{e_a} X} l 
    \end{equation*}
    is parallel to $l$. This implies 
    \begin{equation*}
        \bfR(e_1, X, l, e_1)=\bfR(e_2, X, l, e_2)=0.
    \end{equation*}
    Consequently, it follows from \eqref{R l x l n eqn} that 
    \begin{equation*}
        \bfR(l, X , l, n)=\mathbf{g}\big( \bfR(l, X) l, n \big)=0.
    \end{equation*}
     Back to \eqref{Lie derivative of omega}, we thus conclude
     \begin{equation*}
         (\mL_{l}\tilde{\omega})(X)=X \big(\kappa_{(l)}\big).
     \end{equation*}
     Since $(\mathcal{AH}_n, l)$ is  weakly isolated, we hence obtain that
      \begin{equation*}
        X \big(\kappa_{(l)}\big) =0  \quad \text{for any} \, \, X\in T\mathcal{AH}_n.
    \end{equation*}
\end{proof}
     Therefore, we conclude that the surface gravity $\kappa_{(l)}$ is a conserved quantity along $\mathcal{AH}_n$.

\subsection{The first law}\label{Subsec: the first law}
We first state the corresponding formula along the null piece $\mathcal{AH}_n$. With 1-form $\tilde{\omega}$ defined in \eqref{rotational potential}, for a vector field $X$ on  $\mathcal{AH}_n$ that is tangent to $M_{\ub}$, as in \cite{A-K1} we define the associated  angular momentum for MOTS $M_{\ub}$ to be
\begin{equation*}
	J_{M}:=-\f{1}{8\pi}\int_{M_{\ub}} \tilde{\omega}(X). 
\end{equation*}
Let $R_M$ be the area radius of MOTS $M_{\ub}$, i.e., $R_M=(A_M/4\pi)^{1/2}$ with $A_M$ being the area of $M_{\ub}$. The so-called canonical horizon energy of $\mathcal{AH}_n$ is further defined to be
\begin{equation}\label{relation EM RM JM}
	E_M=\f{1}{2R_M}\sqrt{R_M^4+4J_M^2}.
\end{equation} 
Along $\mathcal{AH}_n$, the total differential of \eqref{relation EM RM JM} then gives
\begin{equation}\label{the first law null piece}
	d E_M= \f{\kappa_0}{8\pi}d A_M+\O_0 d J_M.
\end{equation}
Here
\begin{equation*}
	\kappa_0(R_M, J_M):=\f{R_M^4-4J_M^2}{2R_M^3\sqrt{R_M^4+4J_M^2}} \quad \text{and} \quad \O_0(R_M, J_M):=\f{\sqrt{R_M^4+4J_M^2}}{2R_M}.
\end{equation*}
The identity \eqref{the first law null piece} is referred as the first law of apparent horizon (isolated horizon) as in \cite{A-K1}.  

\vspace{5mm}
To establish the first law along $\mathcal{AH}_s$, we pick a vector field $X$ on  $\mathcal{AH}_s$, which is tangential to the MOTS $M_{\ub}$.  Let $\hat{\tau}$ be the unit timelike vector normal to $\mathcal{AH}_s$ and $\hat{r}$ be the unit spacelike vector within $\mathcal{AH}_s$ that is outer normal to $M_{\ub}$. We then fix the outgoing null direction to be $l=\hat{\tau}+\hat{r}$. Denote the intrinsic metric of  $\mathcal{AH}_s$ to be $q_{ij}$ and its extrinsic curvature to be $K_{ij}:=\bfD_i \hat{\tau}_j$ with $i,j=1, 2, 3$. Setting $\{e_a \}_{a=1,2}$ to be a basis of $TM_{\ub}$, with null-frame formalism we define
\begin{equation*}
	\chi_{ab}:=\mathbf{g}(\bfD_{e_a} l, e_b) \quad \text{and} \quad \tilde{\eta}_a:=\mathbf{g}(\bfD_{\hat{r}} l, e_a).
\end{equation*}
Due to the spacelikeness of $AH_s$, the following constraint equations for Cauchy data$(q_{ij}, K_{ij})$ are automatically satisfied along $AH_s$:
    \begin{equation}\label{constraint eqn}
        \begin{split}
            H_S:=&\mathcal{R}+(\tr K)^2-K^{ij}K_{ij}=0, \\
            H^i_V:=&D_j(K^{ij}-\tr K q^{ij})=0.
        \end{split}
    \end{equation}
    Here  $\mathcal{R}$ denotes the scalar curvature of $AH_s$ and $D$ is the induced covariant derivative on $AH_s$. 
     
     \vspace{2mm}
    In below we prove a relation between the scalar curvature of  MOTS and $\chi, \tilde{\eta}$.
\begin{lemma}\label{lemma Hs+ri Hi V eqn}
    Along $\mathcal{AH}_s$, it holds
    \begin{equation}\label{Hs+ri Hi V eqn}
        0=\widetilde{\mathcal{R}}+2\nab_a \tilde{\eta}^a-\chi^{ab} \chi_{ab} -2\tilde{\eta}^a \tilde{\eta}_a.
    \end{equation}
    Here $\nab$ is the induced covariant derivative on $M_{\ub}$, $\widetilde{\mathcal{R}}$ is the scalar curvature of  $M_{\ub}$.
\end{lemma}
\begin{remark}
    For the proof of this lemma, we provide two approaches. The first approach is suggested by Ashtekar-Krishnan in  \cite{A-K2}. The second method is new and is based on null-frame formalism.
\end{remark} 
\begin{proof}[First Proof of \Cref{lemma Hs+ri Hi V eqn}]
 From \eqref{constraint eqn}, it suffices to show that
 \begin{equation*}
        H_S+2 r_i H^i_V=\widetilde{\mathcal{R}}+2\nab_a \tilde{\eta}^a-\chi^{ab} \chi_{ab} -2\tilde{\eta}^a \tilde{\eta}_a.
    \end{equation*}
 Note that the Gauss-Codazzi equation along $\mathcal{AH}_s$ can be expressed as
    \begin{equation}\label{Gauss codazzi eqn}
        \mathcal{R}=\widetilde{\mathcal{R}}+(\tr \widetilde{K})^2-\widetilde{K}^{ab}\widetilde{K}_{ab}+2D_i \tilde{\a}^i.
    \end{equation}
    Here $\widetilde{K}_{ab}:=\mathbf{g}(D_{e_a} \hat{r}, e_b)$ is the second fundamental form of the MOTS $M_{\ub}$ along $\mathcal{AH}_s$ and the vector $\tilde{\a}$ is of the form $\tilde{\a}:=D_{\rh} \rh-(D_i \rh^i) \rh=D_{\rh} \rh-(\tr \tK) \rh $. Noting that the outgoing null expansion vanishes for MOTS $M_{\ub}$, we hence have
    \begin{equation}\label{Null expansion zero}
        0=\trch=\mathbf{g}(\bfD_{e_a} l, e^a)=\tr \widetilde{K}+\tr K-K_{\rh \rh}.
    \end{equation}
    Using \eqref{constraint eqn} and \eqref{Gauss codazzi eqn}, a direct computation yields
    \begin{equation}\label{Hs+ r HV eqn}
    \begin{split}
        H_S+2 r_i H^i_V=&\widetilde{\mathcal{R}}+(\tr K)^2+(\tr \tK)^2-K^{ij}K_{ij}-\tK^{ab} \tK_{ab}+2D_j \tilde{\a}^j+2\rh_i  D_j(K^{ij}-\tr K q^{ij}) \\
        =&\widetilde{\mathcal{R}}+(\tr K)^2+(\tr \tK)^2-K^{ij}K_{ij}-\tK^{ab} \tK_{ab}+2D_j \b^j-2(K^{ij}-\tr K q^{ij})D_j r_i
        \end{split}
    \end{equation}
with
\begin{equation*}
    \tilde{\b}^j:=\tilde{\a}^j+\rh_i K^{ij}-\tr K \rh^j=D_{\rh} \rh^j+K_{\rh}^j-(\tr \tK+\tr K)\rh^j.
\end{equation*}
    We then insert \eqref{Null expansion zero} in the expression of $\tilde{\b}$ and deduce
    \begin{equation*}
    \tilde{\b}_{\rh}=K_{\rh \rh}-(\tr \tK+\tr K)=0.
    \end{equation*}
    This implies $\tilde{\b}\in TM_{\ub}$. In view of the definition of $\tilde{\eta}_a$, we also have
    \begin{equation*}
        \tilde{\b}_a=D_{\rh} \rh_a+K_{\rh a}=D_{\rh} \rh_a+D_{\rh} \tauh_a=\tilde{\eta}_a.
    \end{equation*}
    Combining these equations involving $\tilde{\b}$, we hence derive that
    \begin{align*}
        D_j \tilde{\b}^j-(K^{ij}-\tr K q^{ij})D_j r_i=&D_a \tilde{\eta}^a+D_{\rh}\tilde{\eta}_a  \rh^a-K^{ab} D_b \rh_a-K^{a\rh} D_{\rh} \rh_a       +\tr KD^i \rh_i \\
        =&\nab_a \tilde{\eta}^a-\tilde{\eta}_a D_{\rh} \rh^a-K^{ab} \nab_b \rh_a-K^{a\rh} \nab_{\rh} \rh_a       +\tr K \cdot \tr \tK \\
        =&\nab_a \tilde{\eta}^a-(\tilde{\eta}_a+K_{a\rh})(\tilde{\eta}^a-K^{a\rh})-K^{ab} \tK_{ab}    +\tr K \cdot \tr \tK.
    \end{align*}
Back to \eqref{Hs+ r HV eqn}, together with \eqref{Null expansion zero}, we then arrive at
    \begin{align*}
         H_S+2 r_i H^i_V=&\widetilde{\mathcal{R}}+(\tr K+\tr \tK)^2-K^{ij}K_{ij}-\tK^{ab} \tK_{ab}+2\nab_a \tilde{\eta}^a-2\tilde{\eta}^a \tilde{\eta}_a+2K^{a\rh}K_{a\rh}-2K^{ab} \tK_{ab} \\
         =&\widetilde{\mathcal{R}}+2\nab_a \tilde{\eta}^a-\chi^{ab} \chi_{ab} -2\tilde{\eta}^a \tilde{\eta}_a.
    \end{align*}
    Here we use the relation $\chi_{ab}=K_{ab}+\tK_{ab}$.
\end{proof}

Then we move to provide the second proof of \Cref{lemma Hs+ri Hi V eqn} by using the null-frame formalism. We will adopt the null frame $\{e_1, e_2, e_3, e_4 \}$ with $e_4=l=\tauh+\rh$ and $e_3=\tauh-\rh$.
 
\begin{proof}[Second Proof of \Cref{lemma Hs+ri Hi V eqn}]
    With definitions given in \eqref{def curvatures} and \eqref{def Ricci coefficients}, we have the following null structure equations for $\trch$ and the Gauss equation on $M_{\ub}$:\footnote{See the listed null structure equations, null Bianchi equations, constraint equations in  Proposition 7.4.1 of \cite{Chr-Kl} by Christodoulou-Klainerman.}
    \begin{equation*}
        \begin{split}
           e_4 (\trch)+\frac 12 (\trch)^2&=2\div \xi -2\omega \trch+2\xi \cdot (\eta+\etab+2\zeta)-|\chih|^2, \\
e_3 (\trch)+\frac1 2 \trchb \trch &=2\omegab \trch+2\rho- \chih\cdot\chibh+2\div \eta+2(|\eta|^2+\xi\cdot \xib), \\
\bfK&=-\f14\trch \trchb+\f12 \chibh\cdot \chih-\rho.
        \end{split}
    \end{equation*}
    Here $\text{div}$ is the divergence operator on $M_{\ub}$ and $\bfK$ represents the Gauss curvature of $M_{\ub}$. It then follows from $\rh=\f12(e_4-e_3)$ that
    \begin{align*}
        2D_{\rh}( \trch)+\f12(\trch-\trchb) \trch=&-2(\omega+\omegab) \trch+2(\bfK+\f14\trch\trchb)\\&+ 2\div (\xi-\eta)-|\chih|^2-2|\eta|^2+2\xi \cdot (\eta+\etab+2\zeta-\xib).
    \end{align*}
Noting $\trch\equiv 0$ along $\mathcal{AH}_s$ and $\rh\in\mathcal{AH}_s$, together with $\widetilde{\mathcal{R}}\equiv2\bfK$ on $M_{\ub}$, we thus deduce
\begin{equation}\label{tilde R eqn 1}
    0=\widetilde{\mathcal{R}}+2\div (\xi-\eta)-|\chih|^2-2|\eta|^2+2\xi \cdot (\eta+\etab+2\zeta-\xib).
\end{equation}
From the definitions of $\tilde{\eta}_a$ and $K_{a\rh}$, we also have
\begin{align*}
    \tilde{\eta}_a=&\bfg(\bfD_{\rh} l, e_a)=\f12\bfg(\bfD_{e_4-e_3} e_4, e_a)=\xi_a-\eta_a, \\
    K_{a\rh}=&\bfg(\bfD_{a} \tauh, \rh  )=\f{1}{4}\bfg (\bfD_{a} (e_4+e_3), e_4-e_3 )=-\zeta_a.
\end{align*}
Hence, a direct calculation implies
\begin{equation*}
   \etab_a-\xib_a=\f12 \bfg(\bfD_{e_4-e_3} e_3, e_a)=\bfg(\bfD_{\rh} (2\tauh-l), e_a)=2K_{\rh a}-\tilde{\eta}_a=-2\zeta_a-\xi_a+\eta_a.
\end{equation*}
Plugging the equation above into \eqref{tilde R eqn 1}, we therefore arrive at
\begin{align*}
    0=\widetilde{\mathcal{R}}+2\div (\xi-\eta)-|\chih|^2-2|\eta|^2+2\xi \cdot (2\eta-\xi)=\widetilde{\mathcal{R}}+2\div  \tilde{\eta}-|\chih|^2-2| \tilde{\eta}|^2.
\end{align*}
This is equivalent to \eqref{Hs+ri Hi V eqn}, since $\chih_{ab}=\chi_{ab}-\f12 \trch \gamma_{ab}=\chi_{ab}$ along $\mathcal{AH}_s$.
\end{proof}
\vspace{2mm}
Given any function $r=r(R_M)$, we further set the lapse function to be $N_r=|\partial r|$, with $\partial$ denoting the gradient operator on $\mathcal{AH}_s$. With above preparations, we can state the following integral version of the first law for black hole mechanics. This version was first formulated by Ashtekar-Krishnan in \cite{A-K2}.
\begin{proposition}
 Given any vector field $X$ on  $\mathcal{AH}_s$ tangent to MOTSs $M_{\ub}$ and arbitrary function $\O=\O(R_M)$, for any $\ub_2>\ub_1\ge 1$, along $\mathcal{AH}_s$ there holds
	\begin{equation}\label{the first law int}
		\begin{split}
	&4\pi(r_2-r_1)+\int_{M_{\ub_2}} \O j^X-\int_{M_{\ub_1}} \O j^X-\int_{\O_1}^{\O_2}d\O \int_{M_{\ub}}j^X\\=&\f{1}{2}\int_{\mathcal{AH}_s\cap \{\ub_1\le \, \ub \, \le \ub_2 \}} N_r (|\chi|^2+2|\tilde{\eta}|^2)-\f{1}{2}\int_{\mathcal{AH}_s\cap \{\ub_1\le \, \ub \, \le \ub_2 \}} \O P^{ij} \mathcal{L}_X q_{ij}.
\end{split}
	\end{equation}
	Here 
 \begin{equation*}
     j^X:=-K_{ij}X^i \hat{r}^j, \quad P_{ij}:=K_{ij}-\tr K q_{ij}
 \end{equation*}
and
 \begin{equation*}
     r_l=r(R_M(\ub_l)), \quad \O_l=\O(R_M(\ub_l)) \quad  \text{with} \quad  l=1,2.
 \end{equation*}
\end{proposition}
\begin{proof}
    Let $\Delta :=\mathcal{AH}_s\cap \{\ub_1\le \ub \le \ub_2 \}$ be a portion of $\mathcal{AH}_s$ bounded between the MOTS $M_{\ub_1}$ and $M_{\ub_2}$.  By applying \Cref{lemma Hs+ri Hi V eqn}, we get
    \begin{equation*}
        0=\f12 \int_\Delta N_r (\widetilde{\mathcal{R}}-\chi^{ab} \chi_{ab} -2\tilde{\eta}^a\tilde{\eta}_a),
    \end{equation*}
    where we use the fact that the integral of  $\nab^a \tilde{\eta}_a$ is $0$ according to the divergence theorem on $M_{\ub}$. By further employing the Gauss-Bonnet theorem, for the first term on the right side of \eqref{the first law int}, we obtain
    \begin{equation}\label{area balance law}
       \f12\int_\Delta N_r(|\chi|^2+2|\tilde{\eta}|^2)=\f12\int_\Delta N_r \widetilde{\mathcal{R}}= \f12 \int_{r_1}^{r_2} dr \int_{M_{\ub}} \widetilde{\mathcal{R}} =4\pi (r_2-r_1).
    \end{equation}
    \vspace{2mm}
    
    To derive the remaining part on the right of \eqref{the first law int}, that concerns the angular momentum, we proceed to contract $H^i_V$ with $X$ and evaluate the integral
    \begin{equation*}
        \int_{\Delta} \O X_i H^i_V=\int_{\Delta} \O X_i \nab_j(K^{ij}-\tr K q^{ij}).
    \end{equation*}
    It follows from the integration by parts that
    \begin{align}\label{angular momentum integral}
         0=-\int_{M_{\ub_2}} \O j^X+\int_{M_{\ub_1}} \O j^X-\int_{\Delta} D_j(\O X)_i P^{ij}.
    \end{align}
    At the same time, with $\O=\O(R_M)$, we have $D_j \O=|\partial \O|\rh_j$. This implies
    \begin{equation*}
        D_j(\O X)_i=( D_j \O) X_i+\O D_j X_i=|\partial \O| X_i \rh_j+\O D_j X_i.
    \end{equation*}
    Substituting the above expression into \eqref{angular momentum integral} yields
    \begin{equation}\label{angular momentum balance law}
        \int_{M_{\ub_2}} \O j^X-\int_{M_{\ub_1}} \O j^X-\int_{\O_1}^{\O_2}d\O \int_{M_{\ub}}j^X=-\f12\int_{\Delta} \O P^{ij} \mL_X q_{ij},
    \end{equation}
    where we use the fact $X^i\rh_i=0$. The sum of \eqref{area balance law} and \eqref{angular momentum balance law} then gives the desired conservation law \eqref{the first law int}.
\end{proof}
\vspace{2mm}

We then move forward to prove the differential form of the first law. This was formulated in \cite{A-K2} by Ashtekar-Krishnan. The flux across $\Delta =\mathcal{AH}_s\cap \{\ub_1\le \ub \le \ub_2 \}$ is defined to be
$$\mathcal{F}_{\Delta}:=\f{1}{16\pi} \int_{\Delta} \l N_r (|\chi|^2+2|\tilde{\eta}|^2)-\O P^{ab} \mathcal{L}_X q_{ab} \r.$$
Assume there is a well-defined horizon energy $\tilde{E}_M$ such that, at each MOTS $M_{\ub}$, with $\mathcal{F}_{\Delta}$ defined above, it holds
\begin{equation*}
    \tilde{E}_M(\ub_2)-\tilde{E}_M(\ub_1)=\mathcal{F}_{\Delta}.
\end{equation*}
We also define the generalized angular momentum with respect to $X$ along $M_{\ub}$ to be
\begin{equation*}
	\tilde{J}^X_M:=\f{1}{8\pi}\int_{M_{\ub}}j^X.
\end{equation*}
With these notations, we can rewrite \eqref{the first law int} as
\begin{equation}\label{int first law new}
   \f12(r_2-r_2)+\O_2\tilde{J}^X_M(\ub_2)-\O_1\tilde{J}^X_M(\ub_1)-\int_{\O_1}^{\O_2} \tilde{J}^X_M d\O= \tilde{E}_M(\ub_2)-\tilde{E}_M(\ub_1).
\end{equation}
Taking a differentiation on both sides of  \eqref{int first law new}, by virtue of the fact
\begin{equation*}
    dr=\f{d r}{d R_M} (\f{d A_M}{d R_M})^{-1} dA_M=\f{1}{8\pi R_M}\f{dr}{d R_M} dA_M
\end{equation*}
and
\begin{equation*}
    d(\O \tilde{J}^X_M)-\tilde{J}^X_M d\O=\O d\tilde{J}^X_M,
\end{equation*}
we then derive the differential form of the first law along $\mathcal{AH}_s$:
\begin{equation*}
	d \tilde{E}_M:=\f{\overline{\kappa}}{8\pi}d A_M+\O d \tilde{J}^X_M.
\end{equation*}
Here $\overline{\kappa}=\f{1}{2R_M}\f{dr}{d R_M}$ is called the effective surface gravity in \cite{A-K1} and $\O $ corresponds to the angular velocity of MOTS $M_{\ub}$. 
\vspace{2mm}

With a specific choice of functions $r=r(R_M)$ and $\O=\O(R_M)$, the horizon energy $\tilde{E}_M$ can be expressed explicitly as suggested in \cite{A-K2, A-K1}. Noting that $\tilde{J}^X_M$ can be regarded as a function of $R_M$ along $\mathcal{AH}_s$, we set 
    \begin{equation*}
       \overline{\kappa}=\kappa_0(R_M, \tilde{J}^X_M), \quad \O= \O_0(R_M, \tilde{J}^X_M),
    \end{equation*}
    where
    \begin{equation*}
	\kappa_0(R, J):=\f{R^4-4J^2}{2R^3\sqrt{R^4+4J^2}} \quad \text{and} \quad \O_0(R, J):=\f{\sqrt{R^4+4J^2}}{2R}.
\end{equation*}
Thus $r=r(R_M)$ can be solved through the equation $\f{1}{2R_M}\f{dr}{d R_M}=\overline{\kappa}(R_M)$. Recall that the canonical horizon mass is defined to be 
\begin{equation*}
    M(R, J):=\f{1}{2R}\sqrt{R^4+4J^2}.
\end{equation*}
Then a straightforward check yields
\begin{equation*}
    d  M:=\f{\kappa_0}{8\pi}d A+\O_0 d J \qquad  \quad \text{with} \ A:=4\pi R^2.
\end{equation*}
According to the choice of $r=r(R_M)$ and $\O=\O(R_M)$, we then conclude
\begin{equation*}
   \tilde{E}_M=M(R_M, \tilde{J}^X_M).
\end{equation*}
\indent We further demonstrate the consistency between the first law along $\mathcal{AH}_s$ and the first law along the null piece $\mathcal{AH}_n$ as presented in \eqref{the first law null piece}. Notice that the angular momentum $\tilde{J}^X_M$ satisfies
\begin{equation*}
    \tilde{J}^X_M=\f{1}{8\pi} \int_{M_{\ub}}j^X=-\f{1}{8\pi}\int_{M_{\ub}}  K_{ij} X^i \rh^j=\f{1}{8\pi} \int_{M_{\ub}}  \zeta^a X_a.
\end{equation*}
For the angular momentum of MOTS $M_{\ub}$ along $\mathcal{AH}_n$, it reads
\begin{equation*}
    J_{M}:=-\f{1}{8\pi}\int_{M_{\ub}} \tilde{\omega}(X).
\end{equation*}
Since $X\in TM_{\ub}$, we deduce
\begin{equation*}
    \tilde{\omega}(X)=-\f12\bfg(\bfD_{X} l, n)=-\zeta^a X_a.
\end{equation*}
This leads to 
\begin{equation*}
    J_{M}=\f{1}{8\pi}\int_{M_{\ub}}  \zeta^a X_a=\tilde{J}^X_M.
\end{equation*}
Hence the definition of the angular momentum of MOTS $M_{\ub}$ along $\mathcal{AH}_s$ matches with the corresponding one defined on $\mathcal{AH}_n$. Merging these together,  we then get the unified form
    \begin{equation*}
	d {E}_M=\f{\kappa}{8\pi}d A_M+\O d {J}_M.
\end{equation*}

\subsection{The second law}\label{Subsec: the second law}
In black hole thermodynamics, the second law asserts that the entropy is a non-decreasing function of towards the future.  Recall that $A_M$ is defined as the area of MOTS $M_{\ub}=\{r=R(\ub, \theta_1, \theta_2)  \}$ along $\Hb_{\ub}$ and it depends only on $\ub$. Defining  $A_M(\ub)$ as the entropy as in \cite{An: AH, An-Han, A-K1}, in below we prove that $A_M(\ub)$ is non-decreasing with respect to $\ub$.
\vspace{2mm}

Notice that at any point on the MOTS $M_{\ub}$, we have a null frame $\{e'_1, e'_2, e_3', e_4'\}$ with $TM_{\ub}=\text{span} \{ e'_1, e'_2 \}$ and  the outgoing (with respect to $e_4'$) and incoming (with respect to $e_3'$) null expansions of $M_{\ub}$ obey $\trch'=0$ and $\trchb'<0$, respectively.
Choose a non-zero vector $X$ tangent to $\mathcal{AH}$ and normal to the MOTS $M_{\ub}$. When it points outwards, since $X\in \text{span} \{ e'_3, e'_4 \}$, we can express it as 
\begin{equation}\label{X e3 e4 expression}
	X=\a e_4'+\b e_3' \qquad \text{with} \ \a\ge 0.
\end{equation}
Note that  $\a\neq 0$ and this is because the incoming null direction $e_3'$ is transverse to $\mathcal{AH}$. Recalling  \eqref{new frames}, we now have
 \begin{equation*}
e_3'=e_3  \quad  \text{and} \quad e_4'=e_4+2fe^a(R) e_a+f^2|\nab R|^2 e_3.
\end{equation*}
 A direct check gives 
\begin{equation*}
    e_3'(\ub)=e_3(\ub)=0, \quad e_4'(\ub)=e_4(\ub)>0.
\end{equation*}
With $X$ acting on $A_M$, we get
\begin{equation}\label{eqn X(AM) 1}
	X(A_M)=c\f{d A_M}{d \ub} \qquad \text{with} \quad c=X(\ub)=\a e_4'(\ub)>0.
\end{equation}
On the other hand, a straightforward computation implies 
\begin{equation}\label{eqn X(AM) 2}
	\begin{split}
			X(A_M)=\int_{M_{\ub}} \div_{M_{\ub}} X=\int_{M_{\ub}} \a \trch'+\b\trchb' 	=\int_{M_{\ub}} \b\trchb'.
	\end{split}
\end{equation}
Within the null piece $\mathcal{AH}_n$,  note that $X$ must be null, and according to \eqref{X e3 e4 expression} we thus have $\b\equiv0$. Applying \eqref{eqn X(AM) 2}, this yields $X(A_M)=0$ and via \eqref{eqn X(AM) 1} this further implies $\f{dA_M}{d\ub}=0$ in the null piece $\mathcal{AH}_n$.  As for the spacelike piece $\mathcal{AH}_s$, observing that $\b<0$ (since $X$ is spacelike) and $\trchb'<0$, combining \eqref{eqn X(AM) 1} and \eqref{eqn X(AM) 2} we hence deduce that $\f{dA_M}{d\ub}>0$.  Summarizing the above analysis, we then confirm the second law of black hole thermodynamics along $\mathcal{AH}$:
\begin{proposition}\label{The second law}
	Defining $A_M(\ub)$ (the area of MOTS $M_{\ub}$) as the entropy, along $\mathcal{AH}$ it is a non-decreasing function with respect to $\ub$. More specifically, we have 
 $$\f{dA_M}{d\ub}=0 \, \, \text{in the null piece} \,\,  \mathcal{AH}_n \quad \text{and} \quad \f{dA_M}{d\ub}>0 \,\, \text{in the spacelike piece} \,\, \mathcal{AH}_s.$$
\end{proposition}

\end{document}